\documentclass[prx,twocolumn,floatfix,superscriptaddress,longbibliography,notitlepage]{revtex4-1}
\usepackage{graphicx} %
\usepackage{wrapfig}
\usepackage{amsmath}
\usepackage{amsfonts}
\usepackage{amssymb}
\usepackage{amsthm}
\usepackage{appendix}
\usepackage{dsfont}
\usepackage[utf8]{inputenc} 
\usepackage[T1]{fontenc}
\usepackage{mathrsfs}
\usepackage{xcolor}

\newtheorem{theorem}{Theorem}[section]
\newtheorem{corollary}{Corollary}[theorem]
\newtheorem{lemma}[theorem]{Lemma}
\newtheorem{definition}[theorem]{Definition}

\newtheorem{conjecture}[theorem]{Conjecture}

\renewcommand{\>}{\rangle}
\newcommand{\<}{\langle}
\newcommand{\Tr}{\textup{Tr}}

\newcommand{\poly}{\textup{poly}}

\newcommand{\1}{\mathds{1}}
\newcommand{\g}{\mathfrak{g}}
\newcommand{\U}{\mathbf{U}}
\newcommand{\Ad}{\text{Ad}}

\newcommand{\E}{\mathbf{E}}
\newcommand{\K}{\mathbf{K}}
\renewcommand{\O}{\mathbf{O}}
\renewcommand{\H}{\mathbf{H}}
\newcommand{\A}{\mathbf{A}}
\newcommand{\X}{\mathbf{X}}
\newcommand{\W}{\mathbf{W}}
\newcommand{\G}{\mathcal{G}}

\begin{document}

\title{The Adjoint Is All You Need:\\ Characterizing Barren Plateaus in Quantum Ans\"atze}

\author{Enrico Fontana}
\affiliation{Global Technology Applied Research, JPMorgan Chase}
\affiliation{Computer and Information Sciences, University of Strathclyde}
\author{Dylan Herman}
\email{dylan.a.herman@jpmorgan.com}
\author{Shouvanik Chakrabarti}
\author{Niraj Kumar}
\author{Romina Yalovetzky}
\affiliation{Global Technology Applied Research, JPMorgan Chase}
\author{Jamie Heredge}
\affiliation{Global Technology Applied Research, JPMorgan Chase}
\affiliation{School of Physics, The University of Melbourne}
\author{Shree Hari Sureshbabu}
\author{Marco Pistoia}
\affiliation{Global Technology Applied Research, JPMorgan Chase}

\begin{abstract}
    Using tools from the representation theory of compact Lie groups, we formulate a theory of Barren Plateaus (BPs) for parameterized quantum circuits whose observables lie in their dynamical Lie algebra (DLA), a setting that we term Lie algebra Supported Ansatz (LASA). A large variety of commonly used ans\"atze such as the Hamiltonian Variational Ansatz, Quantum Alternating Operator Ansatz, and many equivariant quantum neural networks are LASAs. {In particular, our theory provides, for the first time, the ability to compute the variance of the gradient of the cost function of %
    the quantum compound ansatz.} We rigorously prove that{, for LASA,} the variance of the gradient of the cost function, {for a $2$-design of the dynamical Lie group}, scales inversely with the dimension of the DLA, which agrees with existing numerical observations. {In addition, to motivate the applicability of our results for $2$-designs to practical settings, we show that rapid mixing occurs for LASAs with polynomial DLA.} Lastly, we include potential extensions for handling cases when the observable lies outside of the DLA and the implications of our results.
    
\end{abstract}

\maketitle

\section{Introduction}

Variational quantum algorithms (VQAs) are a popular class of quantum computing heuristics due to their low circuit cost and ability to be trained in a hybrid quantum-classical fashion \cite {cerezo2021variational}. {The community has identified a variety of potential applications for VQAs in the areas of  optimization~\cite{farhi2014quantum, peruzzo2014variational, Liu_2022, Niroula_2022, Herman_2023, shaydulin2023evidence} and machine learning~\cite{Mitarai_2018, farhi2018classification, havlivcek2019supervised, Larocca_2022,  herman2023}.} Unfortunately, the optimization of VQAs can be a computationally challenging task due to (1) exponentially many parameters being required to ensure convergence~\cite{you2021exponentially, you2022convergence, anschuetz2021critical, anschuetz2022beyond, you2023analyzing}, and (2) exponentially many samples being required to estimate gradients, known as the barren plateau (BP) problem~\cite{mcclean2018barren, cerezo2021cost, wang2021noise, martin2023barren}. {In some cases, it has been observed numerically that both of these obstacles to VQA optimization can be mitigated when the chosen parameterized quantum circuit (PQC) obeys certain symmetries~\cite{Larocca2022diagnosingbarren, you2022convergence}.} The symmetries of the ansatz cause its action{, in either the Schr\"{o}dinger or Heisenberg pictures, to break into invariant subspaces. However, there have only been a few cases in which potentially useful symmetries, mostly in the Schr\"{o}dinger picture, have been identified, e.g. permutation invariance~\cite{schatzki2022theoretical}.}

The existing theoretical results on the trainability {and convergence} of ans\"atze with symmetries have been restricted to {the Schr\"{o}dinger picture and a setting called \emph{subspace controllable} }\cite{mcclean2018barren, Larocca2022diagnosingbarren, schatzki2022theoretical,you2022convergence}. %
Subspace controllability occurs when the circuit can express any unitary transformation between states in an invariant subspace and it has been observed that it results in training landscapes that are essentially trap-free~\cite{russell2016quantum,Larocca_2023}. In addition, if the invariant subspaces have small dimension, i.e. scale polynomially in system size, it can be easily shown that BPs are not present {for subspace controllable PQCs}.

{These results however fail in the uncontrollable setting, where the circuit is limited to expressing a subgroup of the unitary group in the invariant subspace.}
With respect to the BPs problem, existing work has observed a desirable feature of subspace uncontrollable circuits~\cite{Larocca2022diagnosingbarren}. In this setting, it appears that the trainability of the ansatz depends on the dimension of the \emph{dynamical Lie algebra} (DLA), which holds almost trivially in the subspace controllable setting since the DLA dimension grows with the square of the subspace dimension. However, existing work has only provided evidence of this connection to the DLA dimension numerically in the uncontrollable setting~\cite{Larocca2022diagnosingbarren}. There are cases where, for an uncontrollable PQCs, the dimension of the effective DLA only grows polynomially in the system size, while the invariant subspace dimension {where the initial state lies} is exponentially growing, such as the quantum compound ansatz~\cite{kerenidis2022quantum, cherrat2023quantum}. Note that the effective DLA is the restriction of the action of the DLA to an invariant subspace.  Thus, this connection between the DLA dimension and BPs has remained unproven in the general setting.

In this work, using a simple but powerful observation regarding the adjoint representation and the representation theory of compact Lie groups, we prove that for a general class of PQCs that the variance of the gradient of  the cost function does fall inversely with the dimension of the effective DLA {for $2$-designs of the dynamical Lie group}.
{As we will show, the Heisenberg picture and the symmetries of the circuit's action on the observable are more suitable for explaining this phenomenon. This will lead to intuitive and commonplace conditions on the observable that are sufficient for this connection to hold. To show the validity of the $2$-design assumption in practice, we show that fast mixing occurs for DLAs with polynomial dimension, and we experimentally verify our formulae for the quantum compound ansatz.}

\section{Results}\label{sec:main_results}
\subsection{General Framework}

Variational Quantum Algorithms (VQAs) consist of optimizing the parameters of parameterized circuits of the form given in the below definition:
\begin{definition}[Periodic ansatz]\label{def:ansatz}
A periodic ansatz constructed from Hermitian generators $\{\tilde\H_1, \dots, \tilde\H_{K}\}$ consists of a unitary of the form
\begin{align}\U(\boldsymbol{\theta}) = \prod_{l=1}^{L}\prod_{k=1}^{K}e^{-\theta_{(l,k)}i\tilde\H_{k}},
\end{align}
an initial state $\rho = \U_0|\mathbf{0}\rangle\langle\mathbf{0}|\U_0^{\dagger}$, {and a Hermitian measurement operator $\O$.}
\end{definition}
{The output of a VQA is the parameter-dependent expectation value $\<\O\>_{\rho} = \Tr\{\U(\boldsymbol\theta)\rho\U^\dagger(\boldsymbol\theta)\O\}$, known as the \emph{cost function}.}

For $n$-qubits, the set of $\U(\boldsymbol{\theta})$ lies in the unique connected subgroup of $\text{SU}(2^n)$, called the {\emph{dynamical Lie group}~\cite{d2021introduction}. It is the subgroup} associated with the real span of the Lie closure (i.e., closure under taking commutators) of the generators:
\begin{equation}\label{eq:dla}
    \g := \text{span}_{\mathbb{R}}\langle i\tilde\H_1, \dots, i\tilde\H_{K} \rangle_{\text{Lie}},
\end{equation} 
which is known in the quantum control literature as the \emph{dynamical Lie algebra} (DLA)~\cite{d2021introduction}.  We denote the dimension {of $\g$} as a real vector space by $d_{\g}$.

We also informally define the notion of Barren Plateau for quantum ans\"{a}tze.
\begin{definition}[Barren Plateau]\label{def:bp}
    A class of quantum ans\"atze experiences a Barren Plateau if the variance of the cost function gradient decays exponentially with system size, i.e. for all $(l, k)$,
    \begin{align}
        \textup{Var}_{\boldsymbol{\theta} \sim \nu}[\partial_{(l,k)}\langle \O\rangle_{\rho}] \in \mathcal{O}\left(\frac{1}{{b}^n}\right),
    \end{align}
    where the system size $n$ is the number of qubits {and $b > 1$}. Typically $\nu$ is the uniform distribution over the range of the parameters.
\end{definition}
Note that in general a BP at initialization may not imply a BP throughout the training trajectory. However, in most cases when $\nu$ is the uniform distribution over parameters, the collection $\U(\boldsymbol{\theta})$ forms an {approximate} 2-design w.r.t. the Haar measure {\emph{on the dynamical Lie group} (this is made explicit in Section \ref{sec:review})}, and due to Haar invariance, a BP at initialization implies a BP throughout training. 
A PQC that experiences a BP is also called untrainable, which follows from the gradient being computationally-infeasible to estimate to arbitrary precision. Otherwise, if the variance only falls as $\Omega\left(1/\text{poly}(n)\right)$, then the PQC is trainable.

\subsection{DLA - BP Connection}

It has been conjectured  that the dimension of the DLA plays a crucial role in characterizing the trainability of VQAs.
More specifically, the following conjecture linking trainability and DLA dimension was put forward:
\begin{conjecture}[Conjecture 1 in~\cite{Larocca2022diagnosingbarren}, paraphrased]
\label{con:conjecture}
    The scaling of the variance of the partial derivatives of the cost function is inversely proportional to the dimension of the DLA:
    \begin{equation}
         \textup{Var}_{\boldsymbol{\theta} \sim \nu}[\partial_{(l,k)}\langle \O\rangle_{\rho}] \in \mathcal{O}\left(\frac{1}{\poly(d_\g)}\right).
    \end{equation}
\end{conjecture}
{In this work, we provide a proof of this conjecture.}
We emphasize that our results show a more explicit scaling of the variance with the DLA dimension, instead of just an upper bound. Thus, our results shed light on when \emph{stronger} versions of the above conjecture hold, e.g., $\Theta(\frac{1}{\poly(d_\g)})$. However, this  depends on the initial state and observable, since the DLA dimension may not always be the quantity dominating the decay.

{It turns out that the connection holds for a certain class of ans\"{a}tze, which we term the class of \emph{Lie Algebra Supported Ansatz} (LASA).}

\begin{definition}[Lie Algebra Supported Ansatz] A Lie Algebra Supported Ansatz (LASA) is a periodic ansatz where the measurement operator $\O$ is such that $i\O$ belongs to the dynamical Lie algebra associated with the circuit generators $\{i\tilde\H_1, \dots, i\tilde\H_{K}\}$.
\end{definition}

\begin{figure}
    \label{fig:lasa}
    \includegraphics[width=0.45\textwidth]{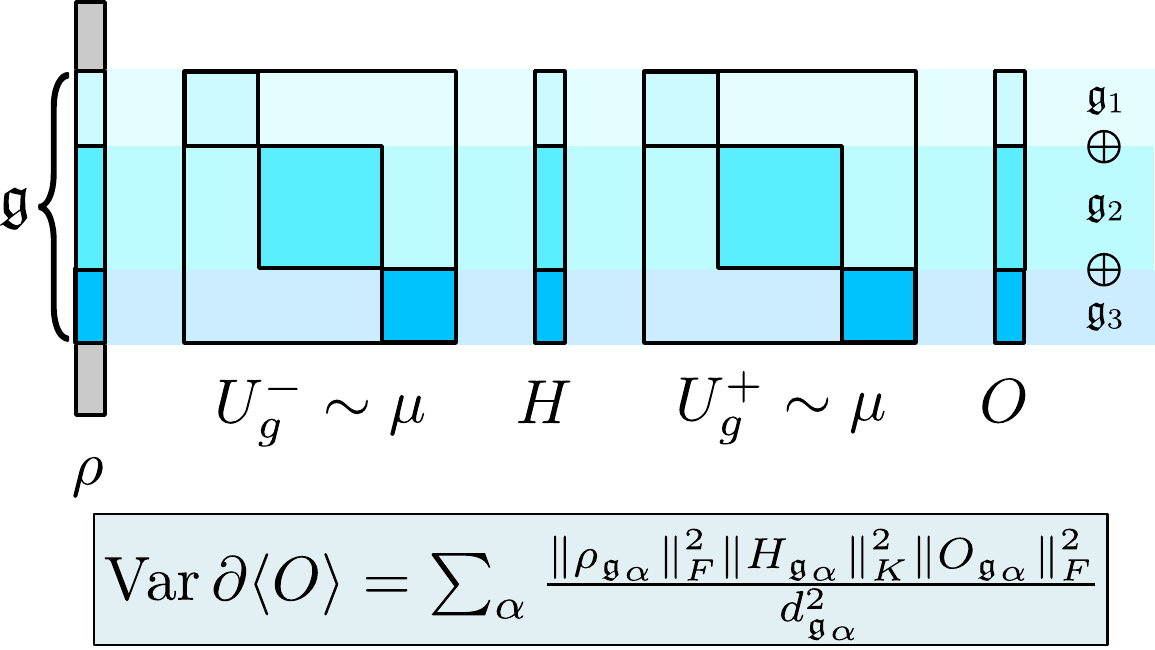}
    {\caption{
    Main result: for the gradient variance, when the observable is in the DLA (as in the case of Lie Algebra Supported Ansatz - LASA) only the components of $\rho$ in the DLA matter, and everything can be computed in the adjoint representation. Specifically, the subscript $\alpha$ for operators corresponds to their orthogonal projection onto the simple ideal $\g_{\alpha}$. When the DLA has multiple ideals, each ideal individually contributes a term to the variance.
    }}
\end{figure}

{In Fig.~\ref{fig:lasa}, we display our main result, which shows that the variance of the gradient has a direct dependence on DLA dimension for LASAs. As will be made  rigorous later, by construction, the action of a LASA on its observable will decompose into invariant subspaces (corresponding to preserved symmetries) each of dimension at most $d_{\g}$.}

While we introduce restrictions on the observable, we note that our results are still far-reaching. This is because LASAs include many commonly-used PQCs such as the Hamiltonian variational ansatz (HVA)~\cite{Wecker_2015} and quantum alternating operator ansatz (QAOA)~\cite{farhi2014quantum, Hadfield_2019}. We also note that all LASAs are equivariant quantum neural networks (EQNNs)~\cite{nguyen2022theory}. %
However, an EQNN is not necessarily a LASA, since there are equivariant operators that may not lie in the DLA. This is because equivariance is defined with respect to a symmetry group of the quantum data, and one could imagine a situation in which the circuit has a small DLA such that other equivariant operators exist outside the DLA.

\subsection{Representation Theoretic Notation}\label{sec:review}

{The following presents the notation used throughout the paper and assumes familiarity with Lie groups and representation theory. The unfamiliar reader is directed to the \emph{Supplementary Information}, where we briefly introduce Lie groups and representation theory.}

{Our focus will be a compact, connected Lie group $G$. The corresponding compact Lie algebra will be denoted by $\g$. The notation $V$ will represent an arbitrary finite-dimensional inner product space over $\mathbb{C}$ or $\mathbb{R}$. If a result does not specify which field is used, then either can be assumed. Additionally, $\mathcal{U}(V)$ will denote the group of isometries on $V$ (i.e., depending on the field,  either the unitary group or orthogonal group), and $\mathfrak{u}(V)$ will denote the set of skew-Hermitian operators on $V$. For either $\mathbb{R}$ or $\mathbb{C}$, we will use $\phi : G \rightarrow \mathcal{U}(V)$ to denote a unitary representation of the group $G$ and $d\phi : \g \rightarrow \mathfrak{u}(V)$ to denote the differential or Lie algebra representation. We will frequently use the notation $\U_{g}$ to denote the element $\phi(g) \in \mathcal{U}(V)$ for some $g \in G$ when the representation $\phi$ and space $V$ are clear from the context.} 

{Recall that the adjoint representation of a Lie group $G$ is the homomorphism:
\begin{align}
    \forall g \in G, \text{Ad}_{g}(k):=gkg^{-1} \in \g, \forall k \in \g,
\end{align}
and the adjoint representaton of a Lie algebra is the homomorphism:
\begin{align}
    \forall h \in \g, \text{ad}_{h}(k):=[h, k] \in \g, \forall k \in \g.
\end{align}}
{For compact simple Lie algebras, since all trace forms are related by a real factor,} we define a scaling constant $I_\phi$ that we call the \textit{index of the representation} (w.r.t. the standard representation) such that:
\begin{equation}
    -\Tr (d\phi(e_i)d\phi(e_j)) = I_\phi \delta_{ij},
\end{equation}
{for $\{ e_i\}$ a basis for $\g$ satisfying:
\begin{align}
    -\Tr(e_ie_j) = \delta_{ij}.
\end{align}}
The constant $I_\phi$ is the same as (twice) the Dynkin index for irreducible representations~\cite{fuchs1995affine}.

{For compact simple Lie algebras we consider a few norms induced by the trace forms. For any $a \in \g$, we define the \emph{standard norm} to be
\begin{equation}
    \|a\|_\g^2 = -\Tr(a^2),
\end{equation}
the \emph{Killing norm}
$\lVert a \rVert_{\text{K}}^2$ to be the norm induced by the Killing form (trace form associated with the adjoint representation), and more generally,  for an arbitrary Lie algebra representation $d\phi$, we denote the usual \emph{Frobenius norm} by $\|d\phi(a)\|_{\text{F}}^2$.  All are related in the natural way via the associated index of the representation, as defined earlier. Specifically, for arbitrary $d\phi$:
\begin{align}
\lVert d\phi(a) \rVert_{\text{F}} &= I_{\phi}\lVert a \rVert_{\g}^2\\
\lVert d\phi(a) \rVert_{\text{K}}^{2} &= \lVert a \rVert_{\text{K}}^{2}=I_{\text{Ad}}\lVert a \rVert_{\g}^2 = \frac{I_{\text{Ad}}}{I_{\phi}}\lVert d\phi(a) \rVert_{\text{F}}^2.\label{eq:killing}
\end{align}
For an arbitrary $\X \in \mathfrak{u}(V)$, we define $\X_{\g}$ to be the orthogonal projection under the Frobenius inner product onto $d\phi(\g)$.}

{Lastly, throughout the paper, all integration, e.g. $\int_{G} f(g) dg$, is with respect to the Haar measure $\mu$ for $G$. The notation $\mu^{\otimes 2}$ will denote the product Haar measure.}
\\

{Let us now place these notions in the context of VQAs. The vector space $V$ on which the group acts is the $n$-qubit Hilbert space $\mathbb{C}^{2^n}$. In general the PQC's dynamical Lie group will be $\phi(G)$ with $\phi$ a faithful (injective) representation and this is what we will assume here. In practice however one always can take $\phi$ to be the identity map, identifying $G$ with the dynamical group and $\g$ with the DLA, without invalidating the results.}

{In this abstract setting there is no notion of parameter space and hence the PQC gradient $\partial_{(l,k)}\langle \O\rangle_{\rho}$ is not well defined. Thus, we introduce the following parameter-independent quantity associated with any compact, connected Lie group:
\begin{definition}[Abstracted Gradient]
{Let $G$ be a compact, connected Lie group with representation $\phi : G \rightarrow \mathcal{U}(V)$. In addition, let $h \in \g$ and $i\O, i\A \in \mathfrak{u}(V)$.   We define the \emph{abstracted gradient} to be the following quantity:}
\begin{align}
    \partial\langle \O \rangle_{\A} := \textup{Tr}\{\U_{g^{-}}^{\dagger}\A\U_{g^{-}}[\H, \U_{g^{+}} \O \U_{g^{+}}^{\dagger}]\},
\end{align}
where $\U_{g^{\pm}} := \phi(g^{\pm})$ for arbitrary $g^{+}, g^{-} \in G$, and $\H = d\phi(h)$.
\end{definition}
Note that now we set the generators to be skew-Hermitian. 
The connection between abstracted and PQC gradients is clear for the periodic ansatz in Definition~\ref{def:ansatz}: for any parameter $\theta_{(l, k)}$ the PQC gradient will be equivalent to an abstracted gradient, with $\U_{g^-}$ ($\U_{g^+}$) being the unitaries preceding (following) the unitary $e^{-\theta_{(l, k)} \H_{k}}$ in the circuit.}

{In our calculations we will look at second moments of the abstracted gradient for $(g^+, g^-) \sim \mu^{\otimes 2}$. This will accurately model the experimental behaviour if for any $\theta_{(l, k)}$ the ansatz takes the form $\W^{(\text{L})}e^{-\theta_{(l, k)} \H_{k}}\W^{(\text{R})}$ with $\W^{(\text{L}/\text{R})}$ random unitaries forming independent $2$-designs for $\phi(G)$.
}
{For a sufficiently deep periodic ansatz, the assumption is valid for parameters in the middle of the PQC whenever randomly initialized, polynomially-sized periodic ans\"atze form approximate 2-designs.}
{It has been shown that this holds for $\g = \mathfrak{su}(2^n)$ or $\mathfrak{so}(2^n)$ and when all generators are in the Pauli group~\cite{haah2024efficient}. It has been widely assumed in literature that this result still holds for ansatz with different DLAs, with only numerical evidence. The following result answers this in the affirmative for LASA with polynomially-sized DLA, showing that rapid mixing to $2$-design still holds when we sample generators from a basis for the DLA.}

{\begin{theorem}[Rapid mixing for polynomial DLA]
\label{thm:rapid}
Consider an orthogonal basis of skew-Hermitian generators $\mathcal{A}:=\{\H_1, \dots, \H_{d_{\g}}\}$ for the DLA with the property that the unitary $e^{-\theta\H_k}$ corresponding to a generator $\H_{k}$ is $t_k$-periodic. In addition, suppose that $d_{\g} = \mathcal{O}(\textup{poly}(n))$. Consider a LASA formed by applying evolutions $e^{-\theta_k \H_k}$ where $\H_{k}$ is selected uniformly at random from the set $\mathcal{A}$ and the parameter $\theta_{k}$ uniformly from $[0, t_k)$. Then, the ansatz is an $\epsilon$-approximate $2$-design for the dynamical group $\mathcal{G}$  after  $\mathcal{O}(\textup{poly}(n)\log(1/\epsilon))$ layers.
\end{theorem}}
{The proof of the above result and its generalization to $t$-designs for arbitrary LASA are in the \emph{Supplementary Information} and are based on techniques used by Ref.~\cite{haah2024efficient} and earlier works. %
Such random walks have been known to converge for some time~\cite{Harrow_2009}, and that  convergence to Haar for exponential DLA is not efficient. However, the above result makes the spectral gap dependence explicit.}

{The approach of studying BPs with 2-designs is standard, e.g. see Ref.~\cite{mcclean2018barren}. Furthermore, as we have shown, it is theoretically motivated in the case of independent, uniformly distributed parameters. However, there may still be settings where the 2-design assumption fails and where our results will not hold, for example other initialization schemes or correlated parameters. Interestingly, there is evidence that both may avoid BPs~\cite{zhang2022escaping, volkoff2021large}, however do not investigate this research direction further.}

Inspired by our overall goal
of analyzing BPs in parameterized quantum circuits, we seek to compute the quantity
\begin{align}
\label{eqn:variance}
    &\textup{GradVar} 
    := \text{Var}_{(g^{+}, g^{-}) \sim \mu^{\otimes 2}}[\partial\langle \O \rangle_{\rho}] \nonumber\\&= \mathbb{E}_{(g^{+}, g^{-}) \sim \mu^{\otimes 2}}[(\partial\langle \O \rangle_{\rho})^2] - (\mathbb{E}_{(g^{+}, g^{-}) \sim \mu^{\otimes 2}}[\partial\langle \O \rangle_{\rho}])^2,
\end{align}
where $\mu$ is the unique Haar measure over $G$ and $\rho$ is the initial quantum state to which all elements of the dynamical group are applied.
{$\mathbb{E}_{(g^{+}, g^{-}) \sim \mu^{\otimes 2}}[\partial\langle \O \rangle_{\rho}]$ can be shown to be zero is general (see \emph{Supplementary Information}), and thus in practice we focus on the second moment:
\begin{equation}\label{eq:sec_moment}
    \text{GradVar} = \mathbb{E}_{(g^{+}, g^{-}) \sim \mu^{\otimes 2}}[(\partial\langle \O \rangle_{\rho})^2].
\end{equation}}

{Using Definition~\ref{def:bp}, a BP occurs when the following holds:}
    \begin{equation}
        \textup{GradVar} \in \mathcal{O}\left(\frac{1}{{b}^n}\right), \;\; {b > 1}.
    \end{equation}
This is the phenomenon that our methods will seek to probe for the specific case of LASAs.

\subsection{Theory of BPs for LASA}

{We now present our theoretical contributions,} which connect the Lie algebra dimension to the scaling of the gradient variance. We note that norms involving the Hermitian observable $\O$ and the skew-Hermitian generator $\H$ have a few interpretations {as mentioned in Section \ref{sec:review}.} However, to be concise and for readability, we present the results in only one form.

{We start by recalling that all compact Lie algebras (and thus groups) are reductive. 
\begin{definition}[Reductive Lie algebra~\cite{hall2013lie}]
A Lie algebra $\g$ is reductive if the adjoint representation is completely reducible, i.e., $\g$  has the following decomposition as a direct sum of Lie algebras:
 \begin{align}
 \label{eqn:splitting}
\g = \bigoplus_{\alpha}\g_{\alpha} \oplus \mathfrak{c},
\end{align}where each $\g_{\alpha} \subset \g$ is a simple ideal and $\mathfrak{c} \subset \g$ is the center of $\mathfrak{g}$. Note that if $G$ is simply connected then $\mathfrak{c} = \{0\}$. 
\end{definition}
This property is essential for proving our main result, as it allows us to extend our expression (Theorem \ref{thm:simple_group}) for the gradient variance for simple Lie groups to the general compact case. If $\g$ is compact, then the $\g_{\alpha}$ will be compact as well~\cite{wiersema2309classification}.
Note that this notion of reducibility is related to what has appeared in prior works, e.g.~\cite{Larocca2022diagnosingbarren, schatzki2022theoretical,kerenidis2022quantum}, the differences are mainly as to whether the group acts on the observable or state. 
We discuss this in detail in Section \ref{sec:comparison}.}

{Next, we present our expression for the variance of the gradient for compact simple groups that applies to each $\g_{\alpha}$ in Equation \eqref{eqn:splitting}.}

\begin{theorem}[Simple group variance]
\label{thm:simple_group}
Let $G$ be a compact, connected simple Lie group with Lie algebra $\g$.
Suppose $\phi$ is a finite-dimensional unitary representation of $G$. 
In addition,  $o, h \in \g$, $i\O = d\phi(o)$, $\H = d\phi(h)$ and $\rho$ a density matrix. 
Then the following holds:
\begin{equation}
\label{eqn:simple_group_gradvar}
   \textup{GradVar}
    =\frac{\lVert \H \rVert_{\textup{K}}^{2}\lVert \O \rVert_{\textup{F}}^{2} \lVert  \rho_{\g} \rVert_{\textup{F}}^{2}}
    {d_{\g}^{2}}.
\end{equation}
\end{theorem}

{If $G$ is compact, one can use the fact that it is reductive and apply Theorem \ref{thm:simple_group} to each of the compact simple ideals to obtain the following:}

\begin{theorem}[Compact group variance]
\label{thm:compact}
Let $G$ be a compact and connected Lie group with Lie algebra $\g$. Suppose $\phi$ is a finite-dimensional unitary representation, $o, h \in \g$, $i\O = d\phi(o)$, $\H = d\phi(h)$, and $\rho$ is a density matrix. 
Then the following holds:
\begin{equation}
\label{eqn:compact_var_eqn}
   \textup{GradVar} = \sum_{\alpha} \frac{\lVert  \H_{\g_{\alpha}} \rVert_{\textup{K}}^{2} \lVert  \O_{\g_{\alpha}} \rVert_{\textup{F}}^{2} \lVert\rho_{\g_{\alpha}} \rVert_{\textup{F}}^{2}}{d_{\g_{\alpha}}^{2}} .
\end{equation}
Note that the center $\mathfrak{c}$ does not contribute to the variance. %
\end{theorem}

{As mentioned in the introduction, the above theorem} is the central result of the paper. It shows that under the assumption of a LASA we can get a \textit{precise mathematical expression} for the gradient variance. Notably this expression is in terms of quantities that are intimately linked with the Lie algebra and the representation, and are well characterized for all simple algebras. %

\subsection{Interpretation of Results}
\label{sec:interpretation}
{The three norms in the numerator can be viewed as effectively measuring the support that each operator has on the simple ideal $d\phi(\g_{\alpha})$. Specifically, $\lVert \O_{\g_{\alpha}}\rVert_{\text{F}}$ and  $\lVert \rho_{\g_{\alpha}}\rVert_{\text{F}}$ being Frobenius norms can be interpreted as generalized measures of purity with respect to $d\phi(\g_{\alpha})$. This concept was actually first introduced in Ref.~\cite{somma2004nature}. A similar interpretation is also valid for the Killing norm $\lVert \H_{\g_{\alpha}} \rVert_{\text{K}}$, however, this time the relevant representation of $\g_\alpha$ is the adjoint representation, and so the norm is scaled by the ratio of the indices as in Equation~\eqref{eq:killing}.} 

{If one is still uncomfortable with the Killing norm,} we note that  $\lVert \H_{\g_{\alpha}} \rVert_{\text{K}}^{2} \le 2d_{\g_{\alpha}} \lVert \H_{\g_{\alpha}} \rVert_{\text{F}}^{2}$ (see the \emph{Supplementary Information}), and so one gets the following upper bound: 
\begin{align}
\label{eqn:upper_bound_grad_var}
    \textup{GradVar} 
    \in \mathcal{O}\left(\sum_{\alpha} \frac{\lVert  \H_{\g_{\alpha}} \rVert_{\text{F}}^{2} \lVert  \O_{\g_{\alpha}} \rVert_{\text{F}}^{2} \lVert \rho_{\g_{\alpha}} \rVert_{\text{F}}^{2}}{d_{\g_{\alpha}}}\right),
\end{align}
which presents the result in terms of more familiar quantities, i.e. Frobenius norms.
{In addition,} we now see that Conjecture \ref{con:conjecture} is explicitly proven (and indeed significantly generalized) for LASA.

{From Equation  \eqref{eqn:compact_var_eqn} or \eqref{eqn:upper_bound_grad_var}, we infer} that a BP can only occur whenever at least one of the terms in the expression leads to exponential decay. More specifically, the gradients will decay exponentially under any of these conditions: the state has exponentially small support over the Lie algebra; the state, the measurement operator and the generator are mostly supported on a subalgebra, $\g_{\alpha}$, the dimension of which is exponentially large; or the support of the state, measurement operator and generator are mutually incompatible on the subalgebras, in the sense that all terms vanish. The second condition amounts to the conjecture of Ref.~\cite{Larocca2022diagnosingbarren}, while the last is a novel prediction of this work, which only occurs in the strict semisimple case.

{Lastly, we conclude with some details on how one might use our results in practice. Since the generators and observables will typically be linear combinations of Pauli strings, one can  try to utilize symbolic computation to reason about the decomposition of $\g$ into simple ideals. A basis for the DLA can be obtained by computing nested commutators symbolically and checking for linear independence as done in Ref.~\cite{Larocca2022diagnosingbarren}, and a stopping condition can be added to avoid an exponential amount of iterations when the DLA ends up being exponential in dimension.}
{If the DLA has polynomial dimension and the generators are sums of Pauli strings, then there is an efficient procedure for discovering the simple ideals. Given a basis for the DLA $\{\E_{k}\}_{k=1}^{d_{\g}}$, compute the $d_{\g} \times d_{\g}$ matrices for each operator $\text{ad}_{\E_{k}}$ in the basis $\{\E_{k}\}_{k=1}^{d_{\g}}$, denote them $\widehat{\text{ad}_{\E_{k}}}$ . These are only polynomially large matrices. The next step is to simultaneously block diagonalize the $\widehat{\text{ad}_{\E_{k}}}$, which will reveal bases for the simple ideals. We can compute the $\lVert\O_{\g_{\alpha}}\rVert_{\text{F}}^2$ and $\lVert\H_{\g_{\alpha}}\rVert_{\text{K}}^2$ norms symbolically. If $\{\A_k\}_{k=1}^{d_{\g_{\alpha}}}$ is a basis for the ideal $\g_{\alpha}$, which can be expressed in terms of sums of Pauli strings given our assumption, then the norm $\lVert\rho_{\g_{\alpha}}\rVert_{\text{F}}^2 = \sum_{k=1}^{d_{\g_{\alpha}}}\Tr(\A_k^{\otimes 2}\rho^{\otimes2})$ can be computed on a quantum computer.}

\subsection{Variance Computation for Quantum Compound Ansatz}%
\label{sec:compound_example}

The quantum compound ansatz is a quantum representation on $n$ qubits ($2^n$-dimensional) of the Lie group $\textup{SO}(n)$ or $\textup{SU}(n)$~\cite{kerenidis2022quantum, cherrat2022quantum}.
Given a general $g \in \textup{SU}(n) ~(\textup{SO}(n))$, one can decompose it into a product of $\textup{SU}(2)~(\textup{SO}(2))$ rotations on 2-dimensional subspaces, which are (generalized) Givens rotations:
\begin{equation}
    \U_{g} = \prod_{(i,j) \in E} \U^{\text{Givens}}_{ij}(g),
\end{equation}
{and are implemented using the fermionic beam splitter (FBS) gate defined in Ref.~\cite{kerenidis2022quantum}.}

The graph $E$ can have various topologies, for example a pyramid or a staircase.
The circuit preserves Hamming weight, and the representation splits into subspaces corresponding to the different Hamming weights. {The analysis of the gradient variance for a more general class of Hamming-weight preserving unitaries appears in Ref.~\cite{monbroussou2023trainability}.}

One can check that the appropriate representation for the generators of a $\textup{SU}(2)$ Givens rotation between qubit $i$ and $j$ is
\begin{align}
    h_x^{ij} &= -\frac{i}{4} (\sigma_x^i \otimes  \sigma_x^j + \sigma_y^i \otimes \sigma_y^j)  \otimes \sigma_z^{\otimes |i-j-1|} \\&= d\phi\left(-\frac{i}{2}X^{(ij)}\right) \\
    h_y^{ij} &= -\frac{i}{4} (\sigma_y^i \otimes \sigma_x^j - \sigma_x^i \otimes \sigma_y^j) \otimes \sigma_z^{\otimes |i-j-1|} \\ &= d\phi\left(-\frac{i}{2}Y^{(ij)}\right) \\
    h_z^{ij} &= -\frac{i}{4}(\sigma_z^i - \sigma_z^j) = d\phi\left(-\frac{i}{2}Z^{(ij)}\right),
\end{align}
where $X^{(ij)}, Y^{(ij)}, Z^{(ij)}$ act as the Pauli operators $\sigma_{x}, \sigma_{y}, \sigma_{z}$ on the $2 \times 2$ block formed by $i$ and $j$, respectively, and are zero otherwise. They are elements of $\mathfrak{su}(n)$, and $\phi$ is the direct sum of the alternating representations for $k=1, \dots, n$, i.e.:
\begin{align}
    V = \bigoplus_{k=1}^{n}\wedge^{k}\mathbb{C}^{n}.
\end{align}
Note that the norm of each of these generators in $\g$ is $1/2$.
Importantly, while the set of generators spans the representation of $\g$, since it is larger than the dimension of $\g$ it is a not linearly independent set.
Note the extra $\sigma_z$'s in the definition of $h_x$ and $h_y$ are reminiscent of the string of $\sigma_z$ in the Jordan--Wigner encoding, only that here they are needed for the algebra to close.
The $\textup{SO}$ case is generated by the $h_y^{ij}$ elements only.

{To clarify why the ansatz is subspace uncontrollable, we can consider the Hamming-weight $n/2$ subspace. On this subspace, the DLA is isomorphic to $\mathfrak{su}(n)$, while the Lie algebra of the  full space of unitary operators on this subspace is isomorphic to $\mathfrak{su}(\binom{n}{n/2})$, hence the compound ansatz cannot enact all unitary transformations.}

{Before proceeding we present a mixing time result to $t$-design for the quantum compound ansatz that is tighter than Theorem \ref{thm:rapid}.
\begin{theorem}[Rapid mixing for Compound Ansatz]
\label{thm:mix_compound}
Consider an $n$-qubit quantum compound ansatz that is a LASA constructed using the set of generators $\{X^{(ij)}$,$Y^{(ij)}$, $\sum_{i=1}^{j} Z^{(ij)}\}$ with rotations angles chosen uniformly at random. Then,  for $t \leq n/2$, the ansatz is an $\epsilon$-approximate $t$-design for the dynamical group $\textup{SU}(n)$  after  $\mathcal{O}(tn\log(1/\epsilon))$ layers.
\end{theorem}
Of course, for BPs $t=2$ is the main interest. The proof follows simply from Theorem \ref{thm:rapid} and is left to the \emph{Supplementary Information}. Not that for the chosen set of generators some of the randomly chosen angles are not independent (i.e. the $\sum_{i=1}^{j}Z^{(ij)}$ type generators).
}

{The following three results utilize our theory of BPs for LASA to show that the quantum compound ansatz can be BP free under uniform initialization.}

{\begin{theorem}
\label{thm:compound_basis_state}
    For a quantum compound ansatz that is also LASA, if the initial state is a computational basis state, then the following holds:
    \begin{align}
        \textup{GradVar} \in \Omega\left(\frac{1}{n^3}\right).
    \end{align}
\end{theorem}}

The conclusion is that $\text{SU}$-compound layers with Lie-algebra supported measurements do not have BPs for any fixed Hamming weight computational basis state.
Note that computational basis states of the same Hamming weight are in an irreducible subspace of the tensor product representation {(see the \emph{Supplementary Information})}.%

\begin{figure}
    \centering
    \includegraphics[width=0.45\textwidth]{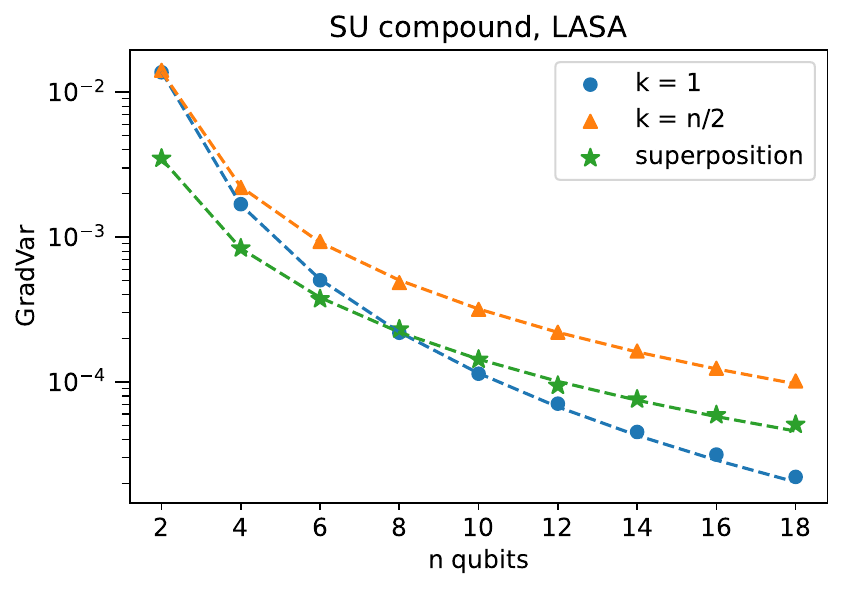}
    {\caption{Gradient variance scaling for SU compound layers, LASA. Dots are numerical results while dotted lines are analytical predictions using the equations in the text. Showing results for computational basis input states of Hamming weight 1 and $n/2$ and the uniform superposition state $|+\>^{\otimes n}$, for $n$ number of qubits ranging from 2 to 18 in steps of 2. The measurement operator is $-ih_{12}^z = (\sigma_1^z - \sigma_2^z)/4$.
    Accounting for the randomness of initialization, there is good agreement of numerical results with the predictions. The error bars are too small to plot. Additional information on the numerics is in the \emph{Supplementary Information}.}}
    \label{fig:su_compound}
\end{figure}

Next, we consider the uniform superposition state $|\psi\rangle = |+\rangle^{\otimes n}$ { and show that the quantum compound ansatz is still BP-free. In addition, in this case, the variance decays exactly with the DLA dimension $n^2-1$.}

{\begin{theorem}\label{thm:compound_uniform_sup}
    For a quantum compound ansatz that is also LASA, if the initial state is a uniform superposition of all computational basis states, then the following holds:
    \begin{align}
        \textup{GradVar} \in \Theta\left(\frac{1}{n^2}\right).
    \end{align}
\end{theorem}}

Thus, we also have no BP with the initial state being the uniform superposition.
We numerically verified the predictions for the various initial states as shown in Fig.~\ref{fig:su_compound}.

Finally, we see how the result can be extended to cover single-qubit measurements. %

{\begin{corollary}\label{cor:compound_non_lasa_single}
    For a quantum compound ansatz with an observable that is composed of single qubit measurements, and if the initial state is a computational basis state or the uniform superposition of all computational basis states, then the following holds:
    \begin{align}
        \textup{GradVar} \in \Omega\left(\frac{1}{\textup{poly}(n)}\right).
    \end{align}
\end{corollary}}
We verify these predictions in Fig.~\ref{fig:single_z}.
\begin{figure}
    \centering
    \includegraphics[width=0.45\textwidth]{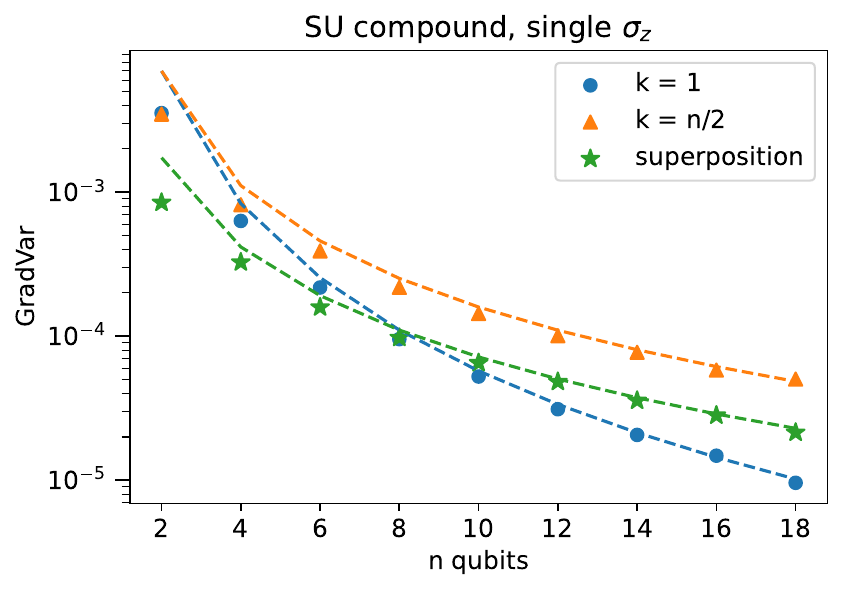}
    \caption{Gradient variance scaling for SU compound layers, non-LASA. The setup is identical to the LASA case, except that here the measurement operator is a single $\sigma_z/4$. We show the analytical prediction derived from the LASA case as explained in the text, and therefore we see a disagreement with numerics, implying that the covariance term is nonzero. Still the scaling is similar, and additionally the numerics converge to the prediction at larger system sizes. {The error bars are too small to plot. Additional information on the numerics is in the \emph{Supplementary Information}.}}
    \label{fig:single_z}
\end{figure}

This answers an open question proposed in Ref.~\cite{cherrat2023quantum}. As a final note, even though $\sigma_i^z$ does not lie in the DLA, single-qubit expectations of observables with respect to the compound ansatz starting from a product state are still know to be classically simulatable~\cite{Brod_2016}.

{Lastly, we present another setting in which  the observable does not lie in the DLA, but this time, the quantum compound ansatz has a BP.
\begin{theorem}
\label{thm:compound_projector}
    For the quantum compound ansatz if the initial state is a computational basis state with Hamming-weight $\frac{n}{2}$ and the observable is a rank-one projector onto another computational basis state in this space, then
    \begin{align}
        \textup{GradVar} \in \mathcal{O}\left({n \choose n/2}^{-1}\right).
    \end{align}
\end{theorem}
We verify the scaling in Fig.~\ref{fig:su_compound_proj}.}

\begin{figure}
    \centering
    \includegraphics[width=0.45\textwidth]{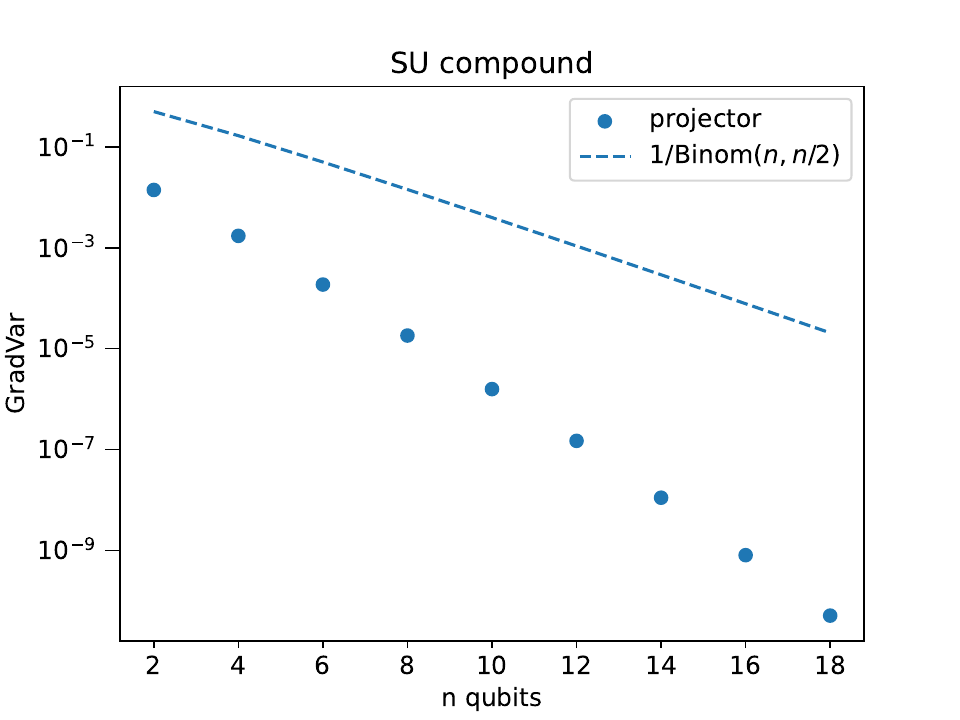}
    \caption{Numerics for gradient variance of SU compound layer, with input state is a computation basis state of Hamming weight $n/2$ and the observable is a projection onto the same state. The resulting algorithm is not a LASA, and indeed displays a BP. We also show the upper bound of ${n \choose n/2}^{-1}$ which appears to be very loose. {The error bars are too small to plot. Additional information on the numerics is in the \emph{Supplementary Information}.}}
    \label{fig:su_compound_proj}
\end{figure}

{Intuitively, the above decay comes from the fact our choice of observable and initial state are rank one projectors, and thus the overlap of traceless parts of both operators will spread across an exponentially large subset of  $\mathfrak{su}$.}
{Theorem \ref{thm:compound_projector} is interesting because the compound ansatz is not very expressive and the depth of the circuit exceeds the shallow regime of $O(\log(n))$~\cite{Cerezo_2021}. We note that the cost function we choose is still global.}

{The details of how the numerical results were obtained are described in the \emph{Supplementary Information}.} %

\subsection{Comparison with Previous Approaches}
\label{sec:comparison}

{As mentioned in the Introduction, previous approaches have taken a state-first or Schr\"{o}dinger picture view point. Specifically, under the action of $G$, the quantum state space $V$ will decompose into invariant subspaces:
\begin{equation}
    V = \bigoplus_{\kappa}V_{\kappa},
\end{equation}
each of which is acted upon by the subrepresentation $\phi_\kappa(G)$. This decomposition is in line with the symmetries that the ansatz obeys, i.e. its commutant~\cite{nguyen2022theory}. If the initial state $\rho \in V_{\kappa}$, then since $G$ preserves this space, the variance calculation restricts to integrating over $\phi_\kappa(G)$. If the restriction of the DLA $d\phi_{\kappa}(\g)$ to the invariant subspace is isomorphic to $\mathfrak{su}(\dim V_{\kappa})$, then one says that PQC is \emph{subspace controllable} on $V_{\kappa}$, otherwise, it is \emph{subspace uncontrollable}. The calculation is possible in the subspace-controllable setting via the Schur-Weyl duality~\cite{Larocca2022diagnosingbarren}, but the subspace uncontrollable setting poses significant obstacles to the calculation of the second moment (Equation~\eqref{eq:sec_moment}) using this approach.}

{In our setting we are instead using the Heisenberg picture and, assuming LASA, considering the action of $d\phi(\g)$ on itself via conjugation, so $V = d\phi(\g)$ in this case and $d\phi$ is the adjoint representation.
Notice that if the DLA is reductive (Equation~\eqref{eqn:splitting}) and $\phi$ is faithful (injective), the decomposition of $V$ respects the decomposition into simple ideals:
\begin{align}
\label{eqn:spliting}
d\phi(\g) = d\phi(\bigoplus_{\kappa} \g_{\kappa}) = \bigoplus_{\kappa} d\phi(\g_{\kappa}) = \bigoplus_{\kappa} d\phi(\g)_{\kappa}.
\end{align}
Thus, the Lie algebra being reductive implies that the adjoint representation splits into irreducible invariant subspaces, which are precisely the simple ideals $d\phi(\g_{\kappa})$. As detailed in \emph{Methods}, this is sufficient to calculate the second moment for any compact Lie group.}

{So in this setting, we always know the invariant subspaces and the representation acting on them, namely the $d\phi(\g_{\kappa})$ and the corresponding adjoint representation. This is a significant simplification from the Schr\"odinger picture approach and enables us to completely circumvent the obstacles posed by the subspace uncontrollable setting. Notice finally that now the invariant subspaces $d\phi(\g_{\kappa})$ reflect the symmetries that are preserved by the evolution of observable instead of the state, so while related this is a different concept of PQC symmetry than the one prior work had explored.}

{Lastly, we would like to emphasize that DLA does not always split into a direct sum over the decomposition of  $V$ into $V_{k}$ for an arbitrary unitary representation. However, this does hold if $d\phi(\g)_{\kappa}$ for subspace $V_{\kappa}$ is simple, like it is for the adjoint and in Ref.~\cite{schatzki2022theoretical}. More specifically, the condition implies that $d\phi(\g)_{\kappa}$ must then be $d\phi(\g_{\alpha})$ for some simple ideal $\g_{\alpha}$.}

\section{Discussion}
In this work we present a general framework for diagnosing the BP phenomenon in Lie algebra supported ans\"{a}tze, which include popular PQCs, such as HVA, QAOA and various equivariant QNNs. Our main contribution is a method that explains the previously mysterious connection between the dimension of the DLA and the rate at which gradients decay. This method has enabled us to analyze the gradient variance for subspace uncontrollable circuits, such as the quantum compound ans\"{a}tze, which was not previously possible with existing techniques from literature.

We note that the kinds of circuits where the simulatability results of Ref.~\cite{goh2023liealgebraic} apply are exactly LASAs. In fact, many of the techniques employed here are similar. As the aforementioned paper links the dimension of the DLA to the performance of the classical simulation of expectations via their algorithm $\g$-sim, we see that at least for LASAs there is a connection between vanishing gradients and simulatability, in the sense that a LASA with polynomial DLA does not see BPs but is simulatable. We note that our results could be applied to the DLAs that have been classified by Ref.~\cite{wiersema2309classification}.

Regarding general VQAs, {when the observable has support outside of the DLA,  we show in the \emph{Supplementary Information} that the same techniques used in the LASA setting can be used to obtain gradient variance expression for general ansatz. Unfortunately,} it can be challenging to determine gradient variance scaling from these expressions in general. Characterizing the gradient variance in this setting would potentially allow for constructing ans\"{a}tze that both do not have BPs and do not have classically simulated expectations. Existing literature has already shown that when the observable lies in the DLA and the DLA has polynomially-growing dimension, then the computation of expectation values can be classically simulated. Potentially, the gradient variance can be shown to still scale inversely with the DLA dimension when the observable has only some small support outside of the DLA, as we have shown for the quantum compound ansatz (Corollary \ref{cor:compound_non_lasa_single}).

Lastly, BPs only correspond to one of two issues that plague VQAs. As mentioned earlier, like BPs, the convergence of VQAs has also only been theoretical characterized in the subspace-controllable setting~\cite{you2022convergence}. Potentially, the framework we have developed can be applied to understanding the projected gradient dynamics that occurs in the uncontrollable setting.
\\

\emph{Note on  Ref.}~\cite{ragone2023unified}:
During the writing of the manuscript, we became aware through a comment in~\cite{wiersema2309classification} that Michael Ragone \textit{et al.} have independently obtained a proof of an extension of the conjecture in Ref.~\cite{Larocca2022diagnosingbarren}. This was later released in Ref.~\cite{ragone2023unified}. We encourage the reader to review both papers for a richer picture of the solution, however we summarize here the most important differences between our works. The main one is that the work of Ragone \textit{et al.} focuses on cost function concentration as opposed to concentration of the partial derivatives. The authors mention, by citing Ref.~\cite{Arrasmith_2022}, that loss function concentration implies concentration of the partial derivatives, and thus provide bounds. However, in our case, we obtain exact expressions for the variance of the partial derivatives, thus revealing the connection between the gradient variance scaling and the Killing norm of the generators. In addition, we include explicit formulae for the gradient variance for the quantum compound ansatz in commonly-used settings, which leads to the novel prediction that it can avoid BPs under Haar initialization. Lastly, we include a discussion on the application of our techniques to observables that lie outside of the DLA. The work by Ragone \textit{et al.} however does include a broader discussion that links BPs in symmetric ans\"atze to other known causes of BPs, including cost-function-induced~\cite{cerezo2021cost} and noise-induced~\cite{wang2021noise}, and thus places the result into a wider context.

{\section{Methods}
In this section, we formally derive the connection between the DLA dimension and the gradient variance, leading to our theory of BPs. Specifically, we present the proofs of the majority of the theorems shown in the \emph{Results} section, the rest are left to the \emph{Supplementary Information}. The main tools that we utilize are the concepts of the adjoint representation and Schur orthogonality. }

{\subsection{The Adjoint Representation Connection}
We start by providing some explanation as to why the connection between the DLA dimension and BPs that agrees with existing numerical evidence is not obvious. It will be the adjoint representation that makes the relationship clear and allow for exact computation of the gradient variance that agrees with existing numerics.}

{As in earlier parts of the text,  the dynamical group $\G$ associated to a periodic ansatz is a unitary representation of some other Lie group $G$. Thus the representation $\phi : G \rightarrow \text{SU}(2^n)$ corresponds to $G$ acting on the $n$-qubit Hilbert space $V$ and $\phi(G) = \G$. Let $\mathcal{M}(\mathbb{C}, 2^{2n})$ denote the set of $2^{2n} \times 2^{2n}$ complex matrices.}

{Before proceeding we make a small note on the compactness of  the dynamical group.} While the dynamical group $\G$ is obviously connected, {it may not be compact (due to lack of closure).} An example is the irrational flow on a torus that occurs when the generators $\tilde{\H}_{i}$ have at least two eigenvalues whose ratio is irrational. The action of these generators will lead to non-periodic orbits. Notice that such non-periodic ans\"atze can occur in principle, for example in QAOA on graphs with random weights. {However, since any Lie subalgebra of $\mathfrak{su}(2^n)$ must be the direct sum of compact simple Lie subalgebras and its center~\cite{wiersema2309classification, knapp1996lie}, ignoring  the center also leads to a compact, connected dynamical subgroup. Thus if $\G$ is not closed, this will be the compact dynamical subgroup we consider. Note that it is harmless to ignore the center, since the component of  the observable in the center of $\g$ does not evolve (in a Heisenberg sense) anyways.}

The variance of the gradient, under Haar initialization, relies on the so-called ``second-moment operator'':
\begin{align}\label{eq:twirl}
   \mathcal{T}: \A \mapsto \int_{G}(\U_{g}\otimes \U_{g})\A(\U_{g}^{\dagger} \otimes \U_{g}^{\dagger})\,dg,
\end{align}
which orthogonally projects onto the set of commuting operators (i.e., commutant) of $\{\U_g \otimes \U_g :  \forall g \in G\}$. {Commutation implies that $\forall \A \in\mathcal{M}(\mathbb{C}, 2^{2n}), \mathcal{T}(\A)$ }must respect the  decomposition of $V^{\otimes 2}$ into irreducible components (invariant subspaces). If $V^{\otimes 2}$ has the following decomposition into irreducible components (not grouping by multiplicity)
\begin{align}
    V^{\otimes 2} = \bigoplus_{\lambda} V_{\lambda},
\end{align}
then
\begin{align}\label{eq:moment_decomp}
    \int_{G}(\U_{g}\otimes \U_{g}) \A (\U_{g}^{\dagger} \otimes \U_{g}^{\dagger})\, dg =\sum_{\lambda}\frac{\Tr[\A P_{\lambda}]}{\dim V_{\lambda}}P_{\lambda},
\end{align}
for orthogonal projectors $P_{\lambda}$ onto $V_{\lambda}$. This projection can also be expressed in terms of the well-known Weingarten function~\cite{Collins_2006, collins2021weingarten}. Notice that the Lie algebra appears to play no role in this discussion. In addition, the inverse scaling of dimension of each $V_{\lambda}$ is apparent. Furthermore, while a general theory of such integrals exists~\cite{diez2022expectation}, they are quite challenging to tackle in practice. Most results in quantum information restrict to the case where $G = \text{SU}(2^n)$, where the commutant is easy to characterize. {Specifically, this leads to the well-known result that approximate $2$-designs for $\text{SU}(2^n)$ have BPs~\cite{mcclean2018barren, Larocca2022diagnosingbarren}.}

Fortunately, the integrals appearing in the theory of VQAs turn out to have substantial simplifications, which furnishes the connection to the dimension of the DLA in certain settings. Our results shed much needed light on this apparently unintuitive phenomenon observed in practice. The first key insight is that $\A$ is always a tensor product of two operators, i.e. if $\O$ is the observable in the quantum circuit, then we get second-moment integrals with $\A = i\O \otimes i\O$, that is,
\begin{align}
\label{eqn:adjoint_int}
    &\int_{G}(\U_{g} i\O \U_{g}^{\dagger}) \otimes (\U_{g} i\O \U_{g}^{\dagger})\, dg \\
    &= \int_{G}\text{Ad}_{g}(i\O) \otimes \text{Ad}_{g}(i\O)\, dg ,
\end{align}
where the relation to the well-known adjoint representation, of $G$, i.e. $\text{Ad}_{g}(i\O) = \U_gi\O \U_{g}^{\dagger}$, is apparent when $i\O$ lies in $\g$. This simple observation is critical in enabling concise expressions for the variance of the gradient%
{,  revealing} the inverse dependence on the dimension of the DLA. {Specifically}, given that the dimension of the adjoint representation is $d_{\g}$ the reason for the scaling becomes more plausible. 

{Note that to connect back to \eqref{eq:moment_decomp}, this can also be viewed as a projection of the subspace \begin{align}
S :=\text{span}_{\mathbb{C}}\{ i\O \otimes i\O : i\O \in d\phi(\g)\} \subset \mathcal{M}(\mathbb{C}, 2^{2n})
\end{align} onto the commutant via an operator called the Casimir.}

{The (split quadratic) \emph{ Casimir operator}, $\K$, for representation $\phi$ is defined as:
\begin{equation} \label{eq:casimir}
    \K = I_\phi^{-1} \sum_i \E_i \otimes \E_i,
\end{equation}
where $\{e_i\}$ is an orthonormal basis under the standard norm for $\g$ and $\E_i = d\phi(e_i)$.
We can also use the Casimir to define an orthogonal projector, $P_{\g}$, from  the space of skew-Hermitian operators on $V$, i.e. $\mathfrak{u}(V)$, onto the subspace $d\phi(\g)$, which is useful when we are dealing with objects not completely supported on the Lie algebra:
\begin{align} \label{eq:projector}
\X_{\g} &:= P_\g \X = -\Tr_1((\X\otimes\1)\K) \\
&= -I_\phi^{-1} \sum_i \Tr(\X\E_i)\E_i,\\
\|\X_{\g}\|_{\text{F}}^2 &= -\Tr((\X \otimes \X)\K) = -I_\phi^{-1} \sum_i \Tr^2(\X\E_i)
\end{align}
where $\X \in \mathfrak{u}(V)$ and $\Tr_1$ is the partial trace over the first subspace. One can check that as expected $P_\g d\phi(a) = d\phi(a)$.}

\subsection{Proof of Theorem \ref{thm:simple_group}}

{The following Lemma is fundamental to our main theorem, it may also be  of independent interest. 
The proof can be found in the \emph{Supplementary Information}.
\begin{lemma}
\label{lemma:main}
  Let $G$ be a compact simple Lie group with Lie algebra $\g$.  Suppose $V$ is a finite-dimensional inner product space, $\phi : G \rightarrow \mathcal{U}(V)$ is a unitary representation of $G$, and $\U_g = \phi(g)$. In addition, $a \in \g$, $\mathbf{A} = d\phi(a)$.
Then the following holds:
    \begin{equation}
        \int_G (\U_g \mathbf{A} \U_g^{\dagger})^{\otimes 2}\; dg %
        = \frac{\|\A\|^2_{\text{F}}}{d_{\g}} \K.
    \end{equation}
\end{lemma}
From Lemma~\ref{lemma:main}, it can be seen that the commutant is the one-dimensional subspace spanned by the Casimir operator, i.e.
\begin{equation}
    \mathcal{T}(R) = \frac{\Tr(R^{\dagger}\K)}{d_\g}\K\;\;\; \forall R \in \mathcal{S}.
\end{equation}
We are also going to frequently use the following identity. Let $a \in \g$ and $\A := d\phi(a)$. Also let $\E_i := d\phi(e_i)$ be a basis for the Lie algebra orthonormal under the standard norm. Then
\begin{equation}\label{eq:proj_norm}
    \|\A\|_{\text{F}}^2 = -I_\phi^{-1}\sum_i \Tr^2(\A\E_i),
\end{equation}
which is important as when working with a quantum circuit one often has access to the representation basis $\{\E_i\}$ but not directly to $\{e_i\}$, so it is a convenient shortcut to calculate $\|a\|_\g^2$.}

\begin{proof}[Proof of Theorem \ref{thm:simple_group}]
{As was shown in the \emph{Results} section, we can assume $\textup{GradVar} = \mathbb{E}_{{g^{+}, g^{-}}\sim \mu^{\otimes 2}}[(\partial\<\O\>_{\rho})^{2}]$.}
Let us write the integral for the second moment in full, and rearrange terms appropriately:

\begin{align}
&\mathbb{E}_{{g^{+}, g^{-}}\sim \mu^{\otimes 2}}[(\partial\<\O\>_{\rho})^{2}]\\&=\iint_{G} (\text{Tr}(\U_{g^{-}}i\rho \U_{g^{-}}^{\dagger}[\H, \U_{g^{+}}i\O \U_{g^{+}}^{\dagger}]))^{2} \, dg^+ dg^- \\
&= \iint_G \Tr\bigg\{(i\rho)^{\otimes 2} \; \U_{g^-}^{\otimes 2} ([\H, \U_{g^+} i\O \U_{g^+}^{\dagger}] \nonumber\\ &\otimes [\H, \U_{g^+} i\O \U_{g^+}^{\dagger}]) \U_{g^{-}}^{\dagger\otimes 2} \bigg\} \; dg^+ dg^-.
\end{align}

{Suppose $\X_{+} := \int_{G}(\U_{g^+}i\O \U_{g^+}^{\dagger})^{\otimes 2}\,dg^+$.} Let us ignore the trace and $\rho$, and expand out the commutators:
\begin{align}
&\iint_G  \U_{g^{-}}^{\otimes 2} (\H\U_{g^+} i\O \U_{g^{+}}^{\dagger}-\U_{g^+} i\O \U_{g^+}^{\dagger}\H) \nonumber\\&\otimes (\H\U_{g^+} i\O \U_{g^+}^{\dagger}-\U_{g^+} i\O \U_{g^+}^{\dagger}\H) \U_{g^-}^{\dagger\otimes 2} \; dg^+ dg^-\\
&=\int_{G}  \U_g^{-\otimes 2}\H^{\otimes 2}\X_{+} \U_{g^-}^{\dagger\otimes 2} \;  dg^-\\
&+ \int_{G}  \U_{g^-}^{\otimes 2}\X_{+}\H^{\otimes 2} \U_{g^-}^{\dagger\otimes 2} \; dg^-\\
&- \int_{G}  \U_{g^-}^{\otimes 2}(\H\otimes \1)\X_{+}(\1 \otimes \H) \U_{g^-}^{\dagger\otimes 2} \; dg^-\\
&- \int_G  \U_g^{-\otimes 2}(\1\otimes \H)\X_{+}(\H\otimes \1) \U_{g^-}^{\dagger\otimes 2} \; dg^-.
\end{align}

We end up with four similar terms. Starting with the common inner integral, since $G$ is compact, we can apply Lemma~\ref{lemma:main} and write 
\begin{align}
\X_{+} = \int_{G}(\U_{g^+} i\O \U_{g^+}^{\dagger})^{\otimes 2}\,dg^+ = \frac{\|\O\|_{\text{F}}^{2}}{d_{\g}} \K.
\end{align}
We can plug this expression back into the earlier expression without the trace and $\rho$, and rearranging terms and using $\K := I_\phi^{-1} \sum_i \E_i\otimes\E_i$ gives:
\begin{equation}
    \frac{\|\O\|_{\text{F}}^{2}}{I_\phi d_{\g}} \sum_{k=1}^{d_{\g}}\int_{G} \U_{g}[\H, \E_k]\U_{g}^{\dagger} \otimes \U_{g}[\H, \E_k]\U_{g}^{\dagger}\,dg.
\end{equation}
Now applying the Lemma again, noting that $\H = \sum_{q}h_{q}\E_{q}$, we have:
\begin{align}
&\frac{\|\O\|_{\text{F}}^{2}}{I_\phi d_{\g}^2} \sum_{j,k=1}^{d_{\g}} \|[\H, \E_k]\|_{\text{F}}^2 \K \\
&=\frac{\|\O\|_{\text{F}}^{2}}{I_\phi d_{\g}^2} \sum_{j,k=1}^{d_{\g}}\frac{\Tr([\H, \E_k]\E_{j})\Tr([\H, \E_k]\E_{j})}{I_\phi}\K\\
&=\frac{\|\O\|_{\text{F}}^{2}}{d_{\g}^2} \sum_{q,r, j,k=1}^{d_{\g}} h_{q} h_{r} \frac{\Tr([\E_{q}, \E_k]\E_{j})\Tr([\E_{r}, \E_k]\E_{j})}{I_\phi^2} \K\\
&=\frac{\|\O\|_{\text{F}}^{2}}{d_{\g}^2} \sum_{q,r, j,k=1}^{d_{\g}}h_{q} h_{r} f_{qk}^j f_{rk}^j \K\\
&=\frac{\|\O\|_{\text{F}}^{2}}{d_{\g}^2} \sum_{q,r=1}^{d_{\g}} h_{q} h_{r} \left(-\sum_{j,k=1}^{d_{\g}} f_{qk}^j f_{rj}^k \right) \K\\
&=\frac{\|\O\|_{\text{F}}^{2}}{d_{\g}^2} \sum_{q,r=1}^{d_{\g}} h_{q} h_{r} \left(-g_{qr} \right) \K\\
&=\frac{\|\O\|_{\text{F}}^{2} \|\H\|_{\text{K}}^{2}}{d_{\g}^{2}} \K,
\end{align}
where we have used anti-symmetry of the commutator braket to reveal that  the inner sum is the Killing form (since $\mathfrak{g}$ is a compact simple Lie algebra, the negative of the Killing form is a valid inner product). Note that $f_{qk}^j = \Tr([\E_q, \E_k]\E_j)$ are the structure constants.%

Now, we can reintroduce the trace and $\rho$ to get:
\begin{align}
&\mathbb{E}_{{g^{+}, g^{-}}\sim \mu^{\otimes 2}}[(\partial\<\O\>_{\rho})^{2}] \\
&=\frac{\|\H\|_{\text{K}}^{2}\|\O\|_{\text{F}}^{2} }{d_{\g}^{2}} \Tr((i\rho)^{\otimes 2} \K)\\
&=\frac{\|\H\|_{\text{K}}^{2}\|\O\|_{\text{F}}^{2}  \|\rho_{\g}\|_{\text{F}}^2}{d_{\g}^{2}}.
\end{align}
\end{proof}

\subsection{Proof of Theorem \ref{thm:compact}}
{The following is a generalization of Lemma \ref{lemma:main} to outside the simple group setting. 
The proof can be found in the \emph{Supplementary Information}.}
\begin{lemma}\label{lemma:main_compact}
\label{lem:compact_gen_int}
Let $G$ be a compact and connected Lie group with Lie algebra $\g$.  Suppose $V$ is a finite-dimensional inner product space, $\phi : G \rightarrow \mathcal{U}(V)$ is a unitary representation of $G$, and $\U_g = \phi(g)$. In addition, $a \in \g$, $\mathbf{A} = d\phi(a)$.
Then the following holds:
\begin{align}
\int_{G}(\U_{g}\A\U_{g}^{\dagger})^{\otimes 2} dg %
= \sum_{\alpha}\frac{\|\A_{\g_{\alpha}}\|^2_{\text{F}}}{d_{\g_{\alpha}}} \K_{\g_{\alpha}} + \A_{\mathfrak{c}}^{\otimes 2},
\end{align}
where %
$\A_{\g_{\alpha}}$ is the image of the component of $a$ in $\g_\alpha$ under $d\phi$. Likewise, $\K_{\g_{\alpha}}$ is the Casimir in the subalgebra $\g_\alpha$.
\end{lemma}

The above result implies that we expect contributions to the variance from the various subalgebras. {Indeed, the final expression} for the variance is remarkably simple, since all the cross terms between different subalgebras vanish, and the abelian subalgebras do not contribute.

\begin{proof}[Proof of Theorem \ref{thm:compact}]
The proof largely follows the strategy of that for simple groups. Define the shorthand $i\O_{\g_{\alpha}} := P_{\mathfrak{g}_{\alpha}}i\O$ and $i\O_{\mathfrak{c}} :=  P_{\mathfrak{c}}i\O$. Like before, we expand the commutator but this time use Lemma~\ref{lemma:main_compact}:
\begin{equation}
\int_{G}(\U_{g^{+}}i\O\U_{g^{+}}^{\dagger})^{\otimes 2} dg^{+} = \sum_{\alpha}\frac{\|\O_{\g_{\alpha}}\|^2_{\text{F}}}{d_{\g_{\alpha}}} \K_{\g_{\alpha}} + \O_{\mathfrak{c}}^{\otimes 2}.
\end{equation}
Now, after applying the commutator and taking the integral over $\U_{g^{-}}$ %
, we find the result is still a summation over $\alpha$ only. This is because, since the subalgebras are ideals, if $\E_k \in d\phi(\g_\alpha)$ then $[\H, \E_k]\in d\phi(\g_\alpha)$, and therefore $\|P_{\g_\beta}[\H, \E_k]\|_{\text{F}} = 0$ if $\beta \neq \alpha$. Thus the cross terms vanish. The contribution from the center also vanishes upon taking the commutator. Thus the result follows.
\end{proof}

\subsection{Proof of Theorem \ref{thm:compound_basis_state}}

Using the identities in Section~\ref{sec:review} we can get forms of the theorems that are practically useful. For example, in the simple group case,
\begin{equation}\label{eq:var_alter}
\textup{GradVar}
=\frac{I_{\text{Ad}} \lVert o \rVert_{\g}^{2}\lVert h\rVert_{\g}^{2}}{d_{\g}^{2}} \sum_i \Tr^2(i\rho \E_i),
\end{equation}
where $\E_{i} = d\phi(e_i)$ for orthonormal basis $\{e_i\}$ for $\g$. This turns out to be the most useful form of the result for the examples below because we will have explicit knowledge of the representation $\phi$. In addition, the representation index, $I_{\phi}$, drops out. 

\begin{proof}[Proof of Theorem \ref{thm:compound_basis_state}] For $\mathfrak{su}(n)$, $d_\g=n^2 - 1$ and the Dynkin index of the adjoint representation is $I_\Ad = 2n$.
Now we work out the state's projected norm.
Choose $\rho$ to be a computational basis state, where it can be shown that $\rho \otimes \rho$ lies in an irreducible subrepresentation of the tensor product representation $\phi \otimes \phi$ (see \emph{Supplementary Information}). Then we only need to focus on the simultaneously diagonal elements of the Lie algebra, that is, the Cartan subalgebra $\mathfrak{h}$. To calculate the Casimir eigenvalue we need to find an orthogonal basis  $\mathscr{H}$ for $\mathfrak{h}$, which cannot be $\{h^{ij}_z\}_{i\neq j}$ since the elements are not linearly independent. 

We can construct a suitable basis for $\mathfrak{h}$ using the formula
\begin{align}
\mathscr{H} &= \frac{i}{2} \bigcup_{m=1}^{n-1} \frac{1}{\sqrt{m(m + 1)}} \left\{m\sigma^z_{m+1} - \sum_{i=1}^m \sigma^z_j \right\} \\
&= \frac{i}{4}\{\sqrt{2}(\sigma^z_2 - \sigma^z_1), \frac{\sqrt{2}}{\sqrt{3}}(2\sigma^z_3 - \sigma^z_2 - \sigma^z_1), \nonumber\\& \frac{1}{\sqrt{3}}(3\sigma^z_4 - \sigma^z_3 - \sigma^z_2 - \sigma^z_1), ...\}
\end{align}
even though this is expressed more cleanly with Pauli $z$s, each element can be obtained as a linear combination of the $\{h^{ij}_z\}$ generators. One can check that the elements are all orthogonal and the norm of their pullback on $\g$ is 1, and the resulting subalgebra has the correct dimension: $\dim \mathfrak{h} = \text{rank}\; \mathfrak{su}(n) = n-1$.
With this, one can explicitly calculate the diagonal part of $I_\phi\K$ for any $n$,
\begin{equation}\label{eq:diag_1}
    I_\phi\text{diag}(\K) = \sum_{\H_i \in \mathscr{H}} \H_i \otimes \H_i
\end{equation}
however the calculation is unwieldy. Fortunately, we can directly infer the final form from symmetry arguments, since by inspection: $\text{diag}(\K)$ is composed of sums of tensor products of two Pauli $z$s, it is symmetric around the tensor product, and furthermore since $\text{SWAP}_{ij}\otimes \text{SWAP}_{ij} \in \phi(G) \otimes \phi(G)$ it must be invariant upon any simultaneous permutation of the qubit indices on the subspaces. Thus,
\begin{equation}\label{eq:diag_2}
    I_\phi\text{diag}(\K) = A \sum_{i=1}^n \sigma^z_i \otimes \sigma^z_i + B \sum_{i\neq j} \sigma^z_i \otimes \sigma^z_j.
\end{equation}
To find the value of $A$, evaluate $\text{diag}(\K)$ on the state $|\Psi\> = |+...+0\>^{\otimes 2}$ using Eqs.~\eqref{eq:diag_1} and \eqref{eq:diag_2}:
\begin{align}
&I_\phi\<\Psi|\text{diag}(\K)|\Psi\> = -\frac{1}{4n(n-1)} (n-1)^2 \<\Psi|\sigma^z_n \otimes \sigma^z_n|\Psi\> \\&= A \<\Psi|\sigma^z_n \otimes \sigma^z_n|\Psi\> \implies A = -\frac{n-1}{4n},
\end{align}
and for $B$, on $|\Psi'\> = |+...+0\> \otimes |+...+0+\>$:
\begin{align}
&I_\phi\<\Psi'|\text{diag}(\K)|\Psi'\> \\&= \frac{1}{4n(n-1)} (n-1)\<\Psi'|(\sigma^z_n \otimes \sigma^z_{n-1})|\Psi'\> \\&= B \<\Psi'|\sigma^z_n \otimes \sigma^z_{n-1}|\Psi'\> \implies B = \frac{1}{4n}.
\end{align}
Now we use this to evaluate the expectation value of $\K$ on a computational basis state of Hamming weight $k$. The first summation in Eq.~\eqref{eq:diag_2} will be constant and equal to $n$, while the second summation will be equal to the number of distinct bits of equal value minus those of different value, $k(k-1) + (n-k)(n-k-1) - 2k(n-k) = (n-2k)^2 - n$. So overall
\begin{align}
&I_\phi\|\rho_{\g}\|_{\text{F}}^2 = \sum_{\H_i \in \mathscr{H}} \Tr^2(i\rho \H_i) \\&=  \frac{n-1}{4} - \frac{(n-2k)^2 - n}{4n} = \frac{k(n-k)}{n}.
\end{align}
Choosing $\O = i h_z^{12}$ and $\H$ any generator, $\|o\|_\g^2 = 1/2 = \|h\|_\g^2$, and the final result is 
\begin{align}
&\textup{GradVar} = \frac{2n(1/2)^2}{(n^2-1)^2}\frac{k(n-k)}{n} \\&= \frac{k(n-k)}{2(n^2-1)^2} \in \Omega\left(\frac{1}{n^3}\right).
\end{align}
\end{proof}

\subsection{Proof of Theorem \ref{thm:compound_uniform_sup}}
\begin{proof}[Proof of Theorem \ref{thm:compound_uniform_sup}] For the uniform superposition of computational basis states, $|\psi\> = |+\>^{\otimes n}$, then $\forall i,j,  \<\psi|h_{ij}^y|\psi\> = \<\psi|h_{ij}^z|\psi\> = 0$. The only nonzero terms involve the Pauli-$x$ type generators.
We can form the corresponding orthogonal generators normalized in $\g$ by $\H_{ij}^x = \sqrt{2} h_{ij}^x$. However, even though there are $\binom{n}{2}$, only the $n-1$ with $j=i+1$ do not annihilate on $|\psi\>$ since the others have $\sigma_z$'s in their definition. For these generators, $\<\psi|\H_{ij}^x|\psi\> = -\frac{i}{2\sqrt{2}}$, giving
\begin{equation}
    I_\phi \|P_\g \rho\|_{\text{F}}^2 = -\sum_{i=1}^{n-1} \lvert\<\psi|\H_{i(i+1)}^x|\psi\>\rvert^2 = \frac{1}{8}(n-1),
\end{equation}
and so
\begin{equation}
     \text{GradVar} = \frac{2n(1/2)^2}{(n^2-1)^2}\frac{(n-1)}{8} = \frac{n(n-1)}{16(n^2-1)^2} \in \Theta\left(\frac{1}{n^2}\right).
\end{equation}
\end{proof}

\subsection{Proof of Corollary \ref{cor:compound_non_lasa_single}}
\begin{proof}[Proof of Corollary \ref{cor:compound_non_lasa_single}]
We expand the variance term for the computational basis state case:
\begin{align}
&\frac{2k(n-k)}{(n^2-1)^2} = \text{Var}_{(g_{+}, g_{-}) \sim \mu^{\otimes 2}}[\partial\langle \sigma_i^{z} - \sigma_j^{z} \rangle]
    \\&=\text{Var}_{(g_{+}, g_{-}) \sim \mu^{\otimes 2}}[\partial\langle \sigma_i^{z}  \rangle]
    + \text{Var}_{(g_{+}, g_{-}) \sim \mu^{\otimes 2}}[\partial\langle \sigma_j^{z} \rangle]\\
    &- 2\text{Cov}_{(g_{+}, g_{-}) \sim \mu^{\otimes 2}}[\partial\langle \sigma_i^{z} \rangle , \partial\langle \sigma_j^{z} \rangle] .
\end{align}
Note since a permutation swapping qubit $i$ with $j$ is a valid compound  $\text{SU}$ matrix, we have that $\partial\langle \sigma_i^{z}  \rangle$ and $\partial\langle \sigma_j^{z}  \rangle$ are identically distributed. Thus,
\begin{align}
&\frac{2k(n-k)}{(n^2-1)^2} = 2\text{Var}_{(g_{+}, g_{-}) \sim \mu^{\otimes 2}}[\partial\langle \sigma_i^{z}  \rangle] \nonumber\\&- 2\text{Cov}_{(g_{+}, g_{-}) \sim \mu^{\otimes 2}}[\partial\langle \sigma_i^{z} \rangle , \partial\langle \sigma_j^{z} \rangle].
\end{align}
Due to the above equality and Cauchy--Schwarz, i.e. $\text{Var}_{(g_{+}, g_{-}) \sim \mu^{\otimes 2}}[\partial\langle \sigma_i^{z}  \rangle] \geq |\text{Cov}_{(g_{+}, g_{-}) \sim \mu^{\otimes 2}}[\partial\langle \sigma_i^{z} \rangle , \partial\langle \sigma_j^{z} \rangle]|$ (recall the variances are equal), we can conclude that $\text{Var}_{(g_{+}, g_{-}) \sim \mu^{\otimes 2}}[\partial\langle \sigma_i^{z}  \rangle]$ must only be polynomially vanishing in $n$, which implies no BP for any $k$ and any single qubit $\sigma_z$ measurement. A similar result can be shown to hold for the uniform superposition state.
\end{proof}

\section*{Ackowledgements}
We thank Iordanis Kerenidis for early discussions on BPs in quantum compound ans\"{a}tze, and Aram Harrow for helpful discussions and feedback on the manuscript. We thank Marco Cerezo and Martin Larocca for discussions on the basics of equivariant QNNs and the role of the DLA. We also thank the members of Global Technology Applied Research at JPMorgan Chase for comments and feedback throughout the project. 

\section*{Data Availability}
We make all the data presented in this paper available online at \url{https://doi.org/10.5281/zenodo.10720106}.

\section*{Code Availability}
We make the code required to reproduce the figures presented in this paper available online at \url{https://doi.org/10.5281/zenodo.10720106}.

\bibliography{main}

\appendix

\onecolumngrid
\section{Introduction to Lie Groups and Representation Theory}
\label{sec:intro_to_lie}
This section presents a short intro to the representation theory of  Lie groups and provides sufficient background understanding the results of the paper and their proofs. The only prerequisite is knowledge of some concepts from algebra and topology. For a more detailed introduction the reader is directed to any one of the following fabulous texts \cite{knapp1996lie, hall2013lie, fulton2013representation, goodman2009symmetry, humphreys2012introduction}.

\subsection{Lie Groups, Lie Algebras, Representations}

A \emph{Lie group}, $G$, is a topological group that is also a smooth manifold. The  \emph{Lie algebra}, $\g$, associated with $G$ is the tangent space at the group's identity element and forms a non-associative algebra with the  Lie bracket operation, denoted $[\cdot, \cdot]$. Specifically, the Lie bracket obeys the Jacobi identity: $\forall h, k, j \in \g,$
\begin{align}
 [h, [k, j]] + [k, [j, h]] + [j, [h, k]] = 0,
\end{align}
and is additionally bilinear and skew symmetric.

Our focus will be \emph{compact Lie groups} (or subgroups of compact Lie groups), which are Lie groups whose topology is compact. In the compact setting, it is without loss of generality to restrict our attention to groups of matrices, called matrix Lie groups (see the Peter-Weyl theorem \cite[Corollary 4.22]{knapp1996lie}). This restriction also has the benefit of simplifying some of the more abstract notions mentioned above. First, one can now view the Lie bracket as the matrix commutator. Second, the matrix exponential maps elements of $\g$ onto analytic curves in neighborhood of the identity, i.e. $\forall t \in \mathbb{R}, \gamma(t) = e^{t\cdot \H} \in G$ with $\H \in \g$ being the tangent vector. Third, \emph{compact Lie algebras}, those associated with a compact Lie group, can be assumed to only contain skew-Hermitian matrices. If the topology of $G$ is both compact and connected, then the matrix exponential from $\g$ to $G$ is surjective, i.e. for any $g \in G, \exists h \in \g$ s.t. $g = e^{h}$.

Since we will be only considering compact Lie groups, we can assume the existence of a finite Haar measure \cite{tao2014hilbert, mele2023introduction}, enabling integration over the whole group. This a uniform measure, $\mu$, that is invariant under left and right translation by group elements, i.e. for some fixed $\g$, $\mu(gh) = \mu(hg), \forall h \in G$.  This enables one to integrate over the group and compute various stastical moments. Throughout the paper, all integration, e.g. $\int_{G} f(g) dg$, is with respect to the Haar measure

A \emph{simple Lie algebra} is one that has dimension greater than one and no non-trivial ideals and is \emph{semi-simple} if $\g$ can be decomposed into a direct sum of simple Lie algebras. Note that  (semi-)simple Lie groups have (semi-)simple Lie algebras. 

Let $V$ denote a real or complex, finite-dimensional inner product space, $\mathcal{U}(V)$ denote the group of isometries on $V$, and $\mathfrak{u}(V)$ denote the algebra of skew-Hermitian operators on $V$. Without loss of generality one could take of course this to be $\mathbb{C}^n$ for some $n$, up to isomorphism.  A unitary \emph{Lie group representation} of $G$ is the following smooth homomorphism $\phi : G \rightarrow \mathcal{U}(V)$, and $\phi$ is called \emph{faithful} if it is injective. The representation makes the space $V$ into a $G$-module. 

A \emph{subrepresentation} of $\phi$ is an invariant subspace of $V$, and a representation is \emph{irreducible} if it has no non-trivial subrepresentations.  Compact Lie groups are guaranteed to have finite-dimensional unitary representations, which actually implies that any representation is \emph{reducible}, i.e. can be decomposed as a direct sum of irreducible subrepresentations (also called irreducible components). More specifically, an the vector space on which $G$ is acting can be equipped with an inner product that makes the action unitary. In addition, representations of simple Lie algebras are always faithful or trivial. We will frequently use the notation $\U_{g}$ to denote the element $\phi(g) \in \mathcal{U}(V)$ for some $g \in G$ when the representation $\phi$ and $V$ are clear from the context.

A map $f$ between representations ${\phi}$ and ${\psi}$ is called \emph{equivariant} if $\forall g \in G, f\circ \phi(g) = \psi(g) \circ f$, and if $f$ is bijective, then the representations are \emph{isomorphic as $G$-modules}. \emph{Schur's lemma} states that any equivariant map between irreducible representations is either an isomorphism or the zero map. Furthermore, if  the representation space is complex, then $f$ must be a multiple of the identity and any two equivariant maps are scalar multiplies of one another.

The group representation induces a \emph{Lie algebra representation} $d\phi : \g \rightarrow \mathfrak{u}(V)$, which is the differential of the smooth map $\phi$.  %
Just as group representations respect the group operation (i.e. are homomorphims), Lie algebra representations respect addition and the Lie bracket (the commutator):
\begin{equation}
\label{eqn:commutator}
    [d\phi(x), d\phi(y)] = d\phi([x, y]),
\end{equation}
specifically it is an algebra homomorphism.

Given an orthornomal basis $\{ v_j \}_{j}$ for $V$, for $i, j$, the associated \emph{matrix coefficient} is defined as $\phi_{i,j}(g) = \langle \phi(g)v_i, v_j\rangle$. For example, the matrix coefficients of the adjoint representation are the structure constants with respect to the basis $\{ v_j \}_{j}$ for $\g$. \emph{Schur orthogonality} states that for non-isomorphic irreducible representations $\phi$ and $\psi$: 
\begin{align}
    \int_{G}\phi_{ij}(g)\overline{\psi_{kl}(g)}dg = 0,
\end{align} and for the same representation with orthonormal basis $\{ v_j \}_{j}$ for $V_\phi$:   
\begin{align}
    \int_{G}\phi_{ij}(g)\overline{\phi_{kl}(g)}dg = \frac{\delta_{ik}\delta_{jl}}{\dim V_{\phi}}.
\end{align}

\subsection{Useful Representations and Norms}

There are a few types of representations that we will refer to frequently. If a compact Lie group $G$ contains $n \times n$ unitary matrices,  its \emph{standard representation} is the natural action on $V = \mathbb{C}^n$. The Lie algebra standard representation follows similarly but replacing unitary with skew-Hermitian. The \emph{adjoint representation} consists of $G$ acting on its Lie algebra by conjugation, i.e. $\phi$ is defined by $\phi(g)h =ghg^{-1}, $ $\forall h \in \g$, and $V_{\phi} = \g$. The associated Lie algebra adjoint representation is defined by $d\phi(h)k = [h , k]$  $\forall k \in \g$, and  $V_\phi = \g$. Lastly, the tensor power of a Lie group representation $\phi$, denoted by $\phi \otimes \phi$, is defined as $\forall g \in G, (\phi \otimes \phi)(g) = \phi(g)\otimes \phi(g)$ and acts on $V^{\otimes 2}$. The tensor power of the associated Lie algebra representation, denoted $d\phi \otimes d\phi$, is defined as $\forall h \in \g, (d\phi \otimes d\phi)(h) = d\phi(h) \otimes \1 + \1 \otimes d\phi(h)$ for identity operator $\1$.

Let $\Tr(\cdot)$ denote the standard trace for linear operators. Given a  Lie algebra $\g$, we can set an orthonormal basis for the Lie algebra $\{e_i\}$, where $e_i$ are associated with the standard representation of $\g$, using the \emph{standard trace form}: %
\begin{equation}
    -\Tr(e_ie_j) = \delta_{ij},
\end{equation}
where the negative is to ensure that the form is positive definite. Each representation also has an associated trace form defined in a similar manner, i.e. by multiplying and taking traces of the matrices involved. Specifically, the \emph{Killing form} is the trace form associated with the adjoint representation.

Note that the Lie algebra representation does not, necessarily, preserve these forms, meaning that $d\phi(e_i)$ is not normalized, w.r.t. the trace form for $\phi$, in general. Still, it is a well-known result that for simple Lie algebras any trace form is a scalar multiple of the (nondegenerate) Killing form. 
Furthermore, for compact Lie algebras the trace forms are always real and definite. Thus every two trace forms are a \emph{real} scalar multiple of each other. Thus, for any representation $\phi$ of a compact simple Lie group, we define a scaling constant $I_\phi$ that we call the \textit{index of the representation} (w.r.t. the standard representation) such that:
\begin{equation}
    -\Tr (d\phi(e_i)d\phi(e_j)) = I_\phi \delta_{ij}.
\end{equation}
This is the same as (twice) the Dynkin index for irreducible representations \cite{fuchs1995affine}. We will omit the subscript in the trace when it is obvious which space it is taken in. 
The above discussion also implies that for compact simple Lie algebra, the trace forms induce valid inner produces and norms \cite{knapp1996lie}.

The first norm, which we call the \emph{standard norm}, is one induced by the standard trace form. For any $a \in \g$,
\begin{equation}
    \|a\|_\g^2 = -\Tr(a^2).
\end{equation} 
We call the norm induced by the Killing form, the \emph{Killing norm} and denote it by $\lVert \cdot \rVert_{\text{K}}$. In general, if we are working in a representation $\phi$ then we can define the Frobenius norm of $d\phi(a), a \in \g$:
\begin{equation}\label{eq:frobenius}
    \|d\phi(a)\|_{\text{F}}^2 = -\Tr (d\phi(a)^2).
\end{equation}
All norms are related via the representation index
\begin{align}
&\lVert d\phi(a) \rVert_{\text{F}} = I_{\phi}\lVert a \rVert_{\g}^2\\
&\lVert d\phi(a) \rVert_{\text{K}}^{2} = \lVert a \rVert_{\text{K}}^{2}=I_{\text{Ad}}\lVert a \rVert_{\g}^2 = \frac{I_{\text{Ad}}}{I_{\phi}}\lVert d\phi(a) \rVert_{\text{F}}^2.
\end{align}
Note that the compact Lie algebras are completely classified, and $I_{\text{Ad}}$ is  $\Theta(\sqrt{d_{\g}})$ for all non-exceptional classes (i.e. the dual Coxeter numbers) \cite{habereigenvalues}. In addition, the equality $\lVert d\phi(a) \rVert_{\text{K}}^{2} = \lVert a \rVert_{\text{K}}^{2}$ follows from the faithfulness of representations of simple Lie algebras.

\subsection{Casimir Operators}

As mentioned earlier, the commutator bracket is non-associative product. However, a Lie algebra can be embedded in a larger algebra is associative and equiped with an additional product. This algebra, denoted by $\mathcal{U}(\g)$ is called the \emph{universal enveloping  algebra} of $\g$. The embedding map $i :\g \rightarrow \mathcal{U}(\g)$ satisfies the additional important properly: $\forall h, k \in \g$
\begin{align}
    [i(h) , i(k)] = i(h)i(k) - i(k)i(h),
\end{align}
where juxtaposition represents the additional associative product of $\mathcal{U}(\g)$. Thus, $\mathcal{U}(\g)$ is an associative algebra where the Lie bracket acts as the commutator. We will from now on drop the map $i$ and just use $\cdot$ denote this new associative product.

The algebra $\mathcal{U}(\g)$ is formed by quotienting the tensor algebra of $\g$ so that above desired relation is satisfied. In addition, representations of $\mathcal{U}(\g)$ and $\g$ are one-to-one with each other and respect the new product when extended. 

There is an important element that lies in the center of $\mathcal{U}(\g)$ called the \emph{quadratic Casimir operator} defined as:
\begin{align}
    c = \sum_{j} e_j \cdot e_j,
\end{align}
where $\{e_j\}$ forms a basis for $\g$. With respect to algeba representations, the quadratic Casimir maps to:
\begin{align}
\label{eqn:quadr_casimir_defn}
    \boldsymbol{C}_{\phi} = \sum_{j} d\phi(e_j) \cdot d\phi(e_j),
\end{align}
and under the tensor power of a representation it maps to
\begin{align}
(d\phi \otimes d\phi)(c) &= \sum_{j} (d\phi(e_j) \otimes \1 + \1 \otimes d\phi(e_j)) \cdot (d\phi(e_j) \otimes \1 + \1 \otimes d\phi(e_j)) \\
&= 2\sum_{j} d\phi(e_j) \otimes d\phi(e_j) + (\boldsymbol{C}_{\phi} \otimes \1+ \1 \otimes \boldsymbol{C}_{\phi}).
\end{align}
The first term in the second equality is called the \emph{split quadratic Casimir operator}:
\begin{align}
\label{eqn:split_casimir_defn}
    \K_{\phi} = \sum_{j} d\phi(e_j) \otimes d\phi(e_j).
\end{align}

One may note that Supplementary Equation~\eqref{eqn:quadr_casimir_defn} and \eqref{eqn:split_casimir_defn} have a similar appearance. In fact, as mentioned earlier, the associative product  on $\mathcal{U}(\g)$ was obtained by quotienting the tensor algebra (i.e. $\otimes \mapsto \cdot$). Thus, in the tensor algebra, before quotienting, these two operators are formally identical. In fact, even Ref.~\cite{ragone2023unified} called what we refer to as the split Casimir the quadratic Casimir. 

However, traditionally Supplementary Equation~\eqref{eqn:quadr_casimir_defn}, as written, is what one refers to as the quadratic Casimir and from a representation theoretic point-of-view there are some technical differences between the two. Specifically, $\boldsymbol{C}_{\phi}$ is associated with the representation $\phi$ while $\K_{\phi}$ makes use of the notion of tensor power representation. As such (also shown in~\cite{diez2022expectation}), given irreducible $\phi$, the corresponding split Casimir operator is proportional to the identity on the irreducible components of $\phi \otimes \phi$ but in general is not on $\phi$. 
The version of the split Casimir considered in the main text contains an additional normalization by the representation index $I_{\phi}$. Specifically, the $\K$ used in the main text is defined as $\K = I_{\phi}^{-1}\K_{\phi}$. 

\subsection{Cartan Subalgebras, Weights, and Roots}

There is a particular convenient basis for $\g$ consisting of elements that lie in its complexification, i.e. the algebra $\g^{\mathbb{C}} := \g + i\g$. In addition, any representation of $d\phi$ can be uniquely linearly extended to a representation for the complexification, in the natural way, which we denote by $d\phi_{\mathbb{C}}$. This basis is called the Cartan--Weyl basis, which we denote as
\begin{align}
    \{H_i\}_{i=1}^{r} \cup \{E_{\vec{\alpha}}, E_{-\vec{\alpha}}\}_{\alpha \in \Delta^{+}},
\end{align}
where $E_{\vec{\alpha}}$ and  $E_{-\vec{\alpha}}$ can be chosen to be dual under the Killing form.

The set of elements $\{H_i\}_{i=1}^{r}$ are mutually commuting, which form a subalgebra called a \emph{Cartan subalgebra}. Hence,  if we consider a representation of $\g$, $d\phi_{\mathbb{C}}$, then $\{d\phi_{\mathbb{C}}(H_i)\}_{i=1}^{r}$ are simultaneously diagonalizable linear operators. Thus to each simultaneous eigenvector $v \in V$ we can associate a real-valued linear functional $\vec{\omega}_{v}$, called a weight, that satisfies 
\begin{align}
\label{eqn:egval_eqn}
d\phi_{\mathbb{C}}(H_i)v = \vec{\omega}_v(H_i)v
\end{align}
for $1 \leq i \leq r$. It is known that the weights can be viewed as $r$-dimensional vectors on a lattice. 
The set $\Delta^{+}$ consists of a particular basis (as a lattice) for the weights called the positive simple weights.  The other elements indexed by this set $E_{\vec{\alpha}}, E_{-\vec{\alpha}}$ satisfy the following fundamental property:
if $v$ is a weight vector with weight $\omega$, then $\forall i$.
\begin{align}
d\phi_{\mathbb{C}}(H_i)d\phi_{\mathbb{C}}(E_{\pm\vec{\alpha}})v = (\vec{\omega} \pm \vec{\alpha})(H_i)d\phi_{\mathbb{C}}(E_{\pm\vec{\alpha}})v.
\end{align}
Physicists may recognize these as raising and lowering operators. Then in this basis, the split Casimir is expressed as:
\begin{align}
\K = I_{\phi}^{-1}\left(\sum_{i=1}^{r} d\phi_{\mathbb{C}}(H_i)\otimes d\phi_{\mathbb{C}}(H_i) + \sum_{\vec{\alpha} \in \Delta^+} \left[d\phi_{\mathbb{C}}(E_{\vec{\alpha}})\otimes d\phi_{\mathbb{C}}(E_{-\vec{\alpha}}) + d\phi_{\mathbb{C}}(E_{-\vec{\alpha}})\otimes d\phi_{\mathbb{C}}(E_{\vec{\alpha}})\right]\right).
\end{align}
Lastly, if $V$ is an irreducible representation, it can be expressed as a direct sum of simultaneous, orthogonal eigenspaces of the $\{H_i\}_{i=1}^{r}$ called weight spaces. Specifically, if $V_{\vec{\omega}}$ is the weight space associated to weight $\vec{\omega}$ then $\forall v \in V_{\vec{\omega}}$ and $\forall i$:
\begin{align}
d\phi_{\mathbb{C}}(H_i)v = \vec{\omega}(H_i)v,
\end{align}
generalizing Supplementary Equation~\eqref{eqn:egval_eqn}. 

As mentioned earlier the Killing form is positive definite on $\mathfrak{g}$ and thus can be used to induce a valid inner product on its dual space (which contains roots). More specifically, if $T_{ij}$ is the Killing form, then for any two roots $\vec{\alpha}, \vec{\beta} \in \Delta^{+}$:
 \begin{equation}
        (\vec{\alpha}, \vec{\beta}) = [T^{-1}]_{ij} \alpha_i \beta_j.
    \end{equation}

We begin with a very brief review some of the relevant concepts from the representation theory of complex semisimple Lie algebras and direct the interested reader to standard references for more details \cite{knapp1996lie, hall2013lie, fulton2013representation, goodman2009symmetry, humphreys2012introduction}.

Let $\{H_i\}_{i=1}^{r}$ denote a basis for $\mathfrak{h}$ that is orthonormal with respect to the Killing form (the Killing form is positive definite on $\mathfrak{h}$). Let $d\phi : \g \rightarrow \mathfrak{u}(V)$ denote a finite-dimensional representation and $d\phi' : \g^{\mathbb{C}} \rightarrow \mathfrak{gl}(V)$ the unique complex-linear extension. Since each $d\phi'(H_i)$ can be simultaneous diagonalized, to each simultaneous eigenvector $v \in V$ we associate a real-valued linear functional $\vec{\omega}_{v}$, called a weight, that satisfies $d\phi'(H_i)v = \vec{\omega}_v(H_i)v$. The vector $v$ is called a weight vector, and the space of all vectors associated with the weight, is called a weight space. When $d\phi$ is the adjoint representation the functionals $\vec{\omega}_v$ are called roots.

There is a partial ordering on the weights, i.e. $\mu \geq \tau $ if $\mu - \tau$ has all positive coefficients in its expansion in terms of  $\Delta_{+}$. A consequence is that if $|\vec{\Lambda}\rangle$ is a highest weight vector (unique up to scaling), then $d\phi'(E_{\vec{\alpha}})|\vec{\Lambda}\rangle = 0$. 

For a simple root $\vec{\alpha}_i$, its coroot is defined as
\begin{align}
    \alpha^{\vee}_i = \frac{2\vec{\alpha}^{(i)}}{(\vec{\alpha}^{(i)}, \vec{\alpha}^{(i)})},
\end{align}
which is an element of the double dual space and hence identifiable with elements of $\mathfrak{h}$. Then the associated Gram matrix is \cite{fuchs2003symmetries}:
\begin{align}
    G_{{ij}} = (\alpha^{\vee}_i, \alpha^{\vee}_j) = \frac{2}{(\vec{\alpha}^{(i)}, \vec{\alpha}^{(i)})}\frac{2(\vec{\alpha}^{(i)}, \vec{\alpha}^{(j)})}{(\vec{\alpha}^{(j)}, \vec{\alpha}^{(j)})} = \frac{2}{(\vec{\alpha}^{(i)}, \vec{\alpha}^{(i)})}A_{ij},
\end{align}
where $A_{ij}$ is called the Cartan matrix. The dual basis to the set of coroots is called the Dynkin basis or set of fundamental weights. Thus if $\vec{\lambda}$ is a weight expressed in the Dynkin basis, then its norm can be expressed as
\begin{equation}
    (\vec{\lambda}, \vec{\lambda}) 
    = \sum_{ij} G^{-1}_{ij} n_i n_j,
\end{equation}
where $\{n_i\}$ are the integer components of $\vec{\lambda}$ in the Dynkin basis.

\section{Verifying Theory Reproduces Existing Results}

In this section, we apply the results from the main text to three cases where the expression for the gradient variance is already known in literature. 

\subsection{$\text{SU}(d)$}
We should be able to replicate the result for ordinary barren plateaus, that is, the fully controllable case where we have the group $\textup{SU}(d)$ acting via the standard representation. In that case, the basis is the normalized multi-qubit basis giving $I_\phi = 1$. The Lie algebra $\mathfrak{su}(d)$ of dimension $d^2-1$. Since this is a complete orthonormal basis for traceless Hermitian matrices, for a general Hermitian $\A$, $\|P_\g \A\|_{\text{F}}^2 = \Tr(\A_T^2)$. By assumption  $\O$ is traceless and Hermitian and $\H$ is traceless and skew-Hermitian so $\lVert\O\rVert_{\g}^2 = \Tr(\O^2)$ and $\lVert\H\rVert_{\g}^2 = -\Tr(\H^2)$. Finally, $I_\Ad = 2d$ for $\mathfrak{su}(d)$. So
\begin{equation}
   \textup{GradVar}
    = \frac{2d}{(d^2-1)^2}\Tr(\O^2)\Tr(\H^2)\Tr(\rho_T^2).
\end{equation}
This agrees exactly with the subspace-controllable result, see Eq. 13 in \cite{Larocca2022diagnosingbarren}.

\subsection{$\textup{SU}(d/2) \times \textup{SU}(d/2)$}

The quantum circuit realizing the standard representation of the semisimple $G = \textup{SU}(d/2) \times \textup{SU}(d/2) \subset \textup{SU}(d)$ is an example of a subspace-controllable system. The subalgebra is semisimple: $\g = \g_1 \oplus \g_2$ where each subalgebra is isomorphic to $\mathfrak{su}(d/2)$. Since the system is embedded in $\textup{SU}(d)$, $I_{\phi_1} = I_{\phi_2} = 1$. The projector onto each subalgebra is thus equivalent to taking the partial trace of the (traceless) operator on the other subsystem, i.e. $\H = \H_1 \otimes \1 + \1 \otimes \H_2$, and $P_{\g_i}\H = \H_{i}$. For a general Hermitian $\A$, $\|P_{\g_i} A\|_{\text{F}}^2 = \Tr([\Tr_{i^c}(A_T)]_T^2)$. Since we are in a LASA, we can decompose $\O$ and $\H$ into parts supported on each subalgebra. Thus finally
\begin{equation}
   \textup{GradVar} 
     = \sum_{i=1,2} \frac{2(d/2)}{((d/2)^2-1)^2} \Tr(\H_i^2)\Tr(\O_i^2) \Tr([\Tr_{i^c}(\rho_T)]_T^2).
\end{equation}
For states supported on either subalgebra, this agrees with the aforementioned subspace-controllable result. If the traceless part of the state is a simple sum of parts on the subspaces, i.e. if it is unentangled between subspaces, then the variance is the sum of the variances over the subspaces. However, for a state that is maximally entangled the partial trace on each subspace will be proportional to the identity, and therefore $[\Tr_{i^c}(\rho_T)]_T = 0$, giving a zero total variance. This is consistent with the basic quantum information theory notion that local operators cannot distinguish between maximally-entangled subsystems.

\subsection{Spin--$\frac{d-1}{2}$  representation of $\text{SU}(2)$}

Here we consider the case of the spin$-\frac{d-1}{2}$ irreducible representation, $\phi_{d}$, of $\mathfrak{su}(2)$ in an $n$-qubit quantum system, $d=2^n$.
We have $d_\g = 3$. Let us consider a normalized basis for the standard (spin $1/2$) representation of the Lie algebra: $\{e_1, e_2, e_3\} = \{-\frac{i}{\sqrt{2}}\sigma_x, -\frac{i}{\sqrt{2}}\sigma_y, -\frac{i}{\sqrt{2}}\sigma_z\}$.
Since these have commutation relations $[e_i, e_j] = \sqrt{2}\epsilon_{ijk}e_k$, we relate them to the canonical Hermitian spin$-\frac{d-1}{2}$ generators $\{S_x, S_y, S_z\}$ with commutation $[S_x, S_y] = iS_z$ via $\phi_{d}(e_1) = -i\sqrt{2}S_x$, etc. Recall that, as discussed in the main text, we can use the representation index to express the variance as
\begin{equation}\label{eq:var_alter}
\textup{GradVar}
=\frac{I_{\text{Ad}} \lVert o \rVert_{\g}^{2}\lVert h\rVert_{\g}^{2}}{d_{\g}^{2}} \sum_i \Tr^2(i\rho \E_i).
\end{equation}

It is a well-known result that $I_{\text{Ad}} = 2n$ for $\mathfrak{su}(n)$, and thus we have $I_\Ad = 4$ for this case.
Now for the calculation of the norms. Choosing $\H = -iS_x$, we have $\|h\|_{\g}^2 = \frac{1}{2}$. With $\O = S_x+S_y+S_z$, we get $\|o\|_{\g}^2 = \frac{3}{2}$. 
Finally, $\rho = |m\>\<m|$, giving 
\begin{equation}
    \sum_i \Tr^2(i\rho \E_i) = 2\Tr^2(|m\>\<m| S_z ) = 2(\<m| S_z |m\>)^2 = 2m^2.
\end{equation}
Overall using Eq.~\eqref{eq:var_alter} we get the result
\begin{equation}
   \textup{GradVar} = \frac{2m^2}{3},
\end{equation}
which agrees with the result for the same system found in \cite{Larocca2022diagnosingbarren}.

\section{Proofs of Technical Lemmas}

In this section, we provide the proofs of the technical lemmas that were necessary to derive Theorems~\ref{thm:simple_group} and~\ref{thm:compact} of the main text.

The first lemma shows that the first Haar moment is zero for simple groups, which allowed us to only focus on the second moment when computing the variance. This is a commonly observed regarding the first Haar moment of the gradient. However, we include a proof for the general case for completeness.
\begin{lemma}[Vanishing mean]
\label{lemma:vanish}
Let $G$ be a compact Lie group, and $\phi: G \rightarrow \mathcal{U}(V)$ a representation. Then for any $\O, \A \in \mathfrak{gl}(V)$:
\begin{align}
&\mathbb{E}_{{g^{+}, g^{-}}\sim \mu^{\otimes 2}}[\partial\<\O\>_{\A}]=0.
\end{align}
\end{lemma}
Thus we have that for compact Lie groups the variance of the gradient equals its second moment, and so we only need focus on the latter.

\begin{proof}
In full, the first moment is
\begin{equation}
\mathbb{E}_{{g^{+}, g^{-}}\sim \mu^{\otimes 2}}[\partial\<\O\>_{\A}] 
= \int_{G}\text{Tr}\{\U_{g^{-}}\A \U_{g^{-}}^{\dagger}[\H, \U_{g^{+}}\O \U_{g^{+}}^{\dagger}]\}dg^{+}dg^{-}
= \int_{G}\text{Tr}\left\{\U_{g^{-}}\A \U_{g^{-}}^{\dagger}\left[\H, \int_{G} \U_{g^{+}}\O \U_{g^{+}}^{\dagger} dg^{+}\right]\right\}dg^{-}.
\end{equation}
We can show that it is zero by proving that the commutator is zero for any $\O$. Since twirling projects onto the commutant of the group, $e^{\theta \H}$ commutes with the integral. 
This follows from the invariance of the Haar measure.

Let $\U_h = \phi(h)$ for $h \in G$. Then
\begin{align}
    \U_h \left(\int_{G} \U_{g} \O \U_{g}^{\dagger} dg \right)
        = \int_{G} \U_{(hg)} \O \U_{g}^{\dagger} dg
        = \int_{G} \U_{g'} \O \U_{(h^{-1}g')}^{\dagger} dg'
        = \left(\int_{G} \U_{g'} \O \U_{g'}^{\dagger} dg' \right) \U_h,
\end{align}
where we changed the variable of integration to $g' = hg$ (which leaves  the integral unchanged due to Haar translation invariance) and used the fact that $\U_{g}^\dagger = \U_{g^{-1}}$.

Since differentiation commutes with linear operators, and the commutator is linear in each argument, this implies that 
\begin{equation}
    0 = \frac{d}{d\theta} \left[e^{\theta \H},  \int_{G} \U_{g^{+}}\O \U_{g^{+}}^{\dagger} dg^{+}\right] \Big|_{\theta=0} = \left[\frac{de^{\theta \H}}{d\theta}\Big|_{\theta=0},  \int_{G} \U_{g^{+}}\O \U_{g^{+}}^{\dagger} dg^{+}\right] = \left[\H, \int_{G} \U_{g^{+}}\O \U_{g^{+}}^{\dagger} dg^{+}\right].
\end{equation}
\end{proof}

The next lemma was essential for computing the gradient variance for the simple group case (Theorem~\ref{thm:simple_group} in the main text):

\begin{lemma}[Main Text Lemma \ref{lemma:main}]
  Let $G$ be a compact simple Lie group with Lie algebra $\g$.  Suppose $V$ is a finite-dimensional inner product space, $\phi : G \rightarrow \mathcal{U}(V)$ is a unitary representation of $G$, and $\U_g = \phi(g)$. In addition, $a \in \g$, $\mathbf{A} = d\phi(a)$.
Then the following holds:
    Then
    \begin{equation}
        \int_G (\U_g \mathbf{A} \U_g^{\dagger})^{\otimes 2}\; dg %
        = \frac{\|\A\|^2_{\text{F}}}{d_{\g}} \K,
    \end{equation}
where $\K$ is the split Casimir.
\end{lemma}
\begin{proof}
Since differentiation commutes with linear maps, one can show that %
$\U_g \mathbf{A}\U_g^\dagger = d\phi(\Ad_g(a))$. %
Let $\{e_j\}$ be a basis for $\mathfrak{g}$, then $\{\E_{j}\}$, where each $\E_{j} = d\phi(e_j)$ is skew-Hermitian, is a basis for $d\phi(\g)$. Then, by linearity we have
\begin{equation}
    \U_g \mathbf{A} \U_g^\dagger = \Ad_g(\mathbf{A}) = \sum_{ji} a_i [\Ad_g]_{ij} \E_j,
\end{equation}
where 
\begin{equation}
\U_g \E_i \U_g^\dagger = \sum_{j}[\Ad_g]_{ji}\E_j = d\phi\left(\sum_{ji}[\Ad_g]_{ji}e_j\right).
\end{equation}
Explicitly, $\{[\Ad_g]_{ji}\}$ are the matrix coefficients for the adjoint representation of $G$, and are real since we are dealing with a real Lie algebra. 

Since $a \in \g$, we have that $a = \sum_i a_i e_i$. Then the LHS becomes
\begin{equation}
    \int_G \Ad_g(\mathbf{A}) \otimes \Ad_g(\mathbf{A}) \; dg
    = \sum_{ii'jj'} a_i a_{i'} \int_G [\Ad_g]_{ij}[\Ad_g]_{i'j'}\; dg \;\; (\E_j \otimes \E_{j'}).
\end{equation}

Now let's use the fact that the adjoint representation is irreducible for simple Lie groups.
This allows us to use  Schur orthogonality \cite[Corollary 4.10]{knapp1996lie} to write
\begin{equation}
\label{eqn:adj_orthogonal}
    \int_G [\Ad_g]_{ij}[\Ad_g]_{i'j'}\; dg = \delta_{ii'} \delta_{jj'} \frac{1}{d_{\mathfrak g}}.
\end{equation}
Note, the theorem requires conjugation of one of the terms, however all the coefficients in the adjoint representation are real so we can ignore this. In addition, since $\g$ is compact, its complexification is simple and thus the complex extension of the adjoint representation is irreducible, allowing us to apply Schur orthogonality.
This finally gives
\begin{equation}
    \int_G (\U_g \mathbf{A} \U_g^\dagger)^{\otimes 2} \; dg = \frac{1}{d_{\mathfrak g}} \sum_i a_i^2 \; \sum_j \E_j \otimes \E_j.
\end{equation}
By definition of $\K$, we have that $\sum_j \E_j \otimes \E_j = I_\phi \K$.
\end{proof}

The following is a generalization of Lemma~\ref{lemma:main} to outside the simple group setting and was needed for proving Theorem~\ref{thm:compact} of the main text.
\begin{lemma}[Main Text Lemma \ref{lemma:main_compact}]
Let $G$ be a compact and connected Lie group with Lie algebra $\g$.  Suppose $V$ is a finite-dimensional inner product space, $\phi : G \rightarrow \mathcal{U}(V)$ is a unitary representation of $G$, and $\U_g = \phi(g)$. In addition, $a \in \g$, $\mathbf{A} = d\phi(a)$.
Then the following holds:
\begin{align}
\int_{G}(\U_{g}\A\U_{g}^{\dagger})^{\otimes 2} dg %
= \sum_{\alpha}\frac{\|\A_{\g_{\alpha}}\|^2_{\text{F}}}{d_{\g_{\alpha}}} \K_{\g_{\alpha}} + \A_{\mathfrak{c}}^{\otimes 2},
\end{align}
where %
$\A_{\alpha}$ is the image of the component of $a$ in $\g_\alpha$ under $d\phi$. Likewise, $\K_{\alpha}$ is the split Casimir in the subalgebra $\g_\alpha$.
\end{lemma}
\begin{proof}
Since $\g$ is reductive, the algebra's adjoint representation, $\text{ad}$ breaks into a direct sum of irreducible representations, i.e. the simple ideals of  $\g$ and its center $\mathfrak{c}$. The simple ideals $\g_{\alpha}$ must correspond to non-isomorphic simple $\g$-modules, and since $G$ is connected, they correspond to  non-isomorphic simple $G$-modules. Furthermore, the center $\mathfrak{c}$ breaks up into a direct sum of trivial representations. Thus, the Schur orthogonality relations imply that cross terms are zero, and the integral breaks up:

\begin{align}
\int_G (\U_g \mathbf{A} \U_g^{\dagger})^{\otimes 2}\; dg =  \left(\sum_{\alpha} \int_G \text{Ad}_{\mathbf{U}_{g_{\alpha}}}(\mathbf{A}_{\g_{\alpha}})^{\otimes 2} \; dg_{\alpha}\right) + \mathbf{A}_{\mathfrak{c}}^{\otimes 2} =   \sum_{\alpha}\frac{\|\A_{\g_{\alpha}}\|^2_{\text{F}}}{d_{\g_{\alpha}}} \K_{\g_{\alpha}} + \mathbf{A}_{\mathfrak{c}} ^{\otimes 2},
\end{align}
where the last equality follows from applying Lemma~\ref{lemma:main} to the components of the direct sum. In addition $\A_{\g_{\alpha}} \in d\phi(\g_{\alpha}), \A_{\mathfrak{c}} \in d\phi(\mathfrak{c})$ and $\A = \sum_{\alpha} \A_{\g_{\alpha}} + \A_{\mathfrak{c}}$.
\end{proof}

The above result implies that we expect contributions to the variance from the various subalgebras. %

Lastly, for completeness, we also proof  the following simple fact used in Section \ref{sec:interpretation}.
\begin{lemma}
For any $\H$ in Lie algebra $\g$, we have
\begin{align}
    \lVert \H \rVert_{\text{K}}^2 \leq 2d_{\g}\lVert \H \rVert_{\textup{F}}^2
\end{align}
\end{lemma}
\begin{proof}
Let $\{\E_k\}_{k=1}^{d_{\g}}$ be an orthonormal basis for $\g$, then 
 \begin{align}
\lVert \H \rVert_{\text{K}}^2 = \sum_{j,k=1}^{d_{\g}}\Tr([\H, \E_{k}]\E_{j})^2 = \sum_{k=1}^{d_{\g}} \lVert [\H, \E_{k}] \rVert^2_{\text{F}} \leq 2d_{\g}\lVert \H \rVert_{\text{F}}^2.
 \end{align}
\end{proof}

\section{Proof of Main Text Theorem~\ref{thm:compound_projector}}
As mentioned in the main text, the compound $\text{SU}$ layers can also be viewed as the direct sum of the alternating representations of $\text{SU}(n)$. The $k$-th alternating representation is $\phi_{k} : \text{SU}(n) \rightarrow \mathcal{U}(\bigwedge^{k}\mathbb{C}^{n})$ and is irreducible. The direct sum is then obviously $\phi: \text{SU}(n) \rightarrow \mathcal{U}(\bigoplus_{k=1}^{n}\bigwedge^{k}\mathbb{C}^{n})$. Let $\{e_k\}$ denote the standard basis for $\mathbb{C}^{n}$. The mapping between the qubit state space and $\bigwedge^{k}\mathbb{C}^{n}$ can be explicitly seen by mapping a computational basis state $|S\rangle \mapsto \bigwedge_{i \in [n] | S_{i}=1 } e_{i}$.
We will restrict our analysis to the $k=n/2$ subspace (if $n$ is not even take $k = n/2 +1$ or $k=n/2-1$), which has dimension exponential in $n$, i.e. $\binom{n}{n/2} =\Omega(2^{n/2})$. Since $\phi_{n/2}$ is faithful, the dimension of Lie algebra of $\phi_{n/2}(\text{SU}(n))$ is the same as $\mathfrak{su}(n)$, i.e. $n^2 -1$.

Since we will not be able to use the adjoint representation trick, we need to tackle  computing the second-moment operator
\begin{align}
\mathcal{T}: \A \mapsto \int_{\text{SU}(n)}(\U_{g}\otimes \U_{g}) \A (\U_{g}^{\dagger} \otimes \U_{g}^{\dagger}) dg 
\end{align}
directly using Schur--Weyl duality \cite[Theorem 6.3]{fulton2013representation}. Recall that $\mathcal{T}$ must respect the decomposition of the tensor product representation, i.e. $ \wedge^{n/2} \mathbb{C}^{n} \otimes \wedge^{n/2} \mathbb{C}^{n}$, into irreducible components. For the current setting, the Pierri formula \cite[Exercise 6.16]{fulton2013representation} implies that the decomposition into irreducible components is 
\begin{align}
    \wedge^{n/2} \mathbb{C}^{n} \otimes \wedge^{n/2} \mathbb{C}^{n} = \bigoplus_{a \in [n/2+1]} V_{\lambda_a},
\end{align}
where $\lambda_a$ denotes the partition of the integer $n$ that has $n/2-a$ 2's and  $2a$ 1's. Specifically, the $\lambda_a$ index the Young diagrams on $2n$ boxes of shape $\lambda_a$. Furthermore, due to the form of the integral that appears when computing the variance of the gradient, we only need to consider even $a$. This is because if $\rho$ is the initial state lying in the Hamming-weight $n/2$ subspace, we are considering the inner product between the integral and a symmetric tensor $\rho \otimes \rho \in  \text{Sym}^{2}(\wedge^{n/2} \mathbb{C}^{n})$. Thus,
\begin{align}
    \text{Sym}^{2}(\wedge^{n/2} \mathbb{C}^{n}) = \bigoplus_{a \in [n/2+1]~\&~ a~\text{is even}} V_{\lambda_a}.
\end{align}

As consequence of Schur--Weyl duality for $\text{SU}(n)$, there is  a basis for $V_{\lambda_{a}} \subset  \wedge^{n/2} \mathbb{C}^{n} \otimes \wedge^{n/2} \mathbb{C}^{n}$, known as the Gelfand--Cetlin basis \cite[Cor 8.1.7]{goodman2009symmetry}, whose elements are in one-to-one correspondence with the \emph{semistandard Young Tableau} (SSYT) of shape $\lambda_{a}$. While this is not an orthogonal basis and thus is not the basis that diagonalizes the HYOs, it will suffice for reasoning about which irreps the tensor $|S\rangle \otimes |S\rangle$  has support on when $|S\rangle $ is a computational basis state. This leads to the following lemma.

\begin{lemma}\label{lemma:comp_basis}
    If $|S\rangle$ is a computational basis state of Hamming weight $k$ then $|S \rangle \otimes |S\rangle$ lies in an irreducible subrepresentation of $\wedge^{k} \mathbb{C}^n \otimes \wedge^{k} \mathbb{C}^n$, specifically $V_{\lambda_0}$.
\end{lemma}
\begin{proof}
Given an SSYT of shape $\lambda_{a}$, an element of the Gelfand--Cetlin basis is formed by symmetrizing over the rows of the SSYT, where elements in the same column correspond to antisymmetrized indices. One can verify that if $|S\rangle$ is a computational basis state, $|S\rangle \otimes |S\rangle$ corresponds to an SSYT of shape $\lambda_0 = (2,\dots, 2)$ and weight also $(2,\dots, 2)$, i.e. $k$ rows and two columns. This SSYT is already symmetric across the rows and so the row symmetrizer acts as identity. The conclusion is that the tensor product of a computational basis state with itself lies in $V_{\lambda_0}$.
\end{proof}
With regards to the integral, the result will be that the projections onto $V_{\lambda_a}$ for $\lambda_a \neq \lambda_0$ will not contribute if our initial state is a computational basis state of Hamming weight $k$.

Another consequence of Schur--Weyl duality is that the dimension of $V_{\lambda_0}$ is equal to the Schur polynomial for partition $\lambda_0$ evaluated at all $1$'s, i.e. $S_{\lambda_0}(1, \dots, 1)$.

\begin{lemma}
\label{lem:jacobi-trudi}
    Let $\lambda_0$ denote the partition $(2,\dots, 2)$ of $n$, for some even integer $n$. Then,
    \begin{align}
    \dim V_{\lambda_0}={n \choose n/2}^2 \frac{n+1}{(n/2 + 1)^2}.
    \end{align}
\end{lemma}

\begin{proof}
    Note that $\lambda'_0$, the conjugate partition to $\lambda_0$, is given by $(n/2,n/2)$. To evaluate the Schur polynomials, we will use the second form of the Jacobi--Trudi identity~\cite[Equation~A.6]{fulton2013representation} which states that for any partition $\lambda$ of $n$, the Schur polynomial $S_\lambda$ is given by,
    \begin{align}
    \label{eqn:jacobi-trudi-2}
        S_\lambda = \det(e_{\lambda'_i + j -i})_{i,j=1}^{l(\lambda')}
    \end{align}
    where $\lambda'$ is the conjugate partition to $\lambda$, $l(\lambda')$ is its length, and $e_k$ denotes the $k^{th}$ elementary symmetric polynomial on $n$ variables. Recall that the elementary symmetric polynomial $e_{k}$ is the sum of all monomials of total degree $k$, where no individual variable has degree greater than 1.

    Specializing \eqref{eqn:jacobi-trudi-2} to our case, we have that
    \begin{align}
        S_{\lambda_0} = \det\begin{bmatrix}
            e_{n/2} & e_{n/2 + 1} \\
            e_{n/2 - 1} & e_{n/2}
        \end{bmatrix} = e_{n/2}^2 - e_{n/2 - 1}e_{n/2 + 1}.
    \end{align}
    It remains to evaluate $e_{n/2},e_{n/2 - 1},e_{n/2 + 1}$ at the points of interest, which can be done in our case from simple counting arguments.

    To evaluate the polynomials at $(1,1,\dots,1)$, we note that each elementary polynomial $e_k$ contains exactly ${n \choose k}$ monomials, each of which evaluates to $1$ at $(1,1,\dots,1)$. Thus $e_k(1,1,\dots,1) = {n \choose k}$. Noticing additionally that ${n \choose n/2 - 1} = {n \choose n/2 + 1} = \frac{n/2}{n/2 + 1}{n \choose n/2}$, we have,
    \begin{align}
    \label{eqn:schur_dim}
        S_{\lambda_0}(1,1,\dots,1) &= {n \choose n/2}^2 - {n \choose n/2 - 1}{n \choose n/2 + 1} \\
        &= {n \choose n/2}^2 \left(1 - \frac{(n/2)^2}{(n/2 + 1)^2}\right) \\
        &= {n \choose n/2}^2 \frac{n+1}{(n/2 + 1)^2}.
    \end{align}
\end{proof}

We now have tools to prove Theorem~\ref{thm:compound_projector} from the main text, which we restate below. 

\begin{theorem}[Main Text Theorem \ref{thm:compound_projector}]
    For the quantum compound ansatz if the initial state is a computational basis state with Hamming-weight $\frac{n}{2}$ and the observable is a rank-one projector onto another computational basis state in this space, then
    \begin{align}
        \textup{GradVar} \in \mathcal{O}\left({n \choose n/2}^{-1}\right).
    \end{align}
\end{theorem}
\begin{proof}
Recall that  $\textup{GradVar} = \mathbb{E}_{{g^{+}, g^{-}}\sim \mu^{\otimes 2}}[(\partial\<\O\>)^{2}]$.
Let us write the integral for the second moment in full, and rearrange terms appropriately:
\begin{align}
&\mathbb{E}_{{g^{+}, g^{-}}\sim \mu^{\otimes 2}}[(\partial\<\O\>)^{2}]=\iint_{G} (\text{Tr}(\U_{g^{-}}\rho \U_{g^{-}}^{\dagger}[\H, \U_{g^{+}}\O \U_{g^{+}}^{\dagger}]))^{2} dg^- dg^+ \\
&= \int_G \Tr\left\{\left(\int_{G}\U_g^{-\otimes 2}\rho^{\otimes 2} \U_g^{-\dagger\otimes 2}dg^{-}\right)  [\H, \U_g^+ \O \U_g^{+\dagger}]^{\otimes 2} \right\} \;dg^+.
\end{align}

Given that $\O$ is a projector onto another computational basis state, i.e. an element of $V_{\lambda_0}$, we can use our above arguments to simplify the integral as follows:

\begin{align}
&\mathbb{E}_{{g^{+}, g^{-}}\sim \mu^{\otimes 2}}[(\partial\<\O\>)^{2}]=\frac{1}{\dim V_{\lambda_0}}\int_G \Tr\left\{P_{\lambda_0}[\H, \U_g^+ \O \U_g^{+\dagger}]^{\otimes 2} \right\} \;dg^+\\
&=\frac{1}{\dim V_{\lambda_0}}\Tr\big\{\int_{G}(\U_{g^+}\O \U_{g^+}^{\dagger})^{\otimes 2}dg^+[P_{\lambda_0}\H^{\otimes 2} - (\H\otimes \1)P_{\lambda_0}(\1\otimes \H) \nonumber \\&- (\1 \otimes \H)P_{\lambda_0}(\H\otimes \1) + \H^{\otimes 2}P_{\lambda_0}]\big\}\\
&=\frac{1}{\dim V_{\lambda_0}^2}\Tr\big\{P_{\lambda_0}[P_{\lambda_0}\H^{\otimes 2} - (\H\otimes \1)P_{\lambda_0}(\1\otimes \H) \nonumber \\&- (\1 \otimes \H)P_{\lambda_0}(\H\otimes \1) + \H^{\otimes 2}P_{\lambda_0}]\big\}\\
&=\frac{1}{\dim V_{\lambda_0}^2}\left(2\Tr[P_{\lambda_0}\H^{\otimes 2}] - \Tr[P_{\lambda_0}(\H\otimes \1)P_{\lambda_0}(\1\otimes \H)] - \Tr[P_{\lambda_0}(\1 \otimes \H)P_{\lambda_0}(\H\otimes \1)]\right).
\end{align}

One can observe that the elements of the Gelfand--Cetlin basis for $V_{\lambda_0}$, i.e. elements in $\text{Im}(P_{\lambda_{0}})$, are invariant under any swaps across the center tensor product, i.e. row swaps on the SSYT.  Furthermore, since any $\U_{g} \otimes \U_{g}$ commutes with $\text{SWAP}$ across the center $\otimes$, we have: $P_{\lambda_0}\text{SWAP} = \text{SWAP}P_{\lambda_0} = P_{\lambda_0}$. This gives

\begin{align}
&\Tr[P_{\lambda_0}(\H\otimes \1)P_{\lambda_0}(\1\otimes \H)]\\
&=\Tr[P_{\lambda_0}(\H\otimes \1)\text{SWAP}P_{\lambda_0}\text{SWAP}(\1\otimes \H)]\\
&=\Tr[P_{\lambda_0}\text{SWAP}(\1\otimes \H)P_{\lambda_0}(\H\otimes \1)\text{SWAP}]\\
&=\Tr[P_{\lambda_0}(\1\otimes \H)P_{\lambda_0}(\H\otimes \1)],
\end{align}

and

\begin{align}
&\Tr[P_{\lambda_0}(\H\otimes \1)P_{\lambda_0}(\1\otimes \H)]\\
&=\Tr[P_{\lambda_0}(\H\otimes \1)P_{\lambda_0}\text{SWAP}(\1\otimes \H)]\\
&=\Tr[P_{\lambda_0}(\H\otimes \1)P_{\lambda_0}(\H\otimes \1)\text{SWAP}]\\
&=\Tr[P_{\lambda_0}(\H\otimes \1)P_{\lambda_0}(\H\otimes \1)].
\end{align}

Note that $\tilde{\H} := \1 \otimes \H + \H \otimes \1$ commutes with $P_{\lambda_0}$ since  $\phi(\text{SU}(n))$ commutes with $P_{\lambda_0}$ and $\tilde{\H} \in d\phi(\mathfrak{su}(n)) \otimes d\phi(\mathfrak{su}(n))$. 

All of the above results imply that:

\begin{align}
&\Tr[P_{\lambda_0}(\1\otimes \H^2 + \H^2 \otimes \1)] + 2\Tr[P_{\lambda_0}\H^{\otimes 2}]\\
&=\Tr[P_{\lambda_0}(\1\otimes \H + \H\otimes \1)^{2}]\\
&=\Tr[P_{\lambda_0}(\1\otimes \H + \H\otimes \1)P_{\lambda_0}(\1\otimes \H + \H\otimes \1)]\\
&=\Tr[P_{\lambda_0}(\H\otimes \1)P_{\lambda_0}(\H\otimes \1)] +
\Tr[P_{\lambda_0}(\H\otimes \1)P_{\lambda_0}(\1\otimes \H)] \\&\quad+
\Tr[P_{\lambda_0}(\1\otimes \H)P_{\lambda_0}(\H\otimes \1)]+
\Tr[P_{\lambda_0}(\1\otimes \H)P_{\lambda_0}(\1\otimes \H)]\\
&=4\Tr[P_{\lambda_0}(\H\otimes \1)P_{\lambda_0}(\H\otimes \1)],
\end{align}

which implies that
\begin{align}
&2\Tr[P_{\lambda_0}\H^{\otimes 2}] - 2\Tr[P_{\lambda}(\H\otimes \1)P_{\lambda_0}(\H\otimes \1)]  =2\Tr[P_{\lambda_0}(\H\otimes \1)P_{\lambda_0}(\H\otimes \1)] - \Tr[P_{\lambda_0}(\1\otimes \H^2 + \H^2\otimes \1)].
\end{align}

Since $\H$ is skew-Hermitian and $P_{\lambda_0}$ is Hermitian, we have that $-\H^2 \succcurlyeq 0 $, and 
\begin{align}
    -\Tr[P_{\lambda_0}(\H\otimes \1)P_{\lambda_0}(\H\otimes \1)] = \lVert P_{\lambda_0} (\H \otimes \1)\rVert_{F}^{2}.
\end{align}

Thus,

\begin{align}
\mathbb{E}_{{g^{+}, g^{-}}\sim \mu^{\otimes 2}}[(\partial\<\O\>)^{2}]&=\frac{2}{\dim V_{\lambda_0}^2}\left(\Tr[P_{\lambda_0}\H^{\otimes 2}] - \Tr[P_{\lambda_0}(\H \otimes \1) P_{\lambda_0}(\H \otimes \1)]\right)
\\&=\frac{1}{\dim V_{\lambda_0}^2}(2\Tr[P_{\lambda}(\H\otimes \1)P_{\lambda_0}(\H\otimes \1)] - \Tr[P_{\lambda_0}(\1\otimes \H^2 + \H^2\otimes \1)] )\\
&=\frac{1}{\dim V_{\lambda_0}^2}(\Tr[P_{\lambda_0}(\1\otimes -\H^2 + -\H^2\otimes \1)]
-2\lVert P_{\lambda_0}(\H \otimes \1) \rVert_{F}^2)\\
&=\frac{2}{\dim V_{\lambda_0}^2}(\Tr[P_{\lambda_0}(\1\otimes -\H^2)P_{\lambda_0}]
-\lVert P_{\lambda_0}(\H \otimes \1) \rVert_{F}^2)\\
&\leq\frac{2}{\dim V_{\lambda_0}^2}\Tr[P_{\lambda_0}(\1\otimes -\H^2)P_{\lambda_0}]\\
&\leq\frac{2\lVert\H\rVert_{2}^{2}}{\dim V_{\lambda_0}},
\end{align}
where the fourth equality follows from invariance of $P_{\lambda_0}$ under $\text{SWAP}$, and the last inequality follows from $\1\otimes -\H^2 \succcurlyeq 0$.

Since $\H$ is effectively the restriction of one of the quantum compound ansatz generators to  the Hamming-weight $n/2$ subspace, the spectral norm is constant.
Thus
\begin{align}
    \mathbb{E}_{{g^{+}, g^{-}}\sim \mu^{\otimes 2}}[(\partial\<\O\>)^{2}] \in \mathcal{O}(1/\dim V_{\lambda_0}).
\end{align}
The result follows by plugging in the result of Supplementary Lemma~\ref{lem:jacobi-trudi}.
\end{proof}

\section{Projected Norm Lower Bound}

There is actually another interpretation of the projected norm $\lVert \rho_{\alpha} \rVert_{\text{F}}^{2}$ in terms of a different norm that has a deeper connection to the simple ideals $\g_{\alpha}$ and leads to a generic lower bound on $\lVert \rho_{\alpha} \rVert_{\text{F}}^{2}$. This section makes use of the representation theory of semisimple Lie algebra (see Supplementary Note~\ref{sec:intro_to_lie} for an introduction).

Let the set $\Delta^{+}$ will denote the collection of positive simple roots.  In addition if $T_{ij}$ is the metric tensor for the Killing form, then for roots $\vec{\alpha}, \vec{\beta} \in \Delta^{+}$ we define the inner product (recall that the Killing form is positive definite when restricted to the Cartan sublagebra for semisimple Lie algebra):
 \begin{equation}
        (\vec{\alpha}, \vec{\beta})_{\text{w}} := [T^{-1}]_{ij} \alpha_i \beta_j,
    \end{equation}
which linearly extends to weights expressed in terms of $\Delta^{+}$. The induced norm will be denoted $\lVert \cdot \rVert_{\text{w}}$.

The following lemma characterizes the action of the split Casimir on weight vectors. While potentially already a well-known result, we could not find an existing reference. Thus we include a short proof for completeness, which is a simple computation.

\begin{lemma}
Let $\g$ be a simple Lie algebra.  Suppose $V$ is a finite-dimensional inner product space and $d\phi : \g \rightarrow \mathfrak{u}(V)$ is a representation of $\g$. If $|\vec{\lambda}\rangle \in V$ is a weight vector with weight $\vec{\lambda}$ and $\K$ is the normalized split Casimir, then
\begin{equation}
    \<\vec{\lambda}|^{\otimes 2}\K|\vec{\lambda}\rangle^{\otimes 2} = \frac{\lVert\vec{\lambda}\rVert_{\text{w}}^2}{I_{\phi}}.
\end{equation}
\end{lemma}
\begin{proof}

Let us the denote the Cartan--Weyl basis for the complexification of $\g$ by:
\begin{align}
    \{H_i\}_{i=1}^{r} \cup \{E_{\vec{\alpha}}, E_{-\vec{\alpha}}\}_{\vec{\alpha} \in \Delta^{+}},
\end{align}
where $\Delta^{+}$ is a set of positive simple roots, the $H_i$ span the Cartan subalgebra and $E_{\vec{\alpha}}, E_{-\vec{\alpha}}$ are the ladder operators. In addition, let $d\phi_{\mathbb{C}}$ denote the complexification of $d\phi$.

We can express the normalized split Casimir in the Cartan--Weyl basis as:
\begin{align}
\K = I_{\phi}^{-1}\left(\sum_{i=1}^{r} d\phi_{\mathbb{C}}(H_i)\otimes d\phi_{\mathbb{C}}(H_i) + \sum_{\vec{\alpha} \in \Delta^+} \left[d\phi_{\mathbb{C}}(E_{\vec{\alpha}})\otimes d\phi_{\mathbb{C}}(E_{-\vec{\alpha}}) + d\phi_{\mathbb{C}}(E_{-\vec{\alpha}})\otimes d\phi_{\mathbb{C}}(E_{\vec{\alpha}})\right]\right).
\end{align}
 
  \begin{align}
        \<\vec{\lambda}|^{\otimes 2} \K|\vec{\lambda}\>^{\otimes 2}&= I_{\phi}^{-1}\left(\sum_{i=1}^{r}(\langle \vec{\lambda} |d\phi_{\mathbb{C}}(H_i)|\vec{\lambda}\rangle)^2
        + \sum_{\vec{\alpha} \in \Delta^+} 2\<\vec{\lambda}|d\phi_{\mathbb{C}}(E_{\vec{\alpha}})|\vec{\lambda}\> \<\vec{\lambda}|d\phi_{\mathbb{C}}(E_{-\vec{\alpha}})|\vec{\lambda}\>\right) \\
        \\&= I_{\phi}^{-1}\left(\sum_{i=1}^{r} \lambda_i^2
        + \sum_{\vec{\alpha} \in \Delta^+} 2\<\vec{\lambda}|d\phi_{\mathbb{C}}(E_{\vec{\alpha}})|\vec{\lambda}\> \<\vec{\lambda}|d\phi_{\mathbb{C}}(E_{-\vec{\alpha}})|\vec{\lambda}\>\right) \\
        &=I_{\phi}^{-1}\left(\sum_{i=1}^{r} \lambda_i^2\right)\\
        &= \frac{\lVert\vec{\lambda}\rVert_{\text{w}}^2}{I_{\phi}},
    \end{align}

where $d\phi_{\mathbb{C}}(E_{\pm\vec{\alpha}})$ is zero because the $E_{\pm\vec{\alpha}}$ move between orthogonal weight spaces.
\end{proof}

Suppose $\rho = |\psi\rangle\langle\psi|$ for some unit vector $|\psi\rangle \in V$, which corresponds to the pure state case for quantum. If the representation under consideration $\phi$ is not irreducible, then, by unitarity, we can decompose $\phi$ into orthogonal irreducible components:
\begin{align}
    V = \bigoplus_{r}V_{\phi_r}.
\end{align}
Furthermore, each irreducible $V_{\phi_{r}}$ can be represented as a direct sum of, mutually orthogonal, weight spaces $V_{\vec{\lambda}^{(k)}}$ for weight $\vec{\lambda}^{(k)}$:
\begin{align}
    V_{\phi_{r}}= \bigoplus_{t} V_{\vec{\lambda}^{(t)}}.
\end{align}
Thus $|\psi\rangle \in V$ can be uniquely expressed as a linear combination of weight vectors:
\begin{align}
|\psi\rangle = \sum_{k}\beta_k |\vec{\lambda}^{(k)}\rangle.
\end{align}
To every unit vector  $|\psi\rangle \in V$ we can associate the following vector:
\begin{align}
\label{eqn:generalized_weight}
    \vec{\psi} = \sum_{k} \lvert \beta_k \rvert^{2} \vec{\lambda}^{(k)},
\end{align}
which by the unit vector assumption is a convex combination of weights. Note that a generalization to non-unit vectors follows trivially, i.e. $c|\psi\rangle \implies c^2\vec{\psi}$. Since $V$ is a complex vector space, a single $\vec{\psi}$ can be associated with multiple unit vectors.

We can linearly extend the inner product for weights, $(\cdot, \cdot)_{\text{w}}$, to the quantity in Supplementary Equation~\eqref{eqn:generalized_weight} to obtain the norm:
\begin{align}
    \lVert \vec{\psi}\rVert_{\text{w}}^2 = \sum_{k, j} \lvert \beta_k \rvert^{2} \lvert \beta_j \rvert^{2}(\vec{\lambda}^{(k)}, \vec{\lambda}^{(j)})_{\text{w}}.
\end{align}

This is the quantity that lower bounds the projected norm, as put concretely in the following result.
\begin{theorem}[Projected Norm Lower Bound]
\label{thm:project_norm_lower_bound}
Suppose $\phi$ is unitary representation of a simple Lie algebra $\g$, then for any unit vector $|\psi\rangle \in V$ the following holds
\begin{equation}
     \|\rho_{\g}\|_{\textup{F}} \ge  \frac{\lVert \vec{\psi}\rVert_{\textup{w}}^2}{I_{\phi}},
\end{equation} 
where $\rho = |\psi\rangle\langle \psi|$.
\end{theorem}
\begin{proof}
    We know that by definition $\|P_\g |\psi\>\<\psi|\|_{\text{F}} =\<\psi|^{\otimes 2}\K|\psi\>^{\otimes 2} $. Using the expression for the normalized split Casimir in the Cartan--Weyl basis:
    \begin{align}
    &\<\psi|^{\otimes 2}\K|\psi\>^{\otimes 2} = I_{\phi}^{-1}\left(\sum_{i=1}^{r}(\langle \psi |d\phi_{\mathbb{C}}(H_i)|\psi\rangle)^2
        + \sum_{\vec{\alpha} \in \Delta^+} 2\<\psi|d\phi_{\mathbb{C}}(E_{\vec{\alpha}})|\psi\> \<\psi|d\phi_{\mathbb{C}}(E_{-\vec{\alpha}})|\psi\>\right).
    \end{align}
    
    Note that
    \begin{align}
    \sum_{\vec{\alpha} \in \Delta^{+}} \<\psi|^{\otimes 2} d\phi_{\mathbb{C}}(E_{\vec{\alpha}}) \otimes d\phi_{\mathbb{C}}(E_{-\vec{\alpha}}) |\psi\>^{\otimes 2}
    &= \sum_{\vec{\alpha} \in \Delta^{+}}\<\psi| d\phi_{\mathbb{C}}(E_{\vec{\alpha}}) |\psi\> \<\psi| d\phi_{\mathbb{C}}(E_{-\vec{\alpha}}) |\psi\>\\
    &= \sum_{\vec{\alpha} \in \Delta^{+}}\<\psi| d\phi_{\mathbb{C}}(E_{\vec{\alpha}}) |\psi\> \<\psi| d\phi_{\mathbb{C}}(E_{\vec{\alpha}})^{\dagger} |\psi\>\\
    &= \sum_{\vec{\alpha} \in \Delta^{+}} |\<\psi| d\phi_{\mathbb{C}}(E_{\vec{\alpha}}) |\psi\>|^2 \ge 0
    \end{align}
    since by the unitarity of the representation: $d\phi_{\mathbb{C}}(E_{-\vec{\alpha}}) = d\phi_{\mathbb{C}}(E_{\vec{\alpha}})^\dagger$.
    Therefore:
    \begin{align}
    \<\psi|^{\otimes 2}\K|\psi\>^{\otimes 2}
    &\ge I_{\phi}^{-1}\sum_{kr}\lvert \beta_k \rvert^{2}\lvert \beta_r \rvert^{2}\sum_{i=1}^{r}\vec{\lambda}^{(k)}(d\phi_{\mathbb{C}}(H_i))\cdot\vec{\lambda}^{(r)}(d\phi_{\mathbb{C}}(H_i)) \\ 
    &= I_{\phi}^{-1}\sum_{kr}\lvert \beta_k \rvert^{2}\lvert \beta_r \rvert^{2}(\vec{\lambda}^{(k)}, \vec{\lambda}^{(r)})_{\text{w}}
    \\ &= \frac{\lVert \psi \rVert_{\text{w}}^2}{I_{\phi}}.
    \end{align}
\end{proof}

Note that $\sum_{\vec{\alpha}} |\<\psi| d\phi_{\mathbb{C}}(E_{\vec{\alpha}}) |\psi\>|^2 > 0$ if and only if the vector has support on at least two weight spaces where the weights differ by a root $\vec{\alpha}$. Thus we have equality, for example, if $|\psi\> = |\vec{\lambda}^{(k)}\>$ for some weight $\vec{\lambda}^{(k)}$, however there can be  weights that differ by multiples of a root, which by definition is not a root, so this is not a necessary condition.

We can use Supplementary Theorem~\ref{thm:project_norm_lower_bound} to give a lower bound on %
the variance when the initial state $\rho$ is pure:
\begin{equation}
\textup{GradVar} \geq \frac{I_{\text{Ad}} \lVert o \rVert_{\g}^{2}\lVert h\rVert_{\g}^{2}}{d_{\g}^{2}} \cdot \lVert \vec{\psi}\rVert_{\text{w}}^2,
\end{equation} 
which helps to remove the dependence on the index $I_\phi$. Note that the only $\phi$-dependent quantity %
is $\lVert \vec{\psi}\rVert_{\text{w}}^2$. %
In the semisimple case, i.e. arbitray LASA, the above result can be applied separately to each simple ideal to yield:
\begin{align}
 \textup{GradVar} \geq \sum_{\alpha}\frac{I_{\text{Ad}_{{\alpha}}} \lVert o \rVert_{\g_{\alpha}}^{2}\lVert h\rVert_{\g_{\alpha}}^{2}}{d_{\g_{\alpha}}^{2}} \cdot \lVert \vec{\psi}\rVert_{\text{w}_{\alpha}}^2,
\end{align}
where the subscript $\alpha$ is added to the norm $\lVert\cdot\rVert_{\text{w}}$ to emphasize it is different for each ideal. Specifically, the norm depends on  the restriction of $d\phi$  to $\g_{\alpha}$, which may break $V$ into a different set of irreducible components for different $\alpha$. Since $I_{\text{Ad}} = \Theta(\sqrt{d_{\g_{\alpha}}})$, this implies that
\begin{align}
\label{eqn:lower_bound_grad_Var}
\textup{GradVar} \in \Omega\left(\sum_{\alpha}\frac{\lVert o \rVert_{\g_{\alpha}}^{2}\lVert h\rVert_{\g_{\alpha}}^{2}\lVert \vec{\psi}\rVert_{\text{w}_{\alpha}}^2}{d_{\g_{\alpha}}^{3/2}}\right).
\end{align}
One can contrast this with the upper bound presented in the main text where $\O$ and $\H$ had $\phi$-dependent norms, i.e. were Frobenius norms. Alternatively, the above lower bound has shifted all the $\phi$ dependence to the $\lVert \vec{\psi}\rVert_{\text{w}_{\alpha}}^2$ norms. 

If $\O$ and $\H$ are chosen such that they have mutual alignment on $\g_{\alpha}$'s that do not have an exponentially growing dimension, then Supplementary Equation~\eqref{eqn:lower_bound_grad_Var} gives more insight into how the initial state needs to be chosen to avoid a BP. In this setting, one can view the terms
\begin{align}
\frac{\lVert o \rVert_{\g_{\alpha}}^{2}\lVert h\rVert_{\g_{\alpha}}^{2}}{d_{\g_{\alpha}}^{3/2}}
\end{align}
as coefficients that weigh the different norms $\lVert \vec{\psi}\rVert_{\text{w}_{\alpha}}^2$. Thus, the goal is to select an initial state such that $\lVert \vec{\psi}\rVert_{\text{w}_{\alpha}}^2$ are decaying slowly on the subalgebras which $\O$ and $\H$ agree on. A poor choice of initial state can cause the variance to fall faster than the DLA dimension. Unfortunately, the $\lVert \vec{\psi}\rVert_{\text{w}_{\alpha}}^2$ can be challenging to determine in practice, even at small scales. Thus, at the moment we consider the lower bound to be more of theoretical interest and use it to further enlighten the BP phenomenon in LASA.

\section{Applicability of Theory beyond LASA}

\label{sec:genobservable}
In this section, we present generalizations of the Theorem~\ref{thm:compact} in the main text to the non-LASA case. Specifically, the same tools utilized to derive the results for LASA  can be used to obtain a lower bound on gradient variance for an arbitrary observable, in terms of the LASA component. Unfortunately, the expression is not as concise as the ones in the main text, and the various factors that contribute to a BP can be challenging to compute. However, there are still some interesting observations that can be made. For clarity, the proofs of the results of  this section have been delayed to a separate subsection.

Let us denote the orthogonal complement, under the standard Frobenius inner product, of $d\phi(\g)$ within $\mathfrak{u}(V)$ as $(d\phi(\g))^c$. Then we can decompose an arbitrary element $\O \in \mathfrak{u}(V)$ (a generic skew-Hermitian operator on $V$) as 
\begin{equation}
    \O = \O_{\g} + \O_{\g^c},
\end{equation}
where $\O_{\g} \in d\phi(\g)$ and $\O_{\g^c} \in (d\phi(\g))^c$. 
Via conjugation, $\phi$ induces an action of $G$ on the whole of $\mathfrak{u}(V)$ (\emph{a real unitary representation}), which breaks up into two subrepresentations. One is the usual adjoint representation $\Ad_{G}$ on $d\phi(\g)$, and the other is $\phi(G)$ acting on $(d\phi(\g))^{c}$ via conjugation. The orthogonal complement is also an invariant subspace.
Since $G$ is compact, both $d\phi(\g)$ and its complement will decompose into a direct sum of irreducible components. Thus overall,
\begin{align}
\label{eqn:u_decomp}
    \mathfrak{u}(V) = \bigoplus_{\kappa}W_{\kappa},
\end{align}
where each $W_{\kappa}$ is an irreducible subspace over $\mathbb{R}$ under conjugation by $\phi(G)$.

The space $d\phi(\g)$ breaks into a direct sum of simple ideals, where all methods from the main text apply. Unfortunately, it is not possible to make general statements about the decomposition of the complement besides that there may be an irreducible component in $(d\phi(\g))^{c}$ that is $G$-isomorphic to one of the simple ideals of $\g$.

For an observable that has support on the complement, the cost function will generally split into three terms: the variance on $d\phi(\g)$, the variance on $(d\phi(\g))^c$, and an interaction term (the covariance). Explicitly, it is the following sum:
\begin{align}
   \textup{GradVar}
    =&\text{Var}_{(g_{+}, g_{-}) \sim \mu^{\otimes 2}}[\partial\langle \O_{\g} \rangle_{\rho}]
    + \text{Var}_{(g_{+}, g_{-}) \sim \mu^{\otimes 2}}[\partial\langle \O_{\g^c} \rangle_{\rho}]\\
    &+ 2\text{Cov}_{(g_{+}, g_{-}) \sim \mu^{\otimes 2}}[\partial\langle \O_{\g}\rangle_{\rho}, \partial\langle \O_{\g^c}\rangle_{\rho}].
\end{align}

Interestingly, for the covariance terms, all that matters is whether there is an irreducible component, $W_{\kappa}$, of $(d\phi(\g))^{c}$ that is $G$-isomorphic to a simple ideal of $\g$, i.e. isomorphic to an irreducible component of $d\phi(\g)$. Furthermore, since each of the simple ideals are non-isomorphic, there is only one cross term per $W_{\kappa}$.  Both of the previous statements follow from Schur orthogonality, which will cause cross terms involving non-isomorphic irreducible representations to be annihilated. Thus, we can further split a general observable into three components
\begin{align}
     \O = \O_{\g} + \O_{\g^c_{\cong}} + \O_{\g^c_{\not\cong}},
\end{align}

where $\O_{\g^c_{\cong}}$ denotes the sum of components that are in an irreducible component that is $G$-isomorphic to a simple ideal and $\O_{\g^c_{\not\cong}}$ is the sum of those that are not.  We will term the former as the \emph{ideal complement} component and the later as the \emph{non-ideal complement} component. We will call $\O_{\g}$ the \emph{ideal} component. So in Supplementary Equation~\eqref{eqn:u_decomp} each $W_{\kappa}$ is  either part of  the ideal, ideal complement or non-ideal complement components. For two indices $\kappa, \kappa'$, the notation $\kappa \cong \kappa'$ will imply that $W_{\kappa}$ is $G$-isomorphic to $W_{\kappa'}$.
Lastly, note that depending on the observable and Lie algebra, the ideal or non-ideal complement components could be empty.

As just discussed, Schur orthogonality allows us to express the general variance as
\begin{align}
\label{eqn:gen_variance}
   \textup{GradVar}
=&\text{Var}_{(g_{+}, g_{-}) \sim \mu_{\alpha}^{\otimes 2}}[\partial\langle \O_{\g} \rangle_{\rho}]\\
&+ \text{Var}_{(g_{+}, g_{-}) \sim \mu^{\otimes 2}}[\partial\langle \O_{\g^c_{\cong}} \rangle_{\rho}]\\ &+ \text{Var}_{(g_{+}, g_{-}) \sim \mu^{\otimes 2}}[\partial\langle \O_{\g^c_{\not\cong}} \rangle_{\rho}]\\
&+ 2\text{Cov}_{(g_{+}, g_{-}) \sim \mu^{\otimes 2}}[\partial\langle \O_{\g}\rangle_{\rho}, \partial\langle \O_{\g^c_{\cong}}\rangle_{\rho}].
\end{align}

 The lower bound on the general variance that we present is in terms of the ideal and ideal complement components of Supplementary Equation \eqref{eqn:gen_variance}, which together we call the \emph{DLA component} of the variance. Before proceeding with the result, we start with defining two new quantities that will appear in the lower bound.

\begin{definition}
\label{def:corr_frob}
Suppose $W_1$ and $W_2$ are $G$-isomorphic irreducible components of $\mathfrak{u}(V)$, with an arbitrary $G$-isomorphism $\gamma: W_2 \rightarrow W_1$. Let $i\mathbf{B}, i\mathbf{C} \in \mathfrak{u}(V)$, and let $i\mathbf{B}_1 ~(i\mathbf{B}_2)$ and $i\mathbf{C}_1 ~(i\mathbf{C}_2)$ denote their orthogonal projections (under the Frobenius inner product) onto $W_{1}~(W_{2})$. Then we define:
\begin{align}
    (\mathbf{B}, \mathbf{C})_{1, 2} = \Tr(\mathbf{B}_1\gamma \mathbf{B}_2)\Tr(\mathbf{C}_1(\gamma^{-1})^{\dagger}\mathbf{C}_{2}),
\end{align}
and it is independent of $\gamma$.
\end{definition}
Intuitively, $(\mathbf{B}, \mathbf{C})_{1, 2}$ is a type of ``product of autocorrelations'' of $\mathbf{B}$ and $\mathbf{C}$ w.r.t. their projections onto the two irreducible components.

The next quantity we define is similar to the previous one but will generalize the Killing norm quantity.
\begin{definition}
\label{def:corr_kill}
Suppose $W_1$ and $W_2$ are $G$-isomorphic irreducible components of $\mathfrak{u}(V)$, with an arbitrary $G$-isomorphism $\gamma: W_2 \rightarrow W_1$. Let $i\mathbf{A} \in \mathfrak{u}(V)$, and $P_1$, $P_2$ be orthogonal projectors onto $W_1$ and $W_2$ respectively. Then, for arbitrary orthonormal bases $\{u_j\}$ and $\{v_k\}$ of $W_1$ and $W_2$, respectively, we define:
\begin{align}
    (\A)_{1, 2} := \sum_{j, k}-\langle u_{j}, (\gamma^{-1})^{\dagger}v_{k}\rangle \Tr(P_1([\A, u_j])\gamma P_2([\A, v_k])_2),
\end{align}
and it is independent of $\gamma$.
\end{definition}
If one takes $\A$ to be such that $\A \in d\phi(\g)$, then the operator $\O \mapsto [\A , \O]$ preserves the subspaces $W_1$ and $W_2$, so there is no need for the projectors $P_1$ and $P_2$. If we had $W_1 = W_2 = d\phi(\g_{\alpha})$, then $\gamma$ would be the identity and the above quantity would be the Killing norm.

Since the variance of the complement must also be positive, we can ignore it, which leads to the following lower bound on the variance for general observables $\O$.

\begin{theorem}[Lower Bound by DLA component]
\label{thm:cov_thm}
Let $\O$ be an arbitrary observable, then
    \begin{align}
        \textup{GradVar} &\geq \sum_{\alpha}\frac{1}{d_{\g_{\alpha}}^2}(\lVert  \H_{\g_{\alpha}} \rVert_{\textup{K}}^{2} \lVert  \O_{\g_{\alpha}} \rVert_{\textup{F}}^{2} \lVert \rho_{\g_{\alpha}} \rVert_{\textup{F}}^{2}\\
        &+\sum_{\kappa | \kappa \cong \alpha}(\lVert \O_\kappa\rVert_{\textup{F}}^{2}\lVert \rho_\kappa\rVert_{\textup{F}}^{2}(i\H)_{(\kappa, \kappa)} \\
        &+ \sum_{\{\kappa' | \kappa' \cong \kappa\}\cup\{\alpha\}}(\O, \rho)_{(\kappa, \kappa')}(i\H)_{(\kappa, \kappa')} )),
    \end{align} 
where $\alpha$ indexes the ideal component and $\kappa$ indexes the ideal complement component.
\end{theorem}

While this quantity can seem daunting, the reason for presenting it is to highlight that the decay of the variance with the DLA dimension can still appear for observables outside of the DLA, and that the techniques used to obtain the gradient variance for LASA in the previous subsection actually apply more generally. This quantity also reveals the full extent to which the size of the DLA plays a role in the gradient variance, as the non-ideal complement has no dependence on it.

If one can ensure that $\O_{\g^{c}}$ has no support on $W_{\kappa}$ isomorphic to some ideal, then one has a lower bound on the variance given completely by the ideal component, using the results for LASA.  The proof of Supplementary Theorem~\ref{thm:cov_thm} highlights that the covariance between $d\phi(\g)$ and its complement is
\begin{align}
    \text{Cov}(\g, \g^{c}) = \sum_{\alpha}\sum_{\kappa \cong \alpha} \frac{(\O, \rho)_{(\kappa, \alpha)}(i\H)_{(\kappa, \alpha)}}{d_{\g_{\alpha}}^2}.
\end{align}

Unfortunately, cases where there is support on $W_{\kappa}$ isomorphic to some ideal can easily occur. For example, take a 2-qubit system, and consider the subgroup $\textup{SU}(2)$ acting on the first qubit. Then take as operator $\O = \sigma_{x}\otimes \1 - \sigma_{x}\otimes \sigma_{z}$, which has $i\sigma_{x}\otimes \1 \in d\phi(\g)$ and $i\sigma_{x}\otimes \sigma_{z} \in (d\phi(\g))^c$. If we take the state $|00\>$, since the group acts only on the first qubit, the second qubit gives expectation zero always, so that for any group element $g$
\begin{align}
    \<\O\>_g 
    &= \<00|\U_g \sigma_{x}\otimes(\1-\sigma_{z}) \U_g^\dagger |00\> \\
    &= \<0|\U_g \sigma_{x} \U_g^\dagger |0\>\<0|(\1-\sigma_{z})|0\> = 0.
\end{align}
Since the cost function is zero for all elements of the group, the variance must be zero. Also note that this barren plateau is not caused by the state, since not having the second qubit would give a nonzero variance.

We note that Supplementary Theorem~\ref{thm:cov_thm} can actually be extended to a full expression of the gradient variance for an arbitrary observable. In general, this will consist of contributions from the ideal, ideal complement and non-ideal complement components. However, the contributions from the non-ideal complement can be challenging to interpret and there are a few technical caveats that need to be addressed. We have placed the more general result in Supplementary Note~\ref{sec:var_non_ideal_comp}.

\subsection{Proof of Supplementary Theorem~\ref{thm:cov_thm}}

We start this section by showing that the quantity from Supplementary Definitions~\ref{def:corr_frob} and~\ref{def:corr_kill} are actually well-defined, specifically that they is independent of the choice of $G$-isomorphism $\gamma$. We will do this by looking at a generalized version of these quantities and show how it comes up when computing inner products of matrix coefficients. This is done by using Schur orthogonality for isomorphic yet not equal representations (the following lemma). We were unable to find this result in literature, and thus we have included a proof for completeness.
\begin{lemma}
\label{lem:mod_schur_orth}
Suppose $\phi : G \rightarrow \mathcal{U}(V)$ is  a unitary representation of $G$, and $V_{1}, V_{2}$ are two isomorphic subrepresentations of $V$. Let $\gamma : V_{2} \rightarrow V_{1}$ be an arbitrary $G$-isomorphism. In addition, we assume that either of the following two conditions is satisfied:
\begin{itemize}
    \item if $\psi$ is over $\mathbb{C}$, then $V_1$ and $V_2$ are irreducible,
    \item or if $\psi$ is over $\mathbb{R}$, then $V_1$ and $V_2$ are irreducible and their complexifications are irreducible.
\end{itemize}
Then, in orthonormal bases, the following holds for the matrix coefficients:
\begin{align}
\dim V_1\int_{G}\phi^{(1)}_{j_1, j_2}(g)\overline{\phi^{(2)}_{k_1, k_2}(g)}dg = \langle \gamma v_{k_2}, u_{j_2}\rangle \langle \gamma^{-1}u_{j_1}, v_{k_1}\rangle.
\end{align}
\end{lemma}
\begin{proof}
The proof of the lemma is a modification of the  standard Schur orthogonality proof found in Ref.~\cite[Corollary 4.10]{knapp1996lie}.
Let $\{ u_j \}$ and $\{ v_k \}$ denote orthonormal bases for $V_{1}$ and $V_{2}$, respectively. In addition, $\phi^{(1)}$ and $\phi^{(2)}$ are the corresponding subrepresentations.  We define a linear operator $B_{a,b} : V_{1} \rightarrow V_{2}$ as $B_{a,b}x := \langle x, u_{b}\rangle v_{a}$ for $x \in V_{1}$. We can relate the inner product of matrix coefficients to this operator $B_{a,b}$:
\begin{align}
\int_{G}\phi^{(1)}_{j_1, j_2}(g)\overline{\phi^{(2)}_{k_1, k_2}(g)}dg &= \int_{G} \langle \phi^{(1)}(g)u_{j_1}, u_{j_2}\rangle\langle \phi^{(2)}(g^{-1})v_{k_2},v_{k_1}\rangle dg \\
&= \langle \int_{G} \langle\phi^{(1)}(g)u_{j_1}, u_{j_2}\rangle\cdot\phi^{(2)}(g^{-1})v_{k_2}dg, v_{k_1}\rangle\\
& = \langle \left[\int_{G} \phi^{(2)}(g^{-1})B_{j_2, k_2}\phi^{(1)}(g) dg\right]u_{j_1}, v_{k_1}\rangle\\
& = \langle \mathcal{A}(B_{j_2, k_2})u_{j_1}, v_{k_1}\rangle,
\end{align}

 where for any linear operator $B$, $\mathcal{A}(B)$ is called the twirling operator:
\begin{align}
    \mathcal{A}(B) = \int_{G} \phi^{(2)}(g^{-1})B\phi^{(1)}(g) dg.
\end{align}
The properties of the Haar measure imply that $\mathcal{A}(B)$ is $G$-equivariant for any $B$. Thus the composition of $\mathcal{A}(B)$ and $\gamma$:
\begin{align}
    \gamma \mathcal{A}(B) : V_{1} \rightarrow V_{1}
\end{align}
is $G$-equivariant. 

The hypothesis of the lemma allows us to apply Schur's lemma regardless if the representation complex or real. Specifically, $\gamma \mathcal{A}(B_{j_2, k_2}) = \lambda_{\gamma} \1$, where $\lambda_{\gamma}$ is complex or real depending on which $\psi$ is. In addition, any $G$-isomorphism of $V_{1}$ and $V_{2}$ is a scalar (in the same field that $\psi$ is over) multiple of $\gamma$.

Thus,
\begin{align}
\int_{G}\phi^{(1)}_{j_1, j_2}(g)\overline{\phi^{(2)}_{k_1, k_2}(g)}dg &=  \langle \mathcal{A}(B_{j_2, k_2})u_{j_1}, v_{k_1}\rangle\\
&=\langle \gamma^{-1} \gamma\mathcal{A}(B_{j_2, k_2})u_{j_1}, v_{k_1}\rangle\\
&=\lambda_{\gamma} \langle \gamma^{-1}u_{j_1}, v_{k_1}\rangle.
\end{align}

Let $n:= \dim V_{\psi^{(1)}} = \dim V_{\psi^{(2)}}$. We can use the following to solve for $\lambda_{\gamma}$:
\begin{align}
\lambda_{\gamma} n &= \Tr(\gamma A(B_{j_2, k_2})) \\&= \int_{G}\sum_{j=1}^{n}\langle \gamma \phi^{(2)}(g^{-1})B_{j_2, k_2}\phi^{(1)}(g)u_j, u_j\rangle dg\\
&= \int_{G}\sum_{j=1}^{n}\langle \phi^{(1)}(g^{-1})\gamma B_{j_2, k_2}\phi^{(1)}(g)u_j, u_j\rangle dg\\
&= \int_{G}\sum_{j=1}^{n}\langle \gamma B_{j_2, k_2}\phi^{(1)}(g)u_j, \phi^{(1)}(g)u_j\rangle dg\\
&=\Tr(\gamma B_{j_2, k_2}),
\end{align}

where for $\gamma B_{j_2, k_2}x = \langle x, u_{j_2}\rangle \gamma v_{k_2}$ we have $\Tr(\gamma B_{j_2, k_2}) = \langle \gamma v_{k_2}, u_{j_2}\rangle$. Solving for $\lambda_{\gamma}$ gives:
\begin{align}
\lambda_{\gamma} = \frac{\langle \gamma v_{k_2}, u_{j_2}\rangle}{n}.
\end{align}

We can plug this result back in to obtain an expression for the inner product of matrix coefficients
\begin{align}
\int_{G}\phi^{(1)}_{j_1, j_2}(g)\overline{\phi^{(2)}_{k_1, k_2}(g)}dg = \frac{\langle \gamma v_{k_2}, u_{j_2}\rangle \langle \gamma^{-1}u_{j_1}, v_{k_1}\rangle}{n}.
\end{align}
The $\gamma$-independence of the left-hand side follows from the $\gamma$-independence of  the right-hand side or by noting that any two $G$-equivariant maps will be scalar multiples of each other.
\end{proof}

The following lemma uses the previous result to generalize Lemma~\ref{lemma:main_compact}.

\begin{lemma}
\label{lem:cross_integral}
Suppose $G$ is a compact matrix Lie group with Lie algebra $\g \subseteq \mathfrak{u}(m)$,  where $\mathfrak{u}(m)$ is the Lie algebra of $m \times m$ skew-Hermitian matrices. Suppose  $\psi : G \rightarrow \mathcal{U}(\mathfrak{u}(m))$ corresponds to the real unitary representation of $G$ where it acts via conjugation on $\mathfrak{u}(m)$, and that $\phi$ corresponds to either $\psi$ or its complexification, $\psi_{\mathbb{C}} : G \rightarrow \mathcal{U}(\mathfrak{gl}(m))$. In addition, suppose $i\O \in \mathfrak{u}(m)$ and $V_1$ and $V_2$ are irreducible representations satisfying the hypotheses of Supplementary Lemma~\ref{lem:mod_schur_orth}, and $i\O_1, i\O_2$ orthogonal projections of $\O$ onto $V_1$ and $V_2$, respectively. Then
\begin{align}
\int_{G}\phi^{(1)}(g)i\O_{1}\otimes \phi^{(2)}(g)i\O_{2} dg = \frac{-\Tr(\O_1\gamma \O_2)\tilde{\K}_{(1,2)}}{\dim V_1},
\end{align}

where \begin{align}
\tilde{\K}_{
(1,2)} := \sum_{j, k} u_{j} \otimes v_{k} \langle u_{j}, (\gamma^{-1})^{\dagger}v_{k}\rangle,
\end{align}
for two arbitrary orthonormal bases $\{u_j\}$ and $\{v_k \}$ for $V_1$ and $V_2$ respectively.
\end{lemma}

\begin{proof}
The orthogonal projections of $\O$ can be expressed in the bases defined earlier for $V_1$ and $V_2$ in the previous lemma: $i\O_1 + i\O_2 = \sum_{j} a_ju_j + \sum_{k} b_kv_j$. Let $n= \dim V_1 = \dim V_2$. Thus,
\begin{align}
n\int_{G}\phi^{(1)}(g)\O_{1}\otimes \phi^{(2)}(g)\O_{2} dg =&n\int_{G}\sum_{j}a_j\phi^{(1)}(g)u_j \otimes \sum_{k}b_k\phi^{(2)}(g)v_k dg\\&=n\int_{G}\sum_{j}a_j\sum_{j'}\phi^{(1)}_{j', j}(g)u_{j'} \otimes \sum_{k}b_k\sum_{k'}\phi^{(2)}_{k',k}(g)v_{k'} dg\\
&=n\sum_{j, k, j', k'}a_jb_k\int_{G}\phi^{(1)}_{j', j}(g)\phi^{(2)}_{k',k}(g)dg \cdot u_{j'} \otimes v_{k'} \\
&=n\sum_{j, k, j', k'}a_jb_k \frac{\langle \gamma v_{k}, u_{j}\rangle \langle \gamma^{-1}u_{j'}, v_{k'}\rangle}{n} \cdot u_{j'} \otimes v_{k'}\\
&=\left(\sum_{j, k}a_jb_k \langle u_{j}, \gamma v_{k}\rangle\right) \left(\sum_{j, k} u_{j} \otimes v_{k} \langle u_{j}, (\gamma^{-1})^{\dagger}v_{k}\rangle\right)\\
&=-\Tr(\O_1\gamma \O_2)\left(\sum_{j, k} u_{j} \otimes v_{k} \langle u_{j}, (\gamma^{-1})^{\dagger}v_{k}\rangle\right)\\
&=-\Tr(\O_1\gamma \O_2)\tilde{\K}_{(1, 2)}.
\end{align}

Note that $\tilde{\K}_{(1, 2)}$ must be invariant under the choice of orthonormal bases for $V_1$ and $V_2$ since the left-hand side and $-\Tr(\O_1\gamma \O_2)$ are. 

\end{proof}

The next lemma presents expression for the covariance terms and the variance of the ideal complement.

\begin{lemma}
\label{lem:generalized_integral_lem}
Suppose $G$ is a compact matrix Lie group with Lie algebra $\g \subseteq \mathfrak{u}(m)$,  where $\mathfrak{u}(m)$ is the Lie algebra of $m \times m$ skew-Hermitian matrices. Suppose  $\psi : G \rightarrow \mathcal{U}(\mathfrak{u}(m))$ corresponds to the real unitary representation of $G$ where it acts via conjugation on $\mathfrak{u}(m)$, and that $\phi$ corresponds to either $\psi$ or its complexification, $\psi_{\mathbb{C}} : G \rightarrow \mathcal{U}(\mathfrak{gl}(m))$. In addition, suppose $i\O \in \mathfrak{u}(m)$ and $V_1$ and $V_2$ are irreducible representations satisfying the hypotheses of Supplementary Lemma~\ref{lem:mod_schur_orth}, and $i\O_1, i\O_2$ orthogonal projections of $\O$ onto $V_1$ and $V_2$, respectively. Lastly,  $\forall g \in G, \U_g = \phi(g)$ and $\H$ is arbitrary element of $d\phi(\g)$. Then the following two equalities hold:
\begin{align}
&\iint_{G} \Tr(\U_{g^{-}}i\rho \U_{g^{-}}^{\dagger}[\H, \U_{g^{+}}i\O_1 \U_{g^{+}}^{\dagger}])\Tr(\U_{g^{-}}i\rho \U_{g^{-}}^{\dagger}[\H, \U_{g^{+}}i\O_2\U_{g^{+}}^{\dagger}])dg^+ dg^- = \frac{(\O, \rho)_{(1,2)}(i\H)_{(1, 2)}}{(\dim V_1)^2}\\
&\iint_{G} (\Tr(\U_{g^{-}}i\rho \U_{g^{-}}^{\dagger}[\H, \U_{g^{+}}i\O_1 \U_{g^{+}}^{\dagger}]))^2dg^+ dg^- = \frac{\lVert \O_1\rVert_{\textup{F}}^{2}\lVert \rho_1\rVert_{\textup{F}}^{2}(i\H)_{(1, 1)}}{(\dim V_1)^2},
\end{align}

where $(\O, \rho)_{1,2}$ and $(i\H)_{(1, \cdot)}$ utilize Supplementary Definitions~\ref{def:corr_frob} and~\ref{def:corr_kill} respectively.
\end{lemma}
\begin{proof}

The  first quantity in the emma statement, without taking the trace, is the following integral:
\begin{align}
\label{eqn:one_two_cov}
\text{Moment}_{(1,2)} := \iint_{G} \U_{g^{-}}i\rho \U_{g^{-}}^{\dagger}[\H, \U_{g^{+}}i\O_1 \U_{g^{+}}^{\dagger}] \otimes  \U_{g^{-}}i\rho \U_{g^{-}}^{\dagger}[\H, \U_{g^{+}}i\O_2 \U_{g^{+}}^{\dagger}] dg^+ dg^-.
\end{align}

We can expand the commutator as in the proof of Theorem~\ref{thm:simple_group} of the main text to isolate out an integral of the form:
\begin{align}
\label{eqn:inner_cross_integral}
 \iint_{G} \U_{g^{+}}i\O_1\U_{g^{+}}^{\dagger}\otimes \U_{g^{+}}i\O_2\U_{g^{+}}^{\dagger} dg^+ =  \frac{-\Tr(\O_1\gamma \O_2)\tilde{\K}_{(1,2)}}{\dim W_1},
\end{align}
where we have used Supplementary Lemma~\ref{lem:cross_integral}. If we fix arbitrary orthonormal bases for $W_1$ and $W_2$, $\{ u_j\}$ and $\{v_j\}$ respectively, we can express $\tilde{\K}_{(1,2)}$ as
\begin{align}
\tilde{\K}_{
(1,2)} := \sum_{j, k} u_{j} \otimes v_{k} \langle u_{j}, (\gamma^{-1})^{\dagger}v_{k}\rangle.
\end{align}

If we plug Supplementary Equation~\eqref{eqn:inner_cross_integral} back into Supplementary Equation~\eqref{eqn:one_two_cov} as done in the proof of Theorem~\ref{thm:simple_group}, then we get:
\begin{align}
\text{Moment}_{(1,2)} &= \frac{-\Tr(\O_1\gamma \O_2)}{\dim W_1} \sum_{j, k}\langle u_{j}, (\gamma^{-1})^{\dagger}v_{k}\rangle\int_{G}\U_{g^{-}}[\H, u_j]\U_{g^{-}}^{\dagger}\otimes \U_{g^{-}}[\H, v_k]\U_{g^{-}}^{\dagger} dg^{-}\\&= \frac{-\Tr(\O_1\gamma \O_2)}{(\dim W_1)^2}\left(\sum_{j, k}-\langle u_{j}, (\gamma^{-1})^{\dagger}v_{k}\rangle \Tr(i[\H, u_j]\gamma i[\H, v_k])\right)\left(\sum_{j, k} u_{j} \otimes v_{k} \langle u_{j}, (\gamma^{-1})^{\dagger}v_{k}\rangle\right)\\&=\frac{-\Tr(\O_1\gamma \O_2)(i\H)_{(1, 2)}}{(\dim W_1)^2}\left(\sum_{j, k} u_{j} \otimes v_{k} \langle u_{j}, (\gamma^{-1})^{\dagger}v_{k}\rangle\right)
\end{align}
where we have applied Supplementary Lemma~\ref{lem:cross_integral} again to obtain the second equality. Also, we have utilized Supplementary Definition~\ref{def:corr_kill} and that $[\H, \cdot]$ preserves $W_1$ and $W_2$.

Finally, the overall integral is
\begin{align}
\Tr(i\rho \otimes i \rho \cdot \text{Moment}_{(1,2)}) &= \frac{-\Tr(\O_1\gamma \O_2)(i\H)_{(1, 2)}}{(\dim W_1)^2}\left(\sum_{j, k} \Tr(i\rho u_{j})\Tr(i \rho v_{k}) \langle u_{j}, (\gamma^{-1})^{\dagger}v_{k}\rangle\right)\\&=\frac{\Tr(\O_1\gamma \O_2)\Tr(\rho(\gamma^{-1})^{\dagger}\rho)(i\H)_{(1, 2)}}{(\dim W_1)^2}\\
&=\frac{(\O, \rho)_{(1,2)}(i\H)_{(1, 2)}}{(\dim W_1)^2}.
\end{align}

The second equality in the lemma follows by notating that $\gamma$ becomes the identity map, and $u_j$, $v_k$ are both indexing elements of the same orthonormal basis for $V_1$. The existence of $G$-isomorphic $V_2$ is not needed.

\end{proof}

Now we have all the tools to prove Supplementary Theorem~\ref{thm:cov_thm}.

\begin{proof}[Proof of Supplementary Theorem~\ref{thm:cov_thm}]

Recall that
\begin{align}
  \textup{GradVar}
\geq&\text{Var}_{(g_{+}, g_{-}) \sim \mu_{\alpha}^{\otimes 2}}[\partial\langle \O_{\g} \rangle_{\rho}]
+ \text{Var}_{(g_{+}, g_{-}) \sim \mu^{\otimes 2}}[\partial\langle \O_{\g^c_{\cong}} \rangle_{\rho}] +\nonumber\\
&+ 2\text{Cov}_{(g_{+}, g_{-}) \sim \mu^{\otimes 2}}[\partial\langle \O_{\g}\rangle_{\rho}, \partial\langle \O_{\g^c_{\cong}}\rangle_{\rho}].
\end{align}

We already know the expressions for $\text{Var}_{(g_{+}, g_{-}) \sim \mu_{\alpha}^{\otimes 2}}[\partial\langle \O_{\g} \rangle]$, which follows from Theorem~\ref{thm:compact} of the main text. There are two kinds of terms that will be in $\text{Var}_{(g_{+}, g_{-}) \sim \mu^{\otimes 2}}[\partial\langle \O_{\g^c_{\cong}} \rangle]$. The first kind are just variance terms from the ideal-complement and the second kind are covariances between terms in the ideal-complement that are isomorphic to each other. We can deal with the second class of terms in a similar way to how we deal with $2\text{Cov}_{(g_{+}, g_{-}) \sim \mu^{\otimes 2}}[\partial\langle \O_{\g}\rangle, \partial\langle \O_{\g^c_{\cong}}\rangle]$.

Note that since we assumed $\g$ is a compact Lie algebra, it follows that the second condition of Supplementary Supplementary Lemma~\ref{lem:mod_schur_orth} is satisfied. This is because any irreducible representation isomorphic to a simple ideal of compact real Lie algebra remains irreducible when complexified. Thus all of the previous lemmas apply for the ideal and ideal-complement components.

Thus using Theorem~\ref{thm:compact} of the main text and Supplementary Lemma~\ref{lem:generalized_integral_lem} we get that:
\begin{align}
&\text{Var}_{(g_{+}, g_{-}) \sim \mu_{\alpha}^{\otimes 2}}[\partial\langle \O_{\g} \rangle] = \sum_{\alpha} \frac{\lVert  \H_{\g_{\alpha}} \rVert_{\text{K}}^{2} \lVert  \O_{\g_{\alpha}} \rVert_{\text{F}}^{2} \lVert \rho_{\g_{\alpha}} \rVert_{\text{F}}^{2}}{d_{\g_{\alpha}}^{2}}\\
&\text{Var}_{(g_{+}, g_{-}) \sim \mu^{\otimes 2}}[\partial\langle \O_{\g^c_{\cong}} \rangle] + 2\text{Cov}_{(g_{+}, g_{-}) \sim \mu^{\otimes 2}}[\partial\langle \O_{\g}\rangle, \partial\langle \O_{\g^c_{\cong}}\rangle] \nonumber\\  &=\sum_{\kappa | \kappa \cong \alpha}(\lVert \O_\kappa\rVert_{\textup{F}}^{2}\lVert \rho_\kappa\rVert_{\textup{F}}^{2}(i\H)_{(\kappa, \kappa)} + \sum_{\{\kappa' | \kappa' \cong \kappa\}\cup\{\alpha\}}(\O, \rho)_{(\kappa, \kappa')}(i\H)_{(\kappa, \kappa')} )).
\end{align}

\end{proof} 

\subsection{Variance of Non-ideal Complement}
\label{sec:var_non_ideal_comp}

One may have noticed that the lemmas derived previously were actually general enough that they can allow us to exactly compute the variance of the non-ideal complement.

The non-ideal complement component of the variance corresponds to those irreducible components $W_{\kappa}$ that are not $G$-isomorphic to any simple ideal of $\g$. In this case, it is not necessarily true that complexification of $W_{\kappa}$ remains irreducible. However, if it is not irreducible, it is known that it must decompose into exactly two complex irreducible representations of $G$ \cite{realreps}. This allows us to then apply the techniques of Supplementary Note~\ref{sec:genobservable} to these complex irreps.

We start by complexifying the representation $\psi$, which is done by linearly extending $\psi$ to the complex vector space $\mathfrak{u}(V) + i\mathfrak{u}(V) \cong \mathfrak{gl}(V)$. The extended inner product is defined by: $\A_1 + i\A_2, \textbf{B}_1 + i\textbf{B}_2 \in \mathfrak{u}(V) + i\mathfrak{u}(V): -\Tr((\A_1 + i\A_2)(\textbf{B}_1 - i\textbf{B}_2))$. In addition, the complexified representation remains unitary w.r.t. this inner product. Thus, $\mathfrak{gl}(V)$ breaks into the following sum of orthogonal components:
\begin{align}
   \mathfrak{gl}(V) =  \bigoplus_{\alpha} (d\phi(\g_{\alpha}) +id\phi(\g_{\alpha}))\oplus \left(\bigoplus_{\alpha}\bigoplus_{\kappa\cong \alpha}(W_{\kappa} + iW_{\kappa})\right) \oplus \left(\bigoplus_{\forall \alpha,\kappa\not\cong \alpha}(W_{\kappa} +iW_{\kappa})\right).
\end{align}
The first two sums are the ideal and ideal complement components and remain irreducible when complexified. These were handled in the Supplementary Note~\ref{sec:genobservable}. However, as stated in the previous paragraph, the terms in the last sum, which correspond to the non-ideal complement, will either remain irreducible or split into two complex irreducible components. Thus, we can further break the sum up as follows:

\begin{align}
   &\mathfrak{gl}(V) =  \bigoplus_{\alpha} (d\phi(\g_{\alpha}) +id\phi(\g_{\alpha}))\oplus \left(\bigoplus_{\alpha}\bigoplus_{\kappa\cong \alpha}(W_{\kappa} + iW_{\kappa})\right) \oplus \left(\bigoplus_{\forall \alpha, \kappa\not\cong \alpha \&  \mathbb{C}-\text{irrep}}(W_{\kappa} +i W_{\kappa})\right) \\&\oplus \left(\bigoplus_{\forall \alpha, \kappa\not\cong \alpha \&  \text{not}\mathbb{C}-\text{irrep}}(V_{\kappa} + V^{*}_{\kappa})\right),
\end{align}
where the $V_{\kappa}$ are complex irreps and $^{*}$ denotes the dual representation, which it may or may not be isomorphic to $V_{\kappa}$. Schur orthogonality now applies, and the covariances terms across an $\oplus$ will be annihilated, and since each $V_{\kappa}$ is complex irreducible all of the results of Schur's lemma apply. 

Thus, using that our observable $\O$ satisfies $i\O \in \mathfrak{u}(V)$ the non-ideal variance component will consist of terms of the form:
\begin{align}
\frac{1}{\dim \tilde{W}_{\kappa}}\left(\lVert \O_\kappa\rVert_{\textup{F}}^{2}\lVert \rho_\kappa\rVert_{\textup{F}}^{2}(i\H)_{(\kappa, \kappa)} + \sum_{\kappa' | \kappa' \cong \kappa}(\O, \rho)_{(\kappa, \kappa')}(i\H)_{(\kappa, \kappa')}\right).
\end{align}
where $\tilde{W}_{\kappa}$ may correspond to either a real or complex irrep. Since $\tilde{W}_{\kappa}$ can be a complex vector space, the variance scaling can depend on the dimension of a vector space that contains operators that are not valid quantum observables, i.e. the spaces $W_{\kappa}$, and a novel prediction of our work. This emphasizes that the variance of the non-ideal complement can have no dependence on the dimension of the DLA, yet still, the techniques we utilizes for the LASA setting still work in general for decomposing the variance and identifying the invariant-subspace dimension dependence.

Thus, if $\O$ only has support on the non-ideal complement, all scenarios are possible. If $\g$ has polynomially growing dimension, then the complement is exponential, and potentially there exists an exponentially-large irreducible subspace that would give a BP. Conversely, if $\g$ is exponentially large and its complement is polynomial, it may be that some choices of $\O$ and $\rho$ would avoid a BP.

\section{Details of Numerical Results}
The numerical results displayed in the figures in the main text were obtained  with Qulacs. For each simulation a periodic ansatz using the quantum compound generators was constructed. Specifically, the ansatz is in the brick architecture as described in Ref.~\cite{cherrat2023quantum}, with alternating layers of 2-qubit gates with 1D connectivity. Since the experiments focused on the SU-compound, each gate was composed of a regular FBS gate (generator $h_y$) followed by a generalized FBS gate (generated $h_x$), independently parameterized, in order to express the entire group.

The experiments used $L = 12n$ layers, where $n$ is the number of qubits. The parameters were uniformly sampled from $[0, 4\pi)$. The scaling of $L$ appeared to be sufficient for constructing an approximate $2$-design.
The gradients were estimated via finite difference, and the variance was estimated over $5000$ gradient samples for $n < 18$ and $1000$ for $n\geq 18$, due to time constraints.

Due to a different convention for the Pauli generators in the simulation software used, we rescaled the analytical predictions in the plots by a factor of $1/4$.

\section{Mixing time to t-designs}

In this section, we prove Theorem~\ref{thm:rapid} of the main text regarding the mixing time to $2$-design for LASA. However, the result proved below is slightly more general and Theorem~\ref{thm:rapid} follows when we restrict the DLA to have polynomial dimension and a $2$-design. 

First, we start with a definition. Let, $\{\H_1, \dots, \H_{d_{\g}}\}$ be a basis of skew-Hermitian operators for the Lie algebra $\g$. We define the \emph{minimum stable Killing rank} $r_{\text{K}}$ to be 
\begin{align}
    r_{\text{K}} = \min_k\frac{\lVert \text{ad}_{i\H_k}\rVert_{\text{F}}^2}{\lVert\text{ad}_{i\H_k}\rVert_{\text{op}}^2}.
\end{align}
Note that $1 \leq r_{\text{K}} \leq d_{\g}$. However, we can actually show that $r_{\text{K}} \in \Omega(\sqrt{d_{\g}})$ by investigating  the potential root systems (see Lemma \ref{lem:min_stable_rank}). Corollary \ref{cor:rapid_compound} shows that the quantum compound ansatz saturates this lower bound. In contrast Ref.~\cite{haah2024efficient} showed that Pauli rotations for $\mathfrak{su}(2^n)$ saturate the upper bound.

\begin{theorem}
\label{thm:t_mixing_time}
Consider an orthogonal basis of skew-Hermitian generators $\mathcal{A}:=\{\H_1, \dots, \H_{d_{\g}}\}$ for the DLA with the property that the unitary $e^{-\theta\H_k}$ corresponding to a generator $\H_{k}$ is $t_k$-periodic. Suppose $\H_{k}$ constitute an irreducible representation of the dynamical group $\mathcal{G}$. Consider a LASA formed by applying evolutions $e^{-\theta_k \H_k}$ where $\H_{k}$ is selected uniformly at random from the set $\mathcal{A}$ and the parameter $\theta_{k}$ uniformly from $[0, t_k)$. Then, when $t < \sqrt{d_{g}}/2$, the ansatz is an $\epsilon$-approximate $t$-design for the dynamical group $\mathcal{G}$  after  $\mathcal{O}(\frac{td_{\g}}{r_k}\log(1/\epsilon))$ layers.
\end{theorem}
\begin{proof}
The proof roughly follows the approach of Ref.~\cite{haah2024efficient} with some generalizations made to handle arbitrary LASA.

Suppose the dynamical group is $\mathcal{G}$ and is an irreducible representation, $\phi$, of some compact, connected group $G$. The irreducibility assumption can be satisfied by restricting to the projection of the initial state onto an irreducible invariant subspace of the $n$-qubit Hilbert space.

We can without loss of generality assume that the  Lie algebra $\g$ is also simple. This is because mixing on each simple ideal will correspond to mixing over the whole group. Then, we can consider the minimal spectral gap across all ideals.

Consider the random walk that at each step uniformly selects an element from $\{\H_k\}_{k=1}^{d_{\g}}$ and an angle from $\theta \in [0, t_k]$ and applies $e^{-\theta \H_k}$. The $t$-th moment  corresponding to the walk is:
\begin{align}
    T:=\sum_{k=1}^{d_{\g}}\frac{1}{d_{\g}}\int_{[0, t_k]} (e^{-\theta_k \H_k})^{\otimes t}\otimes (e^{\theta_k \bar{\H}_k})^{\otimes t}d\theta_k,
\end{align}

and let 
\begin{align}
    T^{*} := \int_{G} (\U_{g})^{\otimes t} \otimes (\bar{\U}_{g})^{\otimes t}dg
\end{align}
be the $t$-th Haar moment. The spectral gap \cite{Harrow_2009, Brand_o_2016} of the walk is 
\begin{align}
   \Delta:= 1 - \lVert T - T^{*} \rVert_{\text{op}},
\end{align} 
where $\lVert \cdot \rVert_{\text{op}}$ denotes the operator norm.
The gap quantifies the complexity of forming an approximate $t$-design as we can exponentially suppress the error with additional steps.

The integral
\begin{align}
\int_{[0, t_k]} (e^{-\theta_k \H_k})^{\otimes t}\otimes (e^{\theta_k \bar{\H}_k})^{\otimes t}d\theta_k
\end{align}
is an orthogonal projection onto commutant of $\{ (e^{-\theta_k\H_{k}})^{\otimes t}, \theta_k \in [0, t_k]\}$ or equivalently the kernel  of \begin{align}
    \sum_{r=1}^{t}\1^{\otimes(r-1)}\otimes (-\H_k  \otimes \1 + \1 \otimes \bar{\H}) \otimes \1^{\otimes(t-r)}.
\end{align}
We will denote the orthogonal projector onto the kernel of $\A$ by $\text{Ker}(\A)$. In addition due the assumption of LASA, $(-\H_k  \otimes \1 + \1 \otimes \bar{\H})$ only acts on tensors corresponding to vectorized forms of skew-Hermitian matrices, and we can thus ignore the trivial representation component. Under this assumption, we can equivalently consider the kernel of 
\begin{align}
\text{ad}_{\H_k}^{\otimes t} := \sum_{r=1}^{t}\1^{\otimes(r-1)}\otimes \text{ad}_{\H_{\text{k}}} \otimes \1^{\otimes(t-r)},
\end{align}
which corresponds to the $t$-th tensor power of the adjoint.

Since the representation $\text{ad}^{\otimes t}$ is completely reducible, and thus decomposes into irreps $\psi_m$,  the integral must respect this decomposition:
\begin{align}
\int_{[0, t_k]} (e^{-\theta_k \H_k})^{\otimes 2}\otimes (e^{-\theta_k \H_k^{*}})^{\otimes 2}d\theta_k &= \text{Ker}(\text{ad}_{\H_k}^{\otimes t})\\
& = \bigoplus_{m}\text{Ker}(\psi_m(\H_k)).
\end{align}

If we consider the max over all $m$ that don't correspond to the trivial rep, since $T^{*}$ and $T$ agree on the trivial components, we have
\begin{align}
\lVert T - T^{*} \rVert_{\text{op}} &= \lVert \frac{1}{d_{\g}}\sum_{k=1}^{d_{\g}} \bigoplus_{m}\text{Ker}(\psi_m(\H_k))\rVert_{\text{op}} \\ &\leq  \max_{m} \frac{1}{d_{\g}}\lVert\sum_{k=1}^{d_{\g}} \text{Ker}(\psi_{m} (\H_k)) \rVert_{\text{op}}.
\end{align} 
 Note that the first inequality follows because including more projectors only increases the operator norm. One can verify that
\begin{align}
\text{Ker}(\psi_{m}(\H_k)) \prec  \1 - \frac{\psi_m(\H_k)^2}{\lVert\psi_m(\H_k)\rVert_{\text{op}}^2},
\end{align}
which gives:

\begin{align}
 d_{\g}\1 - \sum_{k=1}^{d_{\g}} \text{Ker}(\psi_{m} (\H_k)) &\succ  \sum_{k=1}^{d_{\g}}\frac{\psi_m(\H_k)^2}{\lVert\psi_m(\H_k)\rVert_{\text{op}}^2} \\&= \sum_{k=1}^{d_{\g}}
\frac{\lVert \H_k\rVert_{\text{K}}^2}{\lVert\psi_m(\H_k)\rVert_{\text{op}}^2}\frac{\psi_m(\H_k)^2}{\lVert \H_k\rVert_{\text{K}}^2} \\
&\succ \left(\min_{k}\frac{\lVert \H_k \rVert_{\text{K}}^2}{\lVert\psi_m(i\H_k)\rVert_{\text{op}}^2}\right)\sum_{k=1}^{d_{\g}}
\frac{-\psi_m(\H_k)^2}{\lVert \H_k \rVert_{\text{K}}^2} \\
&= \left(\min_{k}\frac{\lVert \H_k\rVert_{\text{K}}^2}{\lVert\psi_m(i\H_k)\rVert_{\text{op}}^2}\right)C_{\psi_m}\\
&\succ \left(\min_{k}\frac{\lVert \H_k \rVert_{\text{K}}^2}{\lVert\psi_m(i\H_k)\rVert_{\text{op}}^2}\right)c_{\psi_m}\1\\
\end{align}
where $C_{\psi_m}$ is the quadratic Casimir in the representation $\psi_m$. Due to $\g$ being simple and $\phi$ irreducible, the $\psi_m(\H_{k})$ are orthogonal if the $\H_{k}$ are.
Since $\psi_m$ is irreducible Schur's lemma gives that $C_{\psi_m} = c_{\psi_m}\1$. %
Thus,
\begin{align}
\Delta &:=  1 - \lVert T - T^{*} \rVert_{\text{op}} \\ &\geq  1 - \max_{m} \frac{1}{d_{\g}}\lVert\sum_{k=1}^{d_{\g}} \text{Ker}(\psi_{m} (i\H_k)) \rVert_{\text{op}}\\ & = \min_{m} \frac{1}{d_{\g}}\lVert d_{\g}\1 -\sum_{k=1}^{d_{\g}} \text{Ker}(\psi_{m} (\H_k)) \rVert_{\text{op}}\\
&=\min_{m | m~\text{non-trivial}}\left(\min_k\frac{\lVert \H_k\rVert_{\text{K}}^2}{\lVert\psi_m(i\H_k)\rVert_{\text{op}}^2}\right)\frac{c_{\psi_m}}{d_{\g}}\\
&\geq\min_{m | m~\text{non-trivial}}\left(\min_k\frac{\lVert \H_k\rVert_{\text{K}}^2}{t^2\lVert\text{ad}_{i\H_k}\rVert_{\text{op}}^2}\right)\frac{c_{\psi_m}}{d_{\g}}.
\end{align}

Since $\g$ is compact simple, it must be isomorphic to either $\mathfrak{su}(m)$, $\mathfrak{so}(m)$, or $\mathfrak{sp}(m)$, and corresponding comlexified algebras are $\mathfrak{sl}(\mathbb{C})$, $\mathfrak{so}(\mathbb{C})$ or $\mathfrak{sp}(\mathbb{C})$. For analyzing the eigenvalue of the Casimir we need to look at the root system for the complexified algebra.

For a highest weight $\lambda$, the eigenvalue of the Casimir is \begin{align}
    \langle \lambda, \lambda \rangle + \langle \lambda, 2\delta \rangle,
\end{align}
where $\delta$ is the Weyl vector for the given root system. In addition, the inner product is w.r.t. the standard Euclidean inner product, and thus to account for our chosen normalization of the $\H_{k}$ w.r.t. the Killing form, we need to divide by $I_{\text{Ad}} \in\Theta(\sqrt{d_{\g}})$. So really we have:
\begin{align}
    c_{\psi_m} \in \Theta\left(\frac{\langle \lambda_{\psi_m}, \lambda_{\psi_m} \rangle + \langle \lambda_{\psi_m}, 2\delta \rangle}{\sqrt{d_{\g}}}\right), 
\end{align}
for irrep $\psi_m$.  Note that below $n \in \Theta(\sqrt{d_{\g}})$.

Following Ref.~\cite{hall2013lie}, the four possible root systems $A_n = \mathfrak{sl}(n+1)$, $B_n = \mathfrak{so}(2n, \mathbb{C})$, $C_n = \mathfrak{sp}(2n, \mathbb{C})$, and $D_n = \mathfrak{so}(2n+1, \mathbb{C})$  are:
\begin{align}
&A_n : e_j - e_{k}, j \neq k\\
&B_n : \pm e_j \pm e_k, j\neq k \cup \pm e_j\\
&C_n : \pm e_j \pm e_k, j\neq k \cup \pm 2e_j\\
&D_n :  \pm e_j \pm e_k, j\neq k.   
\end{align}

Each has as a basis:
\begin{align}
\label{eqn:base_a_n}
&A_n : e_j - e_{j+1}, j=1, \dots, n\\
\label{eqn:base_b_n}
&B_n : e_j - e_{j+1}, j=1, \dots, n-1 \cup e_n\\
\label{eqn:base_c_n}
&C_n : e_j - e_{j+1}, j=1, \dots, n-1 \cup 2e_n\\
\label{eqn:base_d_n}
&D_n : e_j - e_{j+1}, j=1, \dots, n-1 \cup e_{n-1}+e_n,
\end{align}
where the positive roots can be expressed as positive linear combinations of the above basis elements. Thus the positive roots for each system are:
\begin{align}
\label{eqn:a_n_roots}
&A_n : \{e_j - e_{k}~|~j,k \in \{1, \dots, n+1\}~\&~j < k \}\\
\label{eqn:b_n_roots}
&B_n : \{e_j \pm e_{k}, e_j~|~j,k \in \{1, \dots, n\}~\&~j < k\}\\
\label{eqn:c_n_roots}
&C_n : \{e_j \pm e_{k}, 2e_j~|~j,k \in \{1, \dots, n\}~\&~j < k\}\\
\label{eqn:d_n_roots}
&D_n : \{e_j \pm e_{k}~|~j,k \in \{1, \dots, n\}~\&~j < k\}.
\end{align}

The Weyl vector is one-half the sum of the positive roots leading to
\begin{align}
&2\delta_{j}^{(A)} = n-2(j-1)\\
&2\delta_{j}^{(B)} = 2n-2j+1\\
&2\delta_{j}^{(C)} = 2n-2j+2\\
&2\delta_{j}^{(D)} = 2n-2j.
\end{align}

A necessary condition  a highest weight $\lambda$ that is common across all four systems is that the inner product between $\lambda$ and each base root (coroot) must be $\geq 0$ and integer. In other words, it is a dominant integral element of the lattice. We now look at the possible highest weight vectors of the reducible representation $\text{ad}^{\otimes t}$ for each of the root systems to lower bound the eigenvalue of  the Casimir. 

For $A_n$ such vectors $\lambda$ in $\mathbb{R}^{n+1}$ in basis $\{e_j\}_{j=1}^{n+1}$ must satisfy \begin{align}
& \lambda_j \in \mathbb{Z}\\
&\lambda_{j} \geq \lambda_{j+1}\\
&\sum_{j} \lvert \lambda_{j} \vert = 2t,
\end{align}
and be a linear combination of $t$ roots.
Thus one can verify that the minimizer of the Casimir eigenvalue under such constraints is $\sum_{j=1}^{t}e_j - e_{n+2-j}$. This leads to 
\begin{align}
    c_{\psi_m} \in \Omega\left(\frac{2t(n - t + 2)}{n}\right).
\end{align}
The reproduces the result of Ref.\cite{haah2024efficient}.

For $B_n$ in the basis $\{e_j\}_{j=1}^{n}$ we have that a highest  weight in $\mathbb{R}^n$ must satisfy 
\begin{align}
& \lambda_j \in \mathbb{Z}\\
&\lambda_{j} \geq \lambda_{j+1} \geq 0\\
&\sum_{j} \lvert \lambda_{j} \vert \leq 2t, 
\end{align}
and be a linear combination of $t$ roots. This gives a minimizer of $\sum_{j=1}^{t} e_{j}$, leading to
\begin{align}
       c_{\psi_m} \in \Omega\left(\frac{t(2n - t +2)}{n}\right),
\end{align}
which also applies for $C_n$.

Lastly, for $D_n$ we have the following constraints for a highest weight in $\mathbb{R}^n$ we again have
\begin{align}
& \lambda_j \in \mathbb{Z}\\
&\lambda_{j} \geq \lambda_{j+1} \geq 0\\
&\sum_{j} \lvert \lambda_{j} \vert \leq 2t, 
\end{align}
and be a linear combination of $t$ roots. This leads to a minimizer of $\sum_{j=1}^{t}(e_j + e_{j+1})$. However, we can actually just use $\sum_{j=1}^{t}e_j$ as a lower bound leading to
\begin{align}
       c_{\psi_m} \in \Omega\left(\frac{t(2n -t)}{n}\right). 
\end{align}

The conclusion is that in all cases, we see that for $t \leq n$:
\begin{align}
    \frac{c_{\psi_m}}{t^2} \in \Omega(1/t),
\end{align}

so we can conclude that
\begin{align}
    \Delta \in \Omega\left(\min_k\frac{\lVert \H_k\rVert_{\text{K}}^2}{\lVert\text{ad}_{i\H_k}\rVert_{\text{op}}^2 d_{\g}t} \right) = \Omega\left(\min_k\frac{\lVert \text{ad}_{i\H_k}\rVert_{\text{F}}^2}{\lVert\text{ad}_{i\H_k}\rVert_{\text{op}}^2 d_{\g}t} \right) = \Omega\left(\frac{r_{\textup{K}}}{d_{\g}t}\right),
\end{align}
where $r_{\text{K}} := \min_k\frac{\lVert \text{ad}_{i\H_k}\rVert_{\text{F}}^2}{\lVert\text{ad}_{i\H_k}\rVert_{\text{op}}^2}$ is the minimum stable rank of the adjoint representation w.r.t. the chosen DLA basis.
\end{proof}

The result of Ref.~\cite{haah2024efficient} for $\mathfrak{su}(2^n)$ and the basis of Pauli operators follows from noting that under these conditions $r_{\text{K}} \in \Theta(d_{\g})$. Thus the mixing can be efficient for exponential DLAs. 

\begin{lemma}
\label{lem:min_stable_rank}
    For any compact simple Lie algebra, $r_{\text{K}} \in \Omega(\sqrt{d_{\g}})$.
\end{lemma}
\begin{proof}
The value of $r_{\text{K}}$ only depends on elements from the Cartan subalgebra. In addition, for compact simple Lie algebra, all elements are conjugate to an element from a Cartan subalgebra and all Cartan subalgebra are conjugate. 

We can express elements from  the Cartan subalgebra using a basis for the root systems (i.e. \eqref{eqn:a_n_roots}--\eqref{eqn:d_n_roots}). Specifically let $\{\alpha_k\}_{k=1}^{n}$ denote a basis for one of the four root systems $A_n$, $B_n$, $C_n$, or $D_n$. Let $\mathcal{R}$ denote the total set of roots, and for any two roots $\alpha, \beta$, $\lvert \langle \alpha, \beta \rangle \rvert \leq 2$ (see \eqref{eqn:a_n_roots}-- \eqref{eqn:d_n_roots}), where $\langle \cdot, \cdot \rangle$ is the standard Euclidean inner product on the root system. Note that 
\begin{align}
r_{\text{K}}^{(j)} = \frac{\lVert \text{ad}_{i\H_j} \rVert_{\text{F}}^2}{\lVert  \text{ad}_{i\H_j} \rVert_{\text{op}}} = \frac{\sum_{\alpha \in \mathcal{R}} \langle \alpha, i\H_{j} \rangle^2}{\max_{\alpha \in \mathcal{R}} \langle \alpha, i\H_{j}\rangle^2}.
\end{align}

We can express $\H_j$ in terms of the $\alpha_k$, $i\H_{j} = \sum_{k=1}^{n} c_k\alpha_k$ %
Note that we slightly abuse notation since $\H_{j}$ here is really the preimage of $\H_j$  under the  representation $\phi$. However, the Lie algebra being simple implies that the adjoint representations are isomorphic if $\phi$ is not trivial. %
We also have 
\begin{align}
\sum_{\alpha \in \mathcal{R}}\langle \alpha_k, \alpha \rangle^2 \in \Omega(n) \in \Theta(\sqrt{d_{\g}}),
\end{align}
as one can check using \eqref{eqn:a_n_roots}-\eqref{eqn:d_n_roots} and \eqref{eqn:base_a_n}--\eqref{eqn:base_d_n}. Thus, we get
\begin{align}
r_{\text{K}}^{(j)} = \frac{\sum_{k=1}^{n}c_k^2\sum_{\alpha \in \mathcal{R}}\langle \alpha_{k}, \alpha \rangle^2}{\max_{\alpha \in \mathcal{R}}\sum_{k=1}^{n}c_k^2\langle \alpha_{k}, \alpha\rangle^2} \geq  \frac{\sum_{k=1}^{n}c_k^2\sum_{\alpha \in \mathcal{R}}\langle \alpha_k, \alpha \rangle^2}{\sum_{k=1}^{n}4c_k^2} \in \Omega(\sqrt{d_{\g}}),
\end{align}
so $r_{\text{K}} \in \Omega(\sqrt{d_{\g}})$.
\end{proof}

The orthogonality assumption plays a key role in lower bounding the gap using the quadratic Casimir, and seems like an intuitive requirement for faster mixing. However, it is unclear if the full basis assumption can be relaxed and replaced with a dense subset. Numerical evidence appears to showing fast mixing when this condition is relaxed, for example QAOA-like ansatz \cite{Larocca2022diagnosingbarren}.
Still, for sufficiently large $m$, the Baker-Campbell-Hausdorff formula gives that for any $\H_{j}, \H_{k}$:
\begin{align}
\left(e^{-\sqrt{t/m}\H_{j}}e^{-\sqrt{t/m}\H_{k}}e^{\sqrt{t/m} \H_{j}}e^{\sqrt{t/m}\H_{k}}\right)^{m} = e^{-t[ \H_j, \H_k] + \mathcal{O}(t^{3/2}/m^{1/2})},
\end{align}
which can approximate nested commutator for large enough $m$. This is by no means showing that the full basis assumption can be relaxed for faster mixing. However, it shows that a subset of the generators can be used to approximately sample from their nested commutators, and the density of a periodic ansatz within its dynamical Lie group. The mixing of dense subgroups (i.e. nonzero spectral gap) seems to require further conditions \cite{bourgain2011spectral, benoist2016spectral}.

We now present two simple corollaries of Theorem \ref{thm:t_mixing_time} that appeared as theorems in the main text.

\begin{corollary}[Main Text Theorem \ref{thm:rapid}]
\label{cor:rapid}
Consider an orthogonal basis of skew-Hermitian generators $\mathcal{A}:=\{\H_1, \dots, \H_{d_{\g}}\}$ for the DLA with the property that the unitary $e^{-\theta\H_k}$ corresponding to a generator $\H_{k}$ is $t_k$-periodic. In addition, suppose that $d_{\g} = \mathcal{O}(\textup{poly}(n))$. Consider a LASA formed by applying evolutions $e^{-\theta_k \H_k}$ where $\H_{k}$ is selected uniformly at random from the set $\mathcal{A}$ and the parameter $\theta_{k}$ uniformly from $[0, t_k)$. Then, the ansatz is an $\epsilon$-approximate $2$-design for the dynamical group $\mathcal{G}$  after  $\mathcal{O}(\textup{poly}(n)\log(1/\epsilon))$ layers.
\end{corollary}
\begin{proof}
One can see that the theorem follows from the lower bound on $r_{\text{K}}$.
\end{proof}

\begin{corollary}[Main Text Theorem \ref{thm:mix_compound}]
\label{cor:rapid_compound}
Consider an $n$-qubit quantum compound ansatz that is a LASA constructed using the set of generators $\{X^{(ij)}$,$Y^{(ij)}$, $\sum_{i=1}^{j} Z^{(ij)}\}$ with rotations angles chosen uniformly at random. Then,  for $t \leq n/2$, the ansatz is an $\epsilon$-approximate $t$-design for the dynamical group $\textup{SU}(n)$  after  $\mathcal{O}(tn\log(1/\epsilon))$ layers.
\end{corollary}
\begin{proof}
One can check that the chosen set of generators form an orthogonal basis for the DLA. Each of the $X^{(ij)}$ and $Y^{(ij)}$ Givens rotation generators is conjugate to an $Z^{(ij)}$ generator, which is an element of the Cartan subalgebra. Thus from Lemma \ref{lem:min_stable_rank} the result follows.
\end{proof}

\section*{Disclaimer}
This paper was prepared for informational purposes by the Global Technology Applied Research center of JPMorgan Chase \& Co. This paper is not a product of the Research Department of JPMorgan Chase \& Co. or its affiliates. Neither JPMorgan Chase \& Co. nor any of its affiliates makes any explicit or implied representation or warranty and none of them accept any liability in connection with this paper, including, without limitation, with respect to the completeness, accuracy, or reliability of the information contained herein and the potential legal, compliance, tax, or accounting effects thereof. This document is not intended as investment research or investment advice, or as a recommendation, offer, or solicitation for the purchase or sale of any security, financial instrument, financial product or service, or to be used in any way for evaluating the merits of participating in any transaction.

\end{document}